\pdfoutput=1
\documentclass[12pt]{article}

\usepackage[utopia]{mathdesign}

\usepackage[T1]{fontenc}
\usepackage[latin9]{inputenc}
\usepackage{float} 
\usepackage{tabularx} 
\usepackage{subfig} 
\usepackage{rotating} 
\usepackage{adjustbox}
\usepackage{setspace}
\usepackage{natbib}

\usepackage{amsmath,amssymb,amsthm}
\usepackage[colorlinks,citecolor=blue,linkcolor=blue,urlcolor=blue,pagebackref]{hyperref}
\usepackage[margin=1in]{geometry}
\usepackage[labelfont=bf,textfont=bf,labelsep=period]{caption}

\theoremstyle{plain}

\newtheorem{assumption}{Assumption}
\newtheorem{theorem}{Theorem}
\newtheorem{algorithm}{Algorithm}

\newtheorem{lemma}{Lemma}

\theoremstyle{remark}

\newtheorem{remark}{Remark}

\renewcommand{\S}{\mathrm{S}}

\newcommand{\C}{\mathrm{C}}
\newcommand{\E}{\mathrm{E}}
\newcommand{\R}{\mathbb{R}}

\newcommand{\Y}{\mathrm{Y}}

\newcommand{\qedd}{\hfill{\tiny \ensuremath{\blacksquare} }}
\newcommand{\wpas}{\ensuremath{\text{w.p.a.1}\,\,}}
\newcommand{\wpa}{\ensuremath{\text{w.p.a.1}}}
\newcommand{\supp}{\operatorname{supp}}
\newcommand{\argmax}{\operatorname*{arg\,max}}

\begin{document}
\doublespacing

\title{Higher-Order Asset
Pricing Factors via Forward Selection Fama-MacBeth Regression\thanks{We are grateful for the comments by Gavin Feng, Stefano Giglio, Narayana Kocherlakota, Whitney Newey, Robert Novy-Marx, Allan Timmermann, and Dacheng Xiu.}}
\author{Nicola Borri\thanks{Department of Economics and Finance, LUISS University.} \and
Denis Chetverikov\thanks{Department of Economics, UCLA.} \and
Yukun Liu\thanks{Simon Business School, University of Rochester.} \and
Aleh Tsyvinski\thanks{Department of Economics, Yale University.}}
\date{March 4, 2026}
\maketitle

\bigskip

\begin{abstract}
We show that the higher-order terms and interactions of the common sparse linear factors are significantly priced in the cross-section of equity returns. A higher-order model with only a small number of selected higher-order terms from six widely used factors outperforms traditional benchmarks both in-sample and out-of-sample. It also substantially reduces the alphas of the extensive factor zoo, suggesting that the pricing power of many zoo factors is attributable to their exposure to higher-order terms of common linear factors. We identify and rank the most relevant higher-order terms by developing a forward selection Fama-MacBeth procedure.
\end{abstract}


\bigskip

\section{Introduction}

The asset pricing literature has proposed hundreds of candidate factors
to explain the cross-section of asset returns. \citet{cochrane2011presidential},
in his presidential address, refers to this proliferation of candidate
factors as the ``factor zoo.'' This proliferation is itself symptomatic of the
fact that sparse linear factor models, such as those introduced by \citet{fama2015five}
and \citet{Carhart1997}, struggle to fully price the cross-section and hence may not fully span the stochastic discount factor (SDF). In this paper, we show that the interactions of the higher-order terms of the common sparse linear factors are significantly priced in the cross-section of equity returns and capture a significant fraction of the factor zoo. We do so by determining the important interactions via a new factor-selection procedure, which we call the \emph{forward selection Fama-MacBeth regression procedure} (FS-FMB).

The prevalent focus in the literature on sparse linear factor models is primarily driven by considerations of parsimony and is motivated by the first-order Taylor expansion of the SDF. This observation naturally raises the question of whether higher-order terms and interactions of common sparse linear factors, motivated by a higher-order Taylor expansion, also play an important role in pricing the cross-section of expected returns. In this paper, we show that a small and interpretable subset does, by selecting and ranking the higher-order terms. 

A main challenge in studying such relevant terms is dimensionality. Consider, for example, the standard 6-factor model that includes the Fama-French 5 factors (\citeauthor{fama2015five}, \citeyear{fama2015five}) and the momentum factor (\citeauthor{Carhart1997}, \citeyear{Carhart1997}). The higher-order interactions up to degree 3 (4) already have 63 (114) factors. Therefore, we develop a forward selection procedure to choose factors that explain
the stochastic discount factor that is suitable for the case when the number of potentially important factors is large. The method adds factors one by one into the set of selected factors, each time maximizing
the $R^{2}$ of the cross-sectional regression (second step) of the
Fama-MacBeth procedure, thus naturally ranking the importance of the selected factors. After the set of factors is selected, we estimate the SDF loadings by running the cross-sectional regression of sample average returns on sample covariances between returns and selected factors. We derive the rate of convergence of the proposed estimator and also propose a method to debias it, leading to simple inference procedures. 

We show that the performance of the common factors is
substantially improved with their higher-order interactions. We start with
6 of the most common factors in the asset pricing literature. These
are the Fama-French 5 factors of \citet{fama2015five}
and the momentum factor of \citet{jegadeesh1993returns}, and we refer
to this model as the FF5M. We consider higher-orders as well as interactions
of these 6 factors up to degree 3, leading to 57 candidate higher-order
factors. For the test assets, we use a large cross-section of equities
including 484 characteristic-managed portfolios from \citet{kozak2020shrinking}.
We apply our FS-FMB to select the higher-order factors that are most
important in pricing the cross-section of these 484 asset returns.
The procedure selects 7 higher-order factors, including 2 second degree
higher-orders and 5 interactions.\footnote{The FS-FMB procedure requires the choice of a stopping threshold for the second-stage $R^2$ improvement. In our baseline, we set it to be 1\%. In the Appendix, we show that we can use a cross-validation method to choose the optimal stopping threshold, which generates qualitatively and quantitatively similar results.} The inclusion of the 5 interaction terms shows the importance of co-dependence of common asset pricing factors in pricing the cross-section of equity returns. We show that the inclusion of the higher-order factors leads to substantial improvement in pricing the
cross-section of asset returns -- the adjusted cross-sectional $R^{2}$
increases from 31.2\% for the baseline FF5M model to 59\% for the
selected higher-order factor model. Moreover, the intercept alpha
of the higher-order factor model is not statistically significant at the 5\% level. 

We then examine the SDF loadings of the higher-order factors selected
by the FS-FMB using both the non-debiased estimates based on the standard Fama-MacBeth (FMB) two-pass regression and the debiased estimates. Estimates from
the standard FMB two-pass procedure are likely biased because of the
omitted variable bias (see, e.g., \citet{chernozhukov2018double}
and \citet{giglio2021asset}). The directions of the point estimates
are the same and the magnitudes are similar across the two procedures. The
estimates based on the non-debiased procedure are all significant at
the 5\% level, while the estimates based on the debiased procedure are
all significant at the 10\% level with 5 of the 7 SDF loading estimates
being significant at the 5\% level. Overall, the results show that
the higher-order factors are important for empirically capturing the
SDF and significantly price assets in the cross-section.

In our baseline analysis, we demonstrate that higher-order factors
are important components of the SDF and price cross-sectional asset
returns in-sample. We further provide evidence of the out-of-sample
performance of the selected higher-order factor model, where we conduct
two types of out-of-sample tests. Firstly, we present an out-of-sample
test in the asset space. We select higher-order factors and estimate
SDF loadings from the training assets and then evaluate the out-of-sample
performance using the remaining assets. Our results show that the
selected higher-order factors and the performance of the higher-order
factor model are stable in this out-of-sample exercise.

Secondly, we conduct an out-of-sample test in the time series, where
we split the sample into two equal subsamples based on the time series.
We use the first half of the sample as the training sample and the
second half of the sample as the test sample. Then, we estimate the
model in the training sample and evaluate the out-of-sample performance
in the test sample. Compared with the benchmark factor models including
the CAPM, the FF3 factor model, the FF5 factor model, and the FF5M
factor model, the out-of-sample $R^{2}$ of the higher-order factor
model is significantly larger.

Next, we examine whether the selected higher-order factor model can
reduce the pricing power of the factors in the factor zoo. We employ
two empirical procedures. First, we perform standard FMB regressions
to test the significance of the SDF loading of each factor in the
factor zoo, controlling for the higher-order factor model. Second,
we construct factor mimicking portfolios for the seven higher-order
factors and conduct time-series asset pricing tests on the zoo factors.

Under the first approach, we estimate each zoo factor's SDF loading
while controlling for our higher-order factor model. We find that
only 7 out of the 148 zoo factors from \citet{jensen2023there} remain
significant at the 5\% level. The second
approach is based on factor-mimicking portfolios. The benefit is
that we are able to convert the non-tradable higher-order factors
to tradable factors but the projection inevitably loses information
from the original factors. The results using the second approach confirm
that the higher-order factor model significantly reduces the pricing
power of the zoo factors. Overall, these results demonstrate that
the higher-order factors selected by the forward selection Fama-MacBeth procedure substantially reduce the zoo factors, although the linear benchmark
factors cannot.

Because our FS-FMB procedure directly selects and ranks the important higher-order factors while preserving their interpretability, one of its key advantages is that it allows us to connect these factors to underlying economic mechanisms. To this end, we first investigate the economic content of the selected higher-order factors by examining their exposures to a broad cross-section of macroeconomic variables. While unconditional regressions reveal some baseline sensitivity, especially to financial uncertainty, we find that the economic content of these factors is concentrated in extreme economic states. Using quantile regressions, we show that both the explanatory power and the sensitivity to specific macroeconomic variables are markedly amplified in the tails of the distribution. For example, the pseudo-$R^2$ for the market factor squared reaches 70\% at the 90th percentile, and variables related to intermediary capital constraints, which are often insignificant in unconditional regressions, become strongly significant in the tails. These results suggest that the higher-order factors capture non-linear exposure to macroeconomic tail risks, particularly those associated with financial intermediary distress. This provides a potential economic rationale, related to intermediary asset pricing (e.g., \citealp{he2013intermediary,brunnermeier2014macroeconomic}), for why higher-order terms of common factors are priced and why they account for a substantial fraction of the broader factor zoo.

Second, we study how these higher-order factors vary with broader macroeconomic and business-cycle conditions by examining their correlations with two key indicators, U.S. equity market returns and the CBOE VIX index, both in the full sample and during NBER recession periods. For the first two higher-order factors (i.e., SMB2 and SMB2{*}Mom), the magnitude of the correlations with U.S. equity market returns is substantially larger during recessions than in the full sample. For the last four higher-order factors (i.e., Mkt-RF2, Mkt-RF2{*}RMW, Mkt-RF{*}SMB, and HML2{*}Mkt-RF), the correlations with the VIX index are large in magnitude, especially during recessions. Overall, these results show that the selected higher-order factors are closely linked to macroeconomic risk over the business cycle, with particularly strong exposure during downturns.

Additionally, we provide robustness results for our empirical analyses. First, we show that the performance of the higher-order factor model is extremely unlikely to be driven by randomness by simulating white-noise factors. Second, we use the selected higher-order factor model to account for an alternative factor zoo of \citet{ChenZimmermann2021} and reach similar conclusions. Third, we experiment with alternative sets of candidate higher-order factors: (1)  include only higher-orders, (2) only interactions, (3) up to degree 2, and (4) up to degree 4. We find that both higher-orders and interactions are important components of the SDF, especially interactions. Moreover, we show that the selected higher-order factor model up to degree 2 significantly underperforms, while the selected higher-order factor model up to degree 4 has similar
performance to our baseline model up to degree 3. 

The paper contributes to the factor zoo literature. \citet{cochrane2011presidential} identifies the factor zoo phenomenon and calls for research to identify factors that are actually important for the SDF. \citet{harvey2016and} argue that, based on a multiple-hypothesis-testing framework, many zoo factors are not significant. \citet{harvey2021lucky} provide a bootstrap model selection framework that accounts for multiple hypothesis testing. Relatedly, \citet{bryzgalova2023bayesian} and \citet{bryzgalova2025forest} develop a Bayesian framework and a decision-tree framework for evaluating high-dimensional factor models, respectively. \citet{mclean2016does} show that many zoo factors suffer from post-publication bias. \citet{green2017characteristics} use Fama-MacBeth regressions to evaluate 94 firm characteristics and argue that fewer than ten provide significant independent information about expected returns. Our contribution to this literature is to show that a substantial fraction of the factor zoo can be explained by higher-order terms and interactions of common sparse linear factors.

Our paper is methodologically related to \citet{feng2020taming}, but it addresses a different economic question. \citet{feng2020taming} study a factor-screening problem: whether a proposed factor contributes incremental pricing information conditional on a \emph{large} set of existing factors. We study, instead, whether much of the pricing information in the factor zoo can be accounted for by a \emph{small} higher-order factor model built from a few common linear factors. Put differently, we ask whether the proliferation of factors largely reflects higher-order manifestations of a few common factors. Our forward selection Fama-MacBeth method is designed for this economic objective. Unlike the Lasso-based method in \citet{feng2020taming}, which solves a penalized least-squares problem and induces sparsity through coefficient shrinkage, our procedure sequentially selects factors that maximize incremental cross-sectional pricing fit, thereby providing a natural ordered ranking of the importance of higher-order terms. This distinction is important both economically and algorithmically. First, forward selection does not shrink the risk-price estimates of selected factors, thereby preserving their economic interpretation. Second, the sequential nature of forward selection allows us to rank higher-order terms by their incremental contribution to pricing power. This ranking itself has economic content: for example, second-order size effects enter before cubic terms, and interactions dominate pure higher powers. Lasso does not provide such an ordered decomposition because variables can both enter and exit along the regularization path. Finally, our main theoretical contribution is to establish high-dimensional asymptotic theory for the forward selection Fama-MacBeth procedure. In particular, we derive the {\em rate of convergence} of the resulting estimator of SDF loadings in a setting where the number of candidate factors can be large relative to the time dimension. To the best of our knowledge, no such results are available for forward selection in the multi-step Fama-MacBeth framework. This theoretical contribution is conceptually distinct from \citet{feng2020taming}, whose focus is on valid {\em inference} for testing the contribution of new factors using an $\ell_1$-regularized (Lasso-based) estimator.

Our paper also relates to a vast literature on factor models in asset pricing. A major part of the literature focuses on linear factor models, such as the characteristic-based models (e.g., \citeauthor{fama1993common} \citeyearpar{fama1993common,fama1996multifactor,fama2015five}; \citeauthor{Carhart1997}, \citeyear{Carhart1997};
\citeauthor{koijen2017cross}, \citeyear{koijen2017cross}) and statistical-based models (e.g.,
\citeauthor{connor1986performance}, \citeyear{connor1986performance}; \citealp{sandulescu2021model}; \citealp{korsaye2023global}). A small subset of papers emphasizes
non-linearity in explaining asset returns, such as \citet{bansal1993no},
\citet{harvey2000conditional}, and \citet{dittmar2002nonlinear}.
We contribute to this literature by showing that higher-orders and higher-order interactions of common linear factors are important components of the SDF and price the cross-sectional asset returns.

Finally, our paper relates to the machine learning literature. \citet{gu2020empirical}
show that machine learning methods are better in predicting future
stock returns out-of-sample than OLS methods. \citet{freyberger2020dissecting}
propose a group LASSO method to choose firm characteristics that are
most related to future returns. \citet{kozak2020shrinking} show that
machine learning methods applied to the factor zoo can better approximate
the SDF. \citet{lettau2020factors} propose a method for estimating
latent asset pricing factors that fit both the time-series and the
cross-section of asset returns. \citet{giglio2021asset} propose a
method to estimate risk premia of factors even if there are missing
factors. \citet{Borri2024factor} show that a new non-linear single-factor
asset pricing model motivated by the Kolmogorov-Arnold representation
theorem is able to price assets from different asset classes. In this
paper, we propose a forward selection method to estimate a high-dimensional
stochastic discount factor model, isolating the most relevant
higher-order factors and showing their importance.

The rest of the paper is organized as follows. Section \ref{sec:Methodology}
discusses the forward selection Fama-MacBeth procedure. Section
\ref{sec:Empirical-Analysis} presents the data and the main empirical results.
Section \ref{sec:Robustness} presents additional robustness
analyses. The Online Appendix contains technical derivations and further empirical results.

\section{Methodology}\label{sec:Methodology}

In this section, we set up the model. We introduce its main components -- returns, factors, and the stochastic discount factor -- and explain the relationship between them. We then present a forward selection Fama-MacBeth regression procedure as a method for selecting the factors that drive the stochastic discount factor. Using this selection procedure, we propose an estimator of the stochastic discount factor loadings. As we will see, the proposed estimator is biased. We therefore also explain how to debias it and how to perform inference on the stochastic discount factor loadings using the debiased estimator.

\textbf{Notation.} For any integer $p\geq 1$, we denote $[p]=\{1,\dots,p\}$ and $\mathbf 0_p = (0,\dots,0)^{\top}\in\R^p$. For any set $\S\subset[p]$, we define $\S^c = [p]\setminus\S$, and we denote by $|\S|$ the number of elements in $\S$. For any vector $x = (x_1,\dots,x_p)^{\top}$ and any set $\S\subset[p]$, we use $x_{\S} = (x_i)^{\top}_{i\in\S}$ to denote the sub-vector of $x$ consisting of all its components with indices in $\S$. Similarly, for any vector $x_t = (x_{t,1},\dots,x_{t,p})^{\top}$ with an extra index $t$ (or some other index), we denote by $x_{t,\S} = (x_{t,i})_{i\in\S}^{\top}$ the corresponding sub-vector. For any matrix $A = (A_{i,j})_{i\in[N],j\in[p]}$ and any set $\S\subset[p]$, we use $A_{\S} = (A_{i,j})_{i\in[N],j\in\S}$ to denote the sub-matrix of $A$ consisting of all its columns with indices in $\S$. For any vector $x = (x_1,\dots,x_p)^{\top}\in\R^p$, we let $\|x\|_0 = \sum_{j\in[p]} \mathbf I\{x_j\neq 0\}$ be its $\ell_0$-``norm'' and let $\|x\|_1 = \sum_{j\in[p]} |x_j|$, $\|x\|_2 = (\sum_{j\in[p]} |x_j|^2)^{1/2}$, and $\|x\|_{\infty} = \max_{j\in[p]}|x_j|$ be its $\ell_1$-, $\ell_2$-, and $\ell_{\infty}$-norms. For any matrix $A = (A_{i,j})_{i\in[N],j\in[p]}\in\R^{N\times p}$, we let $\|A\|_2 = \sup_{v\in\R^p\colon \|v\|_2=1} \|Av\|_2$ and $\|A\|_{\infty,1} = \max_{i\in[N]}\sum_{j\in[p]} |A_{i,j}|$ be its spectral and $\ell_{\infty,1}-$norms. For any symmetric matrix $A\in\R^{p\times p}$, we let $\lambda_{\min}(A)$ and $\lambda_{\max}(A)$ be its smallest and largest eigenvalues. We use the abbreviation ``\wpa'' to denote ``with probability approaching one.'' For non-random sequences $\{a_{T}\}_{T\geq1}$ and $\{b_{T}\}_{T\geq 1}$ in $\R$, we write $a_{T}\lesssim b_{T}$ if there exists a bounded non-random sequence $\{C_{T}\}_{T\geq 1}$ such that $a_{T}= C_{T} b_{T}$ for all $T\geq 1$. For random sequences $\{a_{T}\}_{T\geq 1}$ and $\{b_{T}\}_{T\geq 1}$ in $\R$, we write $a_{T} \lesssim_P b_{T}$ if there exists a bounded in probability random sequence $\{C_{T}\}_{T\geq1}$ such that $a_{T} = C_{T} b_{T}$ for all $T\geq 1$.


\subsection{Model}

Let $r_{i,t}\in\R$ be the excess return of asset $i\in[N]$ at time period $t\in[T]$ and let $f_t = (f_{t,1},\dots,f_{t,p})^{\top}\in\R^p$ be the vector of factors. Throughout the paper, we assume that the returns and the factors are stationary over time. In addition, we assume that the vector of factors is able to span the stochastic discount factor (SDF):
\begin{equation}\label{eq: sdf}
m_t = 1-\psi^\top (f_t - \E[f_t]),\quad\text{for all }t\in[T],
\end{equation}
where $m_t$ is the SDF and $\psi = (\psi_1,\dots,\psi_p)^{\top}\in\R^p$ is the vector of SDF loadings. 

By definition, since we work with the {\em excess} returns, $m_t$ is orthogonal to the returns in the sense that $\E[m_t r_{i, t}] = 0$ for all $i\in[N]$. Using this equation in combination with \eqref{eq: sdf}, it then follows from Section 6.3 in \citet{C05} that for each asset $i\in[N]$,
\begin{equation}\label{eq: model}
r_{i,t} = \beta_{i}^{\top}(\gamma - \E[f_t]) + \beta_{i}^{\top}f_t + \varepsilon_{i,t},\quad \text{for all }t\in[T], 
\end{equation}
where $\beta_{i} = (\beta_{i,1},\dots,\beta_{i,p})^\top\in\R^p$ is the vector of factor exposures corresponding to the factors $f_t$, $\gamma = (\gamma_1,\dots,\gamma_p)^{\top}\in\R^p$ is the vector of corresponding factor risk premia, and $\varepsilon_{i,t}\in\R$ is the residual satisfying $\E[\varepsilon_{i,t}] = 0$ and $\E[\varepsilon_{i,t}f_t] = \mathbf 0_p$. The vector of SDF loadings $\psi$ is related to the vector of factor risk premia $\gamma$ via the following equation:
\begin{equation}\label{eq: psi-gamma relation}
\psi = (\E[(f_t - \E[f_t])(f_t - \E[f_t])^{\top}])^{-1}\gamma,
\end{equation}
where we assumed that the matrix $\E[(f_t - \E[f_t])(f_t - \E[f_t])^{\top}]$ is invertible. 
In addition, by using standard least-squares projection formulas together with \eqref{eq: model} and \eqref{eq: psi-gamma relation}, one can show that
\begin{equation}\label{eq: mean and covariances}
\E[r_{i,t}] = \sum_{j\in[p]} C_{i,j}\psi_j,
\end{equation}
where $\C_{i,j} = \E[(f_{t,j}-\E[f_{t,j}])(r_{i,t} - \E[r_{i,t}])]$ is the covariance between $r_{i,t}$ and $f_{t,j}$.

In this section, our task is to propose a new estimator of the vector of SDF loadings $\psi$ that can be used in the case when the number of factors $p$ is large, potentially much larger than the number of assets $N$ and/or the number of time periods $T$. As we will see, the classic estimation techniques may fail in this case. In fact, one can show that the SDF loadings cannot be consistently estimated without further assumptions in this case. Therefore, to make progress, and following the literature on estimation of high-dimensional models, we assume that the model is sparse, in the sense that most of the variation in the SDF in equation \eqref{eq: sdf} can be explained by a relatively small number of factors. In other words, only a relatively small number of coefficients in the vector $\psi$ are sufficiently different from zero. 

\subsection{Estimation}\label{sub: estimation}

To develop intuition for the proposed method, let us first recall the classic Fama-MacBeth procedure that can be used to estimate the vector of risk premia $\gamma$ in the low-dimensional case, i.e. when $p$ is small. The procedure consists of two steps. First, using equation \eqref{eq: model}, for each asset $i\in[N]$, we run a time series OLS regression of $r_{i,t}$ on $f_t$, with intercept included, to obtain estimates of factor exposures $\widehat\beta_i$. Second, we run a cross-sectional OLS regression of average return $\bar r_i = T^{-1}\sum_{t=1}^T r_{i,t}$ on $\widehat\beta_i$ (without intercept). The vector of slope estimates in this regression is then an estimator of the vector of risk premia $\gamma$. In turn, once the vector of risk premia is estimated, the vector of SDF loadings $\psi$ can be estimated using an empirical version of equation \eqref{eq: psi-gamma relation}.  

The procedure described above is consistent when the number of factors $p$ is small. When the number of factors $p$ is large, however, all three steps may yield large estimation errors or fail as they require inverting estimates of poorly-conditioned/singular matrices. To deal with this problem, we propose a forward selection algorithm that starts with no factors and then adds factors iteratively one by one with the goal of achieving the highest possible $R^2$ in the cross-sectional regression of the Fama-MacBeth procedure. We stop the algorithm once a desired number of factors is achieved or the growth in the $R^2$ becomes sufficiently small. Once a set of factors is selected, we perform the Fama-MacBeth procedure and use an empirical version of equation \eqref{eq: psi-gamma relation} to estimate $\psi$ as described above on selected factors. We set an estimate of the SDF loadings for non-selected factors to be zero.

To define our method more precisely, for each subset $\S$ of $[p] = \{1,\dots,p\}$, consider the Fama-MacBeth procedure using only factors in the set $\S$. In particular, define
\begin{equation}\label{eq: multivariate beta}
\widehat\beta_{i,\S} = \left(\sum_{t\in[T]} (f_{t,\S} - \bar f_{\S})(f_{t,\S} - \bar f_{\S})^{\top}\right)^{-1} \left(\sum_{t\in[T]} (f_{t,\S} - \bar f_{\S})r_{i,t}\right),
\end{equation}
where $\bar f = (\bar f_1,\dots,\bar f_p)^{\top} = T^{-1}\sum_{t=1}^T f_{t}$ and the inverse of the matrix is understood as the Moore-Penrose generalized inverse if the matrix is singular.\footnote{Note that the vector $\widehat\beta_{i,\S}$ is not a sub-vector of some $p$-dimensional vector $\widehat\beta_i$. In fact, even if we have $j\in\S_1\cap\S_2$ for some $j$, it is clearly possible that components of the vectors $\widehat\beta_{i,\S_1}$ and $\widehat\beta_{i,\S_2}$ corresponding to the factor $j$ may differ.} Also, let $R_{FM}^2(\S)$ be the $R^2$ in the cross-sectional OLS regression of $\bar r_i$ on $\widehat\beta_{i,\S}$. We apply the forward selection algorithm with the aim to maximize $R_{FM}^2(\S)$. For a desired number of factors $\widehat s$, the algorithm takes the following form:

\begin{algorithm}[Forward Selection Fama-MacBeth Procedure]\label{alg: fs fama-macbeth}  \mbox{} \\
\begin{enumerate}
\item Set $\S = \emptyset$.
\item Set $\widehat j = \arg\max_{j\in[p]\setminus\S} R_{FM}^2(\S\cup j)$.
\item Set $\S = \S\cup\widehat j$.
\item If $|\S| = \widehat s$, then stop. Otherwise, proceed to Step 2.
\end{enumerate}
\end{algorithm}

Alternatively, the algorithm can be stopped when adding an extra factor does not significantly increase the $R^2$, i.e. $R_{FM}^2(\S\cup \widehat j) - R_{FM}^2(\S) \leq \epsilon$ for some small value of the tolerance parameter $\epsilon$ and $\widehat j$ appearing in Step 2 of the algorithm. In practice, we use $\epsilon = 1\%$, which seems meaningful from the empirical point of view.\footnote{Additionally, we implement a cross-validation method to select the optimal $\epsilon$ in Appendix \ref{sec:optimal_stopping}.} Also, some factors that are deemed to be important, e.g. Fama-French factors, can be included in the algorithm in the first step, i.e. we can set $\S = \S_0$ for some pre-specified set of factors $\S_0$ in Step 1 of the algorithm instead of setting $\S = \emptyset$.

The result of Algorithm \ref{alg: fs fama-macbeth} is the set of factors $\S$, which we denote as $\widehat{\S}$ throughout the paper. Using this set, we estimate the vector of SDF loadings by $\widehat\psi = (\widehat\psi_1,\dots,\widehat\psi_p)^{\top}$, where we set
\begin{align}\label{eq: psi hat}
\widehat\psi_{\widehat{\S}} & = \left(\frac{1}{T}\sum_{t\in[T]} (f_{t,\widehat{\S}} - \bar f_{\widehat{\S}})(f_{t,\widehat{\S}} - \bar f_{\widehat{\S}})^{\top}\right)^{-1}\left(\sum_{i\in[N]} \widehat\beta_{i,\widehat{\S}}\widehat\beta_{i,\widehat{\S}}^{\top}\right)^{-1}\left( \sum_{i\in[N]} \widehat\beta_{i,\widehat{\S}}\bar r_i \right)
\end{align}
and $\widehat\psi_j = 0$ for all $j\in\widehat{\S}^c$. We will prove consistency and derive the rate of convergence of this estimator in Section \ref{sub: asymptotic theory} below.

\begin{remark}[Equivalent Form of Estimator $\widehat\psi$]\label{rem: eq form of estimator}
For each $i\in[N]$ and $j\in[p]$, let $\widehat \C_{i,j} = T^{-1}\sum_{t\in[T]} (f_{t,j} - \bar f_j)r_{i,t}$ be the time series sample covariance between the factor $f_{t,j}$ and the return $r_{i,t}$. Also, let $\widehat \C = (\widehat \C_{i,j})_{i\in[N],j\in[p]}$ be the matrix of all sample covariances. Under conditions to be imposed in Section \ref{sub: asymptotic theory}, with probability approaching one, the matrices
$
T^{-1}\sum_{t\in[T]} (f_{t,\S} - \bar f_{\S})(f_{t,\S} - \bar f_{\S})^{\top}
$
will be non-singular for all $\S$ used in Algorithm \ref{alg: fs fama-macbeth}. However, whenever these matrices are non-singular, substituting \eqref{eq: multivariate beta} into \eqref{eq: psi hat}, it follows that $\widehat\psi_{\widehat{\S}}$ can be equivalently rewritten as
\begin{equation}\label{eq: equivalent characterization}
\widehat\psi_{\widehat{\S}} = \left(\widehat \C_{\widehat{\S}}^{\top}\widehat \C_{\widehat{\S}}^{\top}\right)^{-1} \widehat \C_{\widehat{\S}}^{\top} \bar r,
\end{equation}
where $\bar r = (\bar r_1,\dots,\bar r_N)^{\top}$ is the vector of sample average returns, which means that $\widehat\psi_{\widehat{\S}}$ can be computed by the cross-sectional OLS regression of $\bar r_i$ on $\widehat \C_{i,\widehat{\S}}$.\qedd
\end{remark}

\begin{remark}[Equivalent Forms of Forward Selection Fama-MacBeth Procedure]\label{rem: eq form of selection}
As in the previous remark, under our conditions, the matrices
$
T^{-1}\sum_{t\in[T]} (f_{t,\S} - \bar f_{\S})(f_{t,\S} - \bar f_{\S})^{\top}
$
will be non-singular for all $\S$ used in Algorithm \ref{alg: fs fama-macbeth} with probability approaching one. However, whenever these matrices are non-singular, it follows from equation \eqref{eq: multivariate beta} that $R_{FM}^2(\S)$ will be numerically equal to the $R^2$ in the cross-sectional OLS regression of $\bar r_i$ on $\widehat \C_{i,\S}$ as multiplying the matrix of covariates by a non-singular matrix does not change the $R^2$ in OLS regressions. Thus, we can equivalently maximize the $R^2$ in the latter regression using Algorithm \ref{alg: fs fama-macbeth} to obtain the same set of factors $\widehat \S$. This equivalent form of Algorithm \ref{alg: fs fama-macbeth}, together with equation \eqref{eq: mean and covariances}, explains why the proposed algorithm is expected to select factors whose SDF loadings are sufficiently different from zero. 

Another equivalent form of Algorithm \ref{alg: fs fama-macbeth} can be obtained by replacing the multivariate betas \eqref{eq: multivariate beta} in the cross-sectional OLS regression by univariate betas. Indeed, for each $i\in[N]$ and $j\in[p]$, let $\widetilde\beta_{i,j}$ be the slope coefficient in the univariate time series OLS regression of $r_{i,t}$ on $f_{t,j}$, with intercept included:
$$
\widetilde\beta_{i,j} = \left(\sum_{t\in[T]} (f_{t,j} - \bar f_{j})^2\right)^{-1}\left(\sum_{t\in[T]} (f_{t,j} - \bar f_{j}) r_{i,t}\right).
$$
Also, let $\widetilde\beta_i = (\widetilde\beta_{i,1},\dots,\widetilde\beta_{i,p})^{\top}$ be the vector of all univariate betas corresponding to the asset $i$. Under conditions in Section \ref{sub: asymptotic theory}, we will have $T^{-1}\sum_{t\in[T]} (f_{t,j} - \bar f_j)^2>0$ for all $j\in[p]$ with probability approaching one, and whenever this happens, the $R^2$ in the cross-sectional OLS regression of $\bar r_i$ on $\widetilde \beta_{i,\S}$ will be equal to the $R^2$ in the cross-sectional OLS regression of $\bar r_i$ on $\widehat \C_{i,\S}$, yielding another equivalent form of Algorithm \ref{alg: fs fama-macbeth}.
\qedd
\end{remark}
\begin{remark}[Relation to Literature]
Forward selection is one of the most popular machine learning algorithms; see \citet{HTF09}. In the context of high-dimensional linear regression models, it was analyzed, for example, by \citet{Z09}, \citet{DK18}, and \citet{K20}. \citet{EKDN18} studied forward selection in the context of a general criterion function optimization and showed that certain guarantees can be provided as long as the criterion function satisfies the so-called restricted strong convexity conditions. However, there are no available results in the literature in the context of the Fama-MacBeth procedure, which is a central estimation technique in finance. Importantly, existing results cannot be directly applied in this context, with the main challenge stemming from the multi-step nature of the Fama-MacBeth procedure, where the regression covariates evolve as new factors are added, complicating the analysis. Our paper thus contributes to the literature by establishing the forward selection algorithm and its econometric properties in the context of the Fama-MacBeth procedure as a method to select factors explaining the SDF variation when the number of factors is large, potentially much larger than the number of time periods. In the context of factor selection, our algorithm is related to but substantially different from that in \citet{harvey2021lucky}, who proposed maximizing the significance of the added factors instead of the regression $R^2$ but did not provide theoretical guarantees for the proposed algorithm. 
\qedd
\end{remark}

\subsection{Debiased Estimation}
Here, we explain how to debias the estimator $\widehat\psi = (\widehat\psi_1,\dots,\widehat\psi_p)^{\top}$ proposed in the previous section. To develop intuition for the debiasing procedure, fix any component of the vector $\widehat\psi$, say $\widehat \psi_j$, and observe that even though Algorithm \ref{alg: fs fama-macbeth} is expected to select factors whose SDF loadings are sufficiently different from zero, it may not select factors whose SDF loadings are close to zero. The estimator $\widehat\psi_j$ may thus be biased for two reasons. First, $\psi_j$ may be close to zero itself, in which case Algorithm \ref{alg: fs fama-macbeth} may not be able to select it, leading to the bias in $\widehat\psi_j$ as long as $\psi_j\neq 0$. Second, even if $\psi_j$ is not close to zero and the factor $j$ is selected, Algorithm \ref{alg: fs fama-macbeth} may fail to select other factors $k$ whose covariance with returns $\C_{i,k}$ is cross-sectionally correlated with $\C_{i,j}$, leading to the omitted variable bias in $\widehat\psi_j$ as long as $\psi_k\neq 0$, which follows from the OLS characterization of the estimator $\widehat\psi_{\widehat{\S}}$ in equation \eqref{eq: equivalent characterization}.



Two reasons for the bias of $\widehat\psi_j$ suggest that we can debias it by appropriately expanding the set $\widehat{\S}$. Specifically, we need to add to the set $\widehat{\S}$ the factor $j$ itself and also all factors $k$ whose covariance with returns $\C_{i,k}$ is cross-sectionally correlated with $\C_{i,j}$. In practice, since neither $\C_{i,j}$ nor $\C_{i,k}$ is observed, we work with sample covariances $\widehat \C_{i,j}$ and $\widehat \C_{i,k}$.

More formally, we proceed as follows. First, as before, let $\widehat{\S}$ be the set $\S$ produced by Algorithm \ref{alg: fs fama-macbeth}. Second, let $\widehat{\S}_{j}$ be the set produced by Algorithm \ref{alg: fs fama-macbeth} with $R^2_{FM}(\S)$ replaced by the $R^2$ in the cross-sectional OLS regression of $\widehat \C_{i,j}$ on $\widehat \C_{i,\S}$ (without intercept) and with $[p]$ replaced by $[p]\setminus\{j\}$. Third, set $\widehat{\S}_{D,j} = \widehat{\S}\cup \widehat {\S}_{j}\cup \{j\}$. Fourth, let $\widehat\psi$ be the estimator defined in Section \ref{sub: estimation} but with the set $\widehat{\S}$ there replaced by the set $\widehat{\S}_{D,j}$ here. The corresponding component $\widehat\psi_j$ of the vector $\widehat\psi$ will be a debiased estimator of $\psi_j$, which we denote as $\widehat\psi_{D,j}$. By repeating this procedure for all $j\in[p]$, we obtain the full vector $\widehat\psi_D = (\widehat\psi_{D,1},\dots,\widehat\psi_{D,p})^{\top}$. In the next section, we will prove that for each $j\in[p]$, the estimator $\widehat\psi_{D,j}$ is $\sqrt T$-consistent and $\sqrt T(\widehat\psi_{D,j} - \psi_j)$ converges to a centered normal distribution. 

\begin{remark}[Relation to Literature]
The debiasing procedure proposed here is similar in flavor to that proposed in the context of Lasso variable selection for high-dimensional linear regression models by \citet{BCH14}. Although our setting is different, the intuition underlying the debiasing procedure is the same. Related procedures, also in the context of Lasso variable selection for high-dimensional linear regression models, were proposed by \citet{ZZ14}, \citet{JM14}, and \citet{VBRD14}. In the context of SDF loadings estimation, our debiasing procedure is related to but different from that in \citet{feng2020taming}, who relied upon Lasso factor selection.\footnote{In asset pricing, factor scaling is economically meaningful because the price of risk depends on the normalization of the factor. Standard Lasso implementations implicitly penalize coefficients relative to factor variance, so rescaling factors changes the effective penalty weights and hence the selected SDF representation. In contrast, the forward-selection criterion is invariant to factor rescaling, as it depends only on cross-sectional pricing fit. See \citet{nagel2021machine} (pp. 48--49) for a discussion of why unequal scaling of factors can materially affect Lasso-based procedures in empirical asset pricing applications. Forward selection can also be easily adapted for desired modifications, such as (i) ensuring that certain variables are guaranteed to enter the model, (ii) ensuring that some variables enter in groups, or (iii) ensuring that some variables do {\em not} enter if other variables are already included.}

\qedd
\end{remark}

\subsection{Asymptotic Theory}\label{sub: asymptotic theory}
Throughout the rest of the paper, we consider asymptotics where $T$ gets large, i.e. $T\to\infty$. We assume that the number of assets $N$ and the number of factors $p$, as well as all coefficients and distributions of random vectors in the model, are allowed to vary with $T$. In most cases, however, we keep this dependence implicit. 

Let $k,K>0$ be constants such that $k\leq K$. Also, let $\bar s_{T}\geq 1$ be an integer that is allowed to depend on $T$ (and on $N$ and $p$ through $T$). In addition, let $\C=(\C_{i,j})_{i\in[N],j\in[p]}$ be the matrix of covariances between returns $r_{i,t}$ and factors $f_{t,j}$ and let $\widehat{\C}=(\widehat{\C}_{i,j})_{i\in[N],j\in[p]}$ be the corresponding matrix of sample covariances. Moreover, for brevity of notation, for all $t\in[T]$, let $v_t = (v_{t,1},\dots,v_{t,p})^\top = f_t - \E[f_t]$ be the demeaned version of the vector of factors $f_t$. Finally, to state some of the assumptions below in a compact form, let $u_{1,t} = 1$ and $u_{2,t} = \psi^{\top}v_t$ for all $t\in[T]$ and $u_{i+2,t} = \varepsilon_{i,t}$ and $u_{i+2+N,t} = \beta_{i}^{\top}v_t$ for all $i\in[N]$ and $t\in[T]$. Similarly, let $w_{t,1} = 1$ and $w_{t,2} = \psi^{\top}v_t$ for all $t\in[T]$ and $w_{t, j+2} = v_{t,j}$ for all $t\in[T]$ and $j\in[p]$.

\begin{assumption}[Uniform Law of Large Numbers, I]\label{as: ulln}
We have
$$
\max_{i\in[2N+2]}\max_{j\in[p+2]}\left|\frac{1}{T}\sum_{t\in[T]} u_{i,t}w_{t,j} - \E[u_{i,t}w_{t,j}] \right|\lesssim_P \sqrt{\frac{\log(Np)}{T}}.
$$
\end{assumption}

Assumption \ref{as: ulln} is a quantitative version of the uniform law of large numbers. For example, given that $\E[\varepsilon_{i,t}v_{t,j}] = 0$, we expect from the classical central limit theorems for time series data that $T^{-1}\sum_{t=1}^T \varepsilon_{i, t}v_{t,j} \lesssim_P \sqrt{1/T}$. Provided that the random variables $\varepsilon_{i, t}v_{t,j}$ have sufficiently light tails, we then also expect that $T^{-1}\sum_{t\in[T]} \varepsilon_{i, t}v_{t,j} \lesssim_P \sqrt{\log (Np)/T}$ uniformly over $(i,j)\in[N]\times[p]$. Assumption \ref{as: ulln} imposes conditions of this type on various random variables appearing in our analysis below. Conditions like Assumption \ref{as: ulln} are very common in the literature on high-dimensional estimation.

\begin{assumption}[SDF Loadings Sparsity]\label{as: sdf loadings}
There exists a set $\S_0 \subset[p]$ such that
$$
|\S_0|\leq \bar s_{T},\quad \|\C\psi - \C_{\S_0}\psi_{\S_0}\|_2^2 \lesssim \frac{N\log(Np)}{T},\quad\text{and}\quad  \|\psi_{\S_0^c}\|_1^2\lesssim \frac{|\S_0|\log(Np)}{T}.
$$
\end{assumption}

Assumption \ref{as: sdf loadings} states that the vector of SDF loadings $\psi$ is approximately sparse. In particular, it requires that a sub-vector $\psi_{\S_0}$ of the vector $\psi$ of size at most $\bar s_{T}$ is sufficient to explain most of the cross-sectional variation of mean returns; see equation \eqref{eq: mean and covariances} for the relationship between SDF loadings and mean returns. Some researchers, e.g. \citet{nagel2021machine}, argue against sparsity in finance but in different contexts, e.g. in the context of explaining cross-sectional variation of mean returns by past returns. We emphasize that the type of sparsity we impose is different. Indeed, {\em all} cross-sectional variation of mean returns is explained by just one factor: the stochastic discount factor. Thus, if it happens that one of the factors in the vector $f_t$ is equal to the stochastic discount factor, Assumption \ref{as: sdf loadings} will be satisfied with a singleton set $\S_0$. Also, although we allow the set $\S_0$ to depend on $T$ (and on $N$ and $p$ through $T$), we omit this dependence for brevity of notation.

\begin{assumption}[Sparse Eigenvalues, I]\label{as: sparse eigenvalues 1}
We have 
$$
\lambda_{\min}\left(\frac{1}{T}\sum_{t\in[T]} ( f_{t,\S} - \bar f_{\S})( f_{t,\S} - \bar f_{\S})^{\top}\right)>0
$$
for all $\S\subset[p]$ such that $|\S|\leq 3\bar s_{T}+1$ \wpa.
\end{assumption}

Provided $\lambda_{\min}(\E[v_tv_t^{\top}])$ is bounded away from zero and $\bar s_{T}$ does not grow too fast, Assumption \ref{as: sparse eigenvalues 1} can be proven to hold if the factors $f_t$ exhibit relatively weak time series dependence and have sufficiently light tails; e.g. see \citet{BC09}. 

\begin{assumption}[Sparse Eigenvalues, II]\label{as: sparse eigenvalues 2}
We have
$$
k\leq \lambda_{\min}\left(\frac{\C_{\S}^{\top}\C_{\S}}{N}\right) \leq \lambda_{\max}\left(\frac{\C_{\S}^{\top}\C_{\S}}{N}\right) \leq K
$$
for all $\S\subset[p]$ such that $|\S|\leq 3 \bar s_{T} + 1$.
\end{assumption}

The lower bound in Assumption \ref{as: sparse eigenvalues 2} is similar in flavor to that in Assumption \ref{as: sparse eigenvalues 1}. It means that there is no multicollinearity in the vectors of covariances $\C_{\{j\}}$ when we restrict attention to a relatively small number of factors $j$. The upper bound should be viewed as a mild technical condition. Note also that Assumption \ref{as: sparse eigenvalues 2} is our key identification condition. Indeed, \citet{lewellen2010skeptical} showed that without this assumption, it is possible that factors with zero SDF loadings can yield high $R^2$ in the cross-sectional OLS regression as long as they are correlated with the factors explaining the SDF. Assumption \ref{as: sparse eigenvalues 2} ensures that this phenomenon does not occur.

\begin{assumption}[Estimator]\label{as: estimator sparsity}
We have $|\widehat{\S}|\leq \bar s_{T}$ \wpa. In addition, for some constant $c>0$, we have $|\widehat{\S}|\geq (1+c)(K/k)|\S_0|\log T$ \wpa.  Moreover, $|\widehat{\S}|\lesssim_P |\S_0|\log T$.
\end{assumption}

The first part of this assumption means that we carry out at most $\bar s_{T}$ rounds in our forward selection Fama-MacBeth procedure, Algorithm \ref{alg: fs fama-macbeth}, so that we select at most $\bar s_{T}$ factors. This condition is rather natural given the assumed sparsity of the true vector $\psi$. The second part of Assumption \ref{as: estimator sparsity} implies that the number of selected factors $|\widehat\S|$ should be larger than the number of important factors $|\S_0|$ at least by a multiplicative constant proportional to $\log T$. We impose this condition because the forward selection algorithm may occasionally select unimportant factors by chance. The third part of Assumption \ref{as: estimator sparsity} gives an upper bound on the number of selected factors but is less essential than the first two parts. Indeed, selecting more factors than the number in the upper bound will undermine the convergence rate of our estimator $\widehat\psi$ but will not make it inconsistent unless way too many factors are selected. 

\begin{assumption}[Expected Returns]\label{as: returns}
We have $\max_{i\in[N]}\E[r_{i,t}^2]\lesssim 1$.
\end{assumption}

\begin{assumption}[Growth Condition, I]\label{as: growth conditions}
We have $\bar s_{T}\log(N T p) = o(T)$. 
\end{assumption}

Assumption \ref{as: returns} is a very mild regularity condition that is expected to hold in most cases. Assumption \ref{as: growth conditions} restricts how fast $\bar s_{T}$, $N$, and $p$ are allowed to grow relative to $T$. Most importantly, it allows the number of factors $p$ to be much larger than the number of time periods $T$. Together with Assumption \ref{as: sdf loadings}, it also requires that the size of the approximating vector $\psi_{\S_0}$ is small relative to the number of time periods $T$. 

\begin{theorem}[Convergence Rate]\label{thm: convergence rate}
Suppose that Assumptions \ref{as: ulln} -- \ref{as: growth conditions} are satisfied. Then
$$
\|\widehat\psi - \psi\|_2 \lesssim_P \sqrt{\frac{|\S_0|\log^2(NTp)}{T}} \quad\text{and}\quad \|\widehat\psi - \psi\|_1 \lesssim_P \sqrt{\frac{|\S_0|^2\log^3(NTp)}{T}}.
$$
\end{theorem}
\begin{remark}[Consistency of $\widehat\psi$]
This theorem implies that the proposed estimator $\widehat\psi$ is consistent in the $\ell_2$-norm if $|\S_0|\log^2(N T p)=o(T)$ and is consistent in the $\ell_1$-norm if $|\S_0^2|\log^3(N T p)=o(T)$. Thus, the estimator remains consistent even if the number of factors $p$ and the number of assets $N$ are growing much faster than the number of time periods $T$. Also, none of the assumptions of Theorem \ref{thm: convergence rate} requires that $N\to\infty$ as $T\to\infty$, which implies that the estimator is consistent even if the number of assets $N$ remains bounded as $T$ gets large. In addition, note that the bounds derived in Theorem \ref{thm: convergence rate} are $\sqrt{\log(NTp)}$ worse than the bounds typically expected for the Lasso-based estimators. We do not know if this is an artifact of our proof technique or this is because forward selection generally requires more selected variables to guarantee the same fit as the Lasso-based selection.
\qedd
\end{remark}

To formulate the asymptotic normality result for the debiased estimator $\widehat\psi_D$, for all $j\in[p]$, let $\eta_j = (\eta_{j,1},\dots,\eta_{j,p})^{\top}\in\R^p$ be the vector defined by $\eta_{j,j} = 0$ and
\begin{equation}\label{eq: eta definition}
\eta_{j,\{j\}^c} = \left(\E[v_{t,\{j\}^c}v_{t,\{j\}^c}^{\top}]\right)^{-1}\E[v_{t,\{j\}^c}v_{t, j}],
\end{equation}
so that $\eta_{j,\{j\}^c}$ is the vector of coefficients of the time-series least-squares projection of the demeaned factor $v_{t,j}$ on other demeaned factors $v_{t,\{j\}^c}$. Also, for all $t\in[T]$, let
$
z_{t,j} = v_{t,j} - \eta_{j,\{j\}^c}^{\top}v_{t,\{j\}^c}
$
be the corresponding residual and let $\sigma_{z,j}^2 = \E[z_{t,j}^2]$ be the residual variance. The asymptotic normality result requires several additional assumptions but, for brevity of the main text, we have placed them in Appendix \ref{sec: additional assumptions}.

\begin{theorem}[Asymptotic Normality]\label{thm: asy normality}
Suppose that Assumptions \ref{as: ulln} -- \ref{as: returns} above and Assumptions \ref{as: first stage sparsity} -- \ref{as: growth conditions 2} listed in Appendix \ref{sec: additional assumptions} are satisfied. Then for all $j\in[p]$,
$$
\sqrt T(\widehat\psi_{D,j} - \psi_j)\to_d N(0,\sigma_{\psi,j}^2),
$$
where
$$
\sigma_{\psi,j}^2 = \lim_{T\to\infty} \frac{1}{T}\sum_{t=1}^T\sum_{s=1}^T \E\left[ \frac{(z_{t, j} m_t - \E[z_{t, j} m_t])(z_{s, j} m_s - \E[z_{s, j} m_s])}{\sigma_{z, j}^4} \right].
$$
\end{theorem}
\begin{remark}[Comparison of $\widehat\psi$ and $\widehat\psi_D$]
It is possible to show that the individual components $\widehat\psi_j$ of the estimator $\widehat\psi$ converge to the corresponding true values $\psi_j$ with the rate that is slower than $\sqrt{1/T}$. It is therefore fair to say that for each $j\in[p]$, the estimator $\widehat\psi_{D,j}$ outperforms $\widehat\psi_j$. However, it is well-known that this component-wise comparison does not mean that the vector-valued estimator $\widehat\psi_D$, as a whole, outperforms $\widehat\psi$. For example, because of the convergence result in Theorem \ref{thm: asy normality}, it is expected that the estimation error $\|\widehat\psi_D-\psi\|_2$ is typically of order $\sqrt{p/T}$, which could be much worse than the estimation error $\|\widehat\psi - \psi\|_2$ derived in Theorem \ref{thm: convergence rate}. It is therefore important to emphasize that both estimators have their advantages: the estimator $\widehat\psi$ is useful when we are interested in the values of all SDF loadings, whereas $\widehat\psi_D$ is useful when we are interested in the SDF loading of a particular factor.
\qedd
\end{remark}

\begin{remark}[Estimating $\sigma_{\psi,j}^2$]
Observe that $\sigma_{\psi,j}^2$ is the long-run variance of the normalized sample average $T^{-1/2}\sum_{t\in[T]}(z_{t,j}m_t - \E[z_{t,j}m_t])/\sigma_{z,j}^2$, which suggests the following plug-in estimator thereof. First, our assumptions in Appendix \ref{sec: additional assumptions} imply that the vector $\eta_{j,\{j\}^c}$ is bounded in the $\ell_1$-norm. We thus can estimate $z_{t,j}$ as $\widehat z_{t,j} = (f_{t,j} - \bar f_j) - \widehat\eta_{j,\{j\}^c}^\top(f_{t,\{j\}^c} - \bar f_{\{j\}^c})$, where $\widehat\eta_{j,\{j\}^c}$ is the Lasso estimator of $f_{t,j} - \bar f_j$ on $f_{t,\{j\}^c} - \bar f_{\{j\}^c}$. Second, because of the consistency of $\widehat\psi$ implied by Theorem \ref{thm: convergence rate}, we can estimate $m_t$ as $\widehat m_t = 1 - \widehat\psi^\top(f_t - \bar f)$. Third, we can estimate $\sigma_{z,j}^2$ as $\widehat\sigma_{z,j}^2 = T^{-1}\sum_{t\in[T]}\widehat z_{t,j}^2$. Finally, we can estimate $\sigma_{\psi,j}^2$ by applying the standard long-run variance estimation formulas to $T^{-1/2}\sum_{t\in[T]}(\widehat z_{t,j}\widehat m_t - \bar z_{j}\bar m_t])/\widehat \sigma_{z,j}^2$, where we set $\bar z_j = T^{-1}\sum_{t\in[T]}\widehat z_{t,j}$ and $\bar m_t = T^{-1}\sum_{t\in[T]}\widehat m_t$. It is straightforward to prove that the estimator $\widehat\sigma_{\psi,j}^2$ is consistent for $\sigma_{\psi,j}^2$ under our conditions.
\qedd
\end{remark}
\begin{remark}[Relation to Literature]
The asymptotic variance $\sigma_{\psi,j}^2$ appearing in Theorem \ref{thm: asy normality} is the same as the asymptotic variance in Theorem 1 of \citet{feng2020taming}. In fact, it follows from the proof of Theorem \ref{thm: asy normality} that our debiased forward-selection-based estimator $\widehat\psi_{D,j}$ is asymptotically equivalent to the corresponding Lasso-based estimator in \citet{feng2020taming}, which can be viewed as a consequence of the results on equivalence of semiparametric estimators in \citet{N94}. Our Theorem \ref{thm: asy normality} thus complements Theorem 1 in \citet{feng2020taming} by expanding the set of machine learning methods that are theoretically proven to work in the context of inference on SDF loadings.\qedd
\end{remark}

\section{Empirical Analysis}\label{sec:Empirical-Analysis}

In this section, we present the data used in the empirical analysis
and the main results of the selection of the higher-order factors by
the forward selection Fama-MacBeth procedure.

\subsection{Data}

We present a description of the data used for the estimation of the
FS-FMB procedure. In order to estimate the model and perform cross-sectional
asset pricing tests, we use data from various sources. Test assets
are 484 characteristic-managed portfolios from \citet{kozak2020shrinking}, ranging from October 1973 to December 2019.\footnote{
Portfolios are from Serhiy Kozak's web page (\href{https://serhiykozak.com/\#/data-code}{https://serhiykozak.com/\#/data-code}).
} These portfolios are constructed by weighting stocks based on their
characteristic signals, which are based on cross-sectional ranks.
Small firms (below 0.01\% of market cap) are excluded. Baseline factors
are the \citet{fama2015five}'s five factors and the momentum factor
of \citet*{jegadeesh1993returns}, which we obtain from Kenneth French's
data library. We refer to these six baseline factors as the FF5M model.
Furthermore, we consider factors for 148 characteristics in 13 themes,
using data from the US equity market, constructed by \citet*{jensen2023there}.\footnote{Factors are capped value
weighted from Global Factor
Data (\href{https://jkpfactors.com/}{https://jkpfactors.com/}).} We refer to these 148 factors as the factor zoo (see, e.g., \citet*{feng2020taming}).

We construct factors based on the FF5M factors. We denote both the
higher-orders and the interactions of FF5M as higher-order factors.
Specifically, for any factor $f_{i}$ in the FF5M, we first form all
the $f_{i}^{2}$ and $f_{i}^{3}$ as candidate factors. Moreover,
we form all the pairwise interactions of degree 2 and 3. That is,
for any two factors $f_{i}$ and $f_{j}$ in the FF5M model, we form
the pairwise interactions of the type $f_{i}\times f_{j}$ and $f_{i}\times f_{j}^{2}$
as candidate factors. Therefore, the total number of candidate higher-order
factors is 57, where there are 6 $f_{i}^{2}$, 6 $f_{i}^{3}$, 15
$f_{i}\times f_{j}$, and 30 $f_{i}\times f_{j}^{2}$.

Table \ref{tab:Summary-statistics} presents summary statistics of
the test assets and factors. For the mean, the standard deviation,
and the quantiles, we report averages across assets. All returns are
monthly, in excess of the US risk-free rate, and reported in percentage.
The average mean return for the test assets is 0.65\% monthly, with
a standard deviation of 4.97\%. The average return for the FF5M factors
is lower and equal to 0.39\% monthly, with a standard deviation of
3.16\%. The average return for the factor zoo is even lower and equal
to 0.23\% monthly, with a standard deviation of 2.71\%. The higher-order
factors have a distribution characterized by more extreme realizations,
as exemplified by the 95\% quantile of the powers of degree 2 and
3 equal, respectively, to 42.75\% and 167.67\%. We note that the higher-order factors are generally non-tradable. Figure \ref{fig:Appendix_Distribution-of-Higher-Order-Factors}
in the Appendix plots the distributions of the returns of the higher-order
factors.

\begin{table}[]

\caption{Summary Statistics\label{tab:Summary-statistics}}

\medskip{}

\begin{minipage}{\textwidth}
\small\singlespacing

This table reports summary statistics of the test assets and factors.
For the mean, the standard deviation and the quantiles, we report
averages across assets. Returns are monthly and in excess of the US
risk-free rate. Test assets are 484 characteristic-managed portfolios from \citet{kozak2020shrinking}, ranging from October 1973 to December 2019.
The Fama-French five factors and the momentum factor (FF5M) are from
Ken French's data library. ``Interactions'' correspond to pairwise
interactions of the FF5M factors of degree 2 (i.e., $f_{i}\times f_{j}$)
and degree 3 (i.e., $f_{i}^{2}\times f_{j}$). ``Powers'' correspond
to the power of the individual FF5M factors of degree 2 (i.e., $f_{i}^{2}$)
and degree 3 (i.e., $f_{i}^{3}$). The factor zoo includes 148 factors
based on US equity data constructed by \citet{jensen2023there}. 
\end{minipage}

\medskip{}

\begin{centering}
\begin{tabular}{lccccccc}
\hline 
$R_{i,t}$ (\%) & Mean & Std & $q_{5\%}$ & $q_{50\%}$ & $q_{95\%}$ & N & T\tabularnewline
\hline 
Test Assets & 0.65 & 4.97 & -7.58 & 0.9 & 8.1 & 484 & 555\tabularnewline
FF5M & 0.39 & 3.16 & -4.63 & 0.4 & 5.13 & 6 & 555\tabularnewline
Interactions degree 2 & -0.48 & 16.47 & -17.82 & -0.02 & 14.12 & 15 & 555\tabularnewline
Interactions degree 3 & 0.53 & 20.14 & -6.6 & 0.02 & 8.44 & 30 & 555\tabularnewline
Powers degree 2 & 11.1 & 28.62 & 0.03 & 3.34 & 42.75 & 6 & 555\tabularnewline
Powers degree 3 & -9.83 & 581.32 & -147.67 & 0.23 & 167.67 & 6 & 555\tabularnewline
Factor zoo & 0.23 & 2.71 & -3.5 & 0.19 & 4.14 & 148 & 555\tabularnewline
\hline 
\end{tabular}
\par\end{centering}
\end{table}

\subsection{Selected Factors by Forward Selection Fama-MacBeth Regressions}

Our goal is to study whether and which of the higher-order factors
constructed from the six factors in the baseline FF5M are important
in explaining the SDF. To select the most relevant higher-order factors
to the SDF, we use the FS-FMB procedure that
we have discussed in detail in Section \ref{sec:Methodology} above.
We start with the six baseline FF5M factors and add the higher-order
interactions one at a time. 

We present the main estimation results from the second step of the
FS-FMB procedure in Table \ref{tab:Cross-sectional-performance}.
Panel A shows the performance of the baseline models using the
484 characteristic-managed portfolios from \citet{kozak2020shrinking}.
The baseline models include the standard CAPM, the \citet{fama1993common}'s
3-factor model (FF3), the \citet{fama2015five}'s 5-factor model (FF5)
and the 5-factor model augmented with the momentum factor of \citet{jegadeesh1993returns}
(FF5M). The adjusted cross-sectional R-squared for the CAPM is just
0.032, and the estimated intercept (0.009) is significantly different
from zero with a $t$-statistic of 3.6. The adjusted cross-sectional
R-squared reaches 0.312 for the FF5M model, while the estimate for
the intercept (0.004) remains significantly different from zero with
a $t$-statistic of 2.2. 

Panel B refers to the FS-FMB procedure. At
each step, we augment the FF5M by one of the additional higher-order
factors, which we select in order to maximize the adjusted cross-sectional
R-squared. The selection procedure stops when the improvement in the
R-squared in a given step is smaller than 1 pp. The forward selection FMB procedure
selects 7 higher-order factors: the SMB factor squared (SMB2), the
interaction between SMB2 and the momentum factor (SMB2{*}Mom), the
interaction between the momentum factor squared and the profitability
factor (Mom2{*}RMW), the market factor squared (Mkt-RF2), the interaction
between Mkt-RF2 and the profitability factor (Mkt-RF2{*}RMW), the
interaction between the market and size factors (Mkt-RF{*}SMB) and
the interaction between the value factor squared and the market factor
(HML2{*}Mkt-RF). Of these 7 higher-order factors that are selected,
2 are higher-orders (SMB2 and Mkt-RF2) and 5 are interactions (SMB2{*}Mom,
Mom2{*}RMW, Mkt-RF2{*}RMW, Mkt-RF{*}SMB, and HML2{*}Mkt-RF). Interestingly,
none of the higher-orders of degree 3 are selected. 

The adjusted cross-sectional R-squared increases from 31.2\% for the
baseline FF5M model to 41\% in the first step when SMB2 is included,
and gradually increases to 59\%, in the last step when all 7 factors
are included and the procedure stops. Across the seven steps, the
estimate for the intercept is around 0.003, with a $t$-statistic
in the final step of 1.819, which is below the threshold for significance
at the 5\% level.

\begin{table}[]
\caption{Cross-Sectional Performance\label{tab:Cross-sectional-performance}}

\medskip{}

\begin{minipage}{\textwidth}
\small\singlespacing

This table reports the results from the second step of the FMB method.
Panel A refers to baseline models: the standard CAPM, the Fama-French
3-factor, the Fama-French 5-factor and the Fama-French 5-factor and
momentum models. The table reports the adjusted cross-sectional R-squared,
the estimate for the intercept ($\alpha$) and associated t-statistic.
The t-statistic is based on Newey-West corrected standard errors.
Panel B refers to the FS-FMB procedure. At each step, we augment the
FF5M by one additional higher-order factor ($h_{j}$) which maximizes
the adjusted cross-sectional R-squared. The method stops when the
improvement in the R-squared in a given step is smaller than 1 pp.
\end{minipage}

\medskip{}

\centering{}%
\begin{tabular}{llccc}
\hline 
\multicolumn{5}{c}{Panel A: Baseline Models}\tabularnewline
\hline 
\# & Model & Adj. R-squared & $\alpha$ & t-stat ($\alpha$)\tabularnewline
\hline 
1 & CAPM & 0.032 & 0.009 & 3.635\tabularnewline
2 & FF3 & 0.126 & 0.01 & 4.701\tabularnewline
3 & FF5 & 0.275 & 0.006 & 2.899\tabularnewline
4 & FF5M & 0.312 & 0.004 & 2.201\tabularnewline
\hline 
\multicolumn{5}{c}{Panel B: FS-FMB procedure}\tabularnewline
\hline 
Step & $h_{j}$ & Adj. R-squared & $\alpha$ & t-stat ($\alpha$)\tabularnewline
\hline 
1 & SMB2 & 0.41 & 0.003 & 1.821\tabularnewline
2 & SMB2{*}Mom & 0.467 & 0.004 & 2.02\tabularnewline
3 & Mom2{*}RMW & 0.495 & 0.003 & 1.324\tabularnewline
4 & Mkt-RF2 & 0.529 & 0.003 & 1.769\tabularnewline
5 & Mkt-RF2{*}RMW & 0.552 & 0.003 & 1.636\tabularnewline
6 & Mkt-RF{*}SMB & 0.572 & 0.004 & 2.161\tabularnewline
7 & HML2{*}Mkt-RF & 0.587 & 0.003 & 1.819\tabularnewline
\hline 
\end{tabular}
\end{table}

We then examine the SDF loadings of the higher-order factors selected
by the FS-FMB procedure. The results are summarized in Table \ref{tab:Estimates-Loadings-Main}.
Panel A shows the non-debiased estimates based on the standard FMB
two-pass regression, while Panel B presents the debiased estimates
as outlined above. The directions
of the point estimates based on the non-debiased procedure and the debiased
procedure are the same. Moreover, the estimates from the two procedures
are similar in magnitude. The estimates based on the non-debiased
procedure are all significant at the 5\% level, while the estimates based
on the debiased procedure are all significant at the 10\% level with
5 of the 7 SDF loading estimates being significant at the 5\% level.
Overall, the results show that the higher-order factors are important
components of the SDF and significantly price assets in the cross-section. 

\begin{table}[]
\caption{SDF loadings\label{tab:Estimates-Loadings-Main}}
\medskip{}

\begin{minipage}{\textwidth}
\small\singlespacing
This table reports the estimates for the SDF loadings ($b_{j}$) associated
with the selected higher-order factors. Panel A reports non-debiased
estimates while Panel B reports debiased estimates. The t-statistic
is based on Newey-West corrected standard errors. Higher-order factors
are selected by the FS-FMB procedure (see Table \ref{tab:Cross-sectional-performance}).
\end{minipage}
\medskip{}

\centering{}%
\begin{tabular}{l>{\centering}p{2cm}>{\centering}p{2cm}>{\centering}p{2cm}>{\centering}p{2cm}}
\hline 
& \multicolumn{2}{c}{Panel A: Non-Debiased Estimates} & \multicolumn{2}{c}{Panel B: Debiased Estimates}\tabularnewline
\hline 
$h_{j}$ & $b_{j}$ & $t(b_{j})$ & $b_{j}$ & $t(b_{j})$\tabularnewline
\hline 
SMB2 & 1.158 & 2.43 & 1.403 & 1.837\tabularnewline
SMB2{*}Mom & -1.193 & -4.074 & -1.004 & -1.812\tabularnewline
Mom2{*}RMW & -1.38 & -4.105 & -1.178 & -2.882\tabularnewline
Mkt-RF2 & 1.637 & 5.619 & 1.878 & 4.606\tabularnewline
Mkt-RF2{*}RMW & -3.53 & -4.169 & -3.87 & -4.301\tabularnewline
Mkt-RF{*}SMB & -1.404 & -3.032 & -2.012 & -3.636\tabularnewline
HML2{*}Mkt-RF & -1.385 & -2.823 & -1.606 & -2.039\tabularnewline
\hline 
\end{tabular}
\end{table}

Figure \ref{fig:Actual-vs.-predicted} plots predicted against realized
returns for all the test assets. Panel A shows the performance of
the FF5M model and panel B presents the higher-order model. The FF5M
has a limited ability in explaining the average returns (Panel A),
in line with the relatively low R-squared of the model. On the other
hand, there is a strong positive relationship between the average
realized returns and the predicted returns for the higher-order model
selected by the FS-FMB procedure.

\begin{figure}[H]
  \centering
  \caption{Actual vs. Predicted Returns\label{fig:Actual-vs.-predicted}}

  \medskip

  \begin{minipage}{\textwidth}
    \small\singlespacing
    This figure plots actual mean returns of the test assets against the
    mean predicted returns from the FF5M model (subplot a) and the FF5M
    model augmented with the higher-order factors selected by the FS-FMB
    procedure (subplot b, see Table~\ref{tab:Cross-sectional-performance}). The
    predicted return is the OLS estimate of the betas times the estimated
    prices of risk.
  \end{minipage}

  \medskip

  \subfloat[\textbf{FF5M model}]{
    \includegraphics[width=0.45\textwidth]{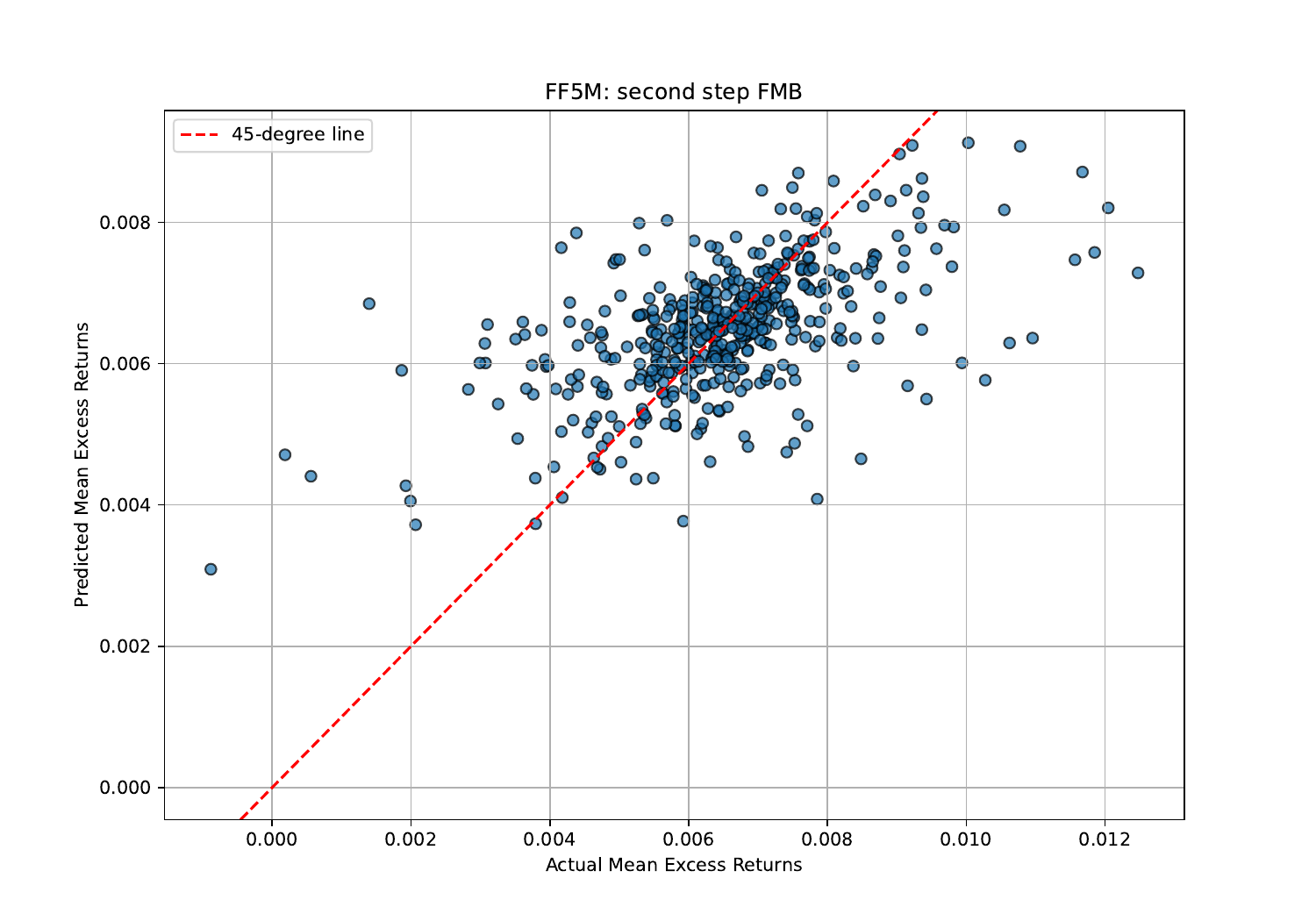}
    \label{fig:actual-vs-pred-ff5m}
  }
  \hfill
  \subfloat[\textbf{FF5M + higher-order factors}]{
    \includegraphics[width=0.45\textwidth]{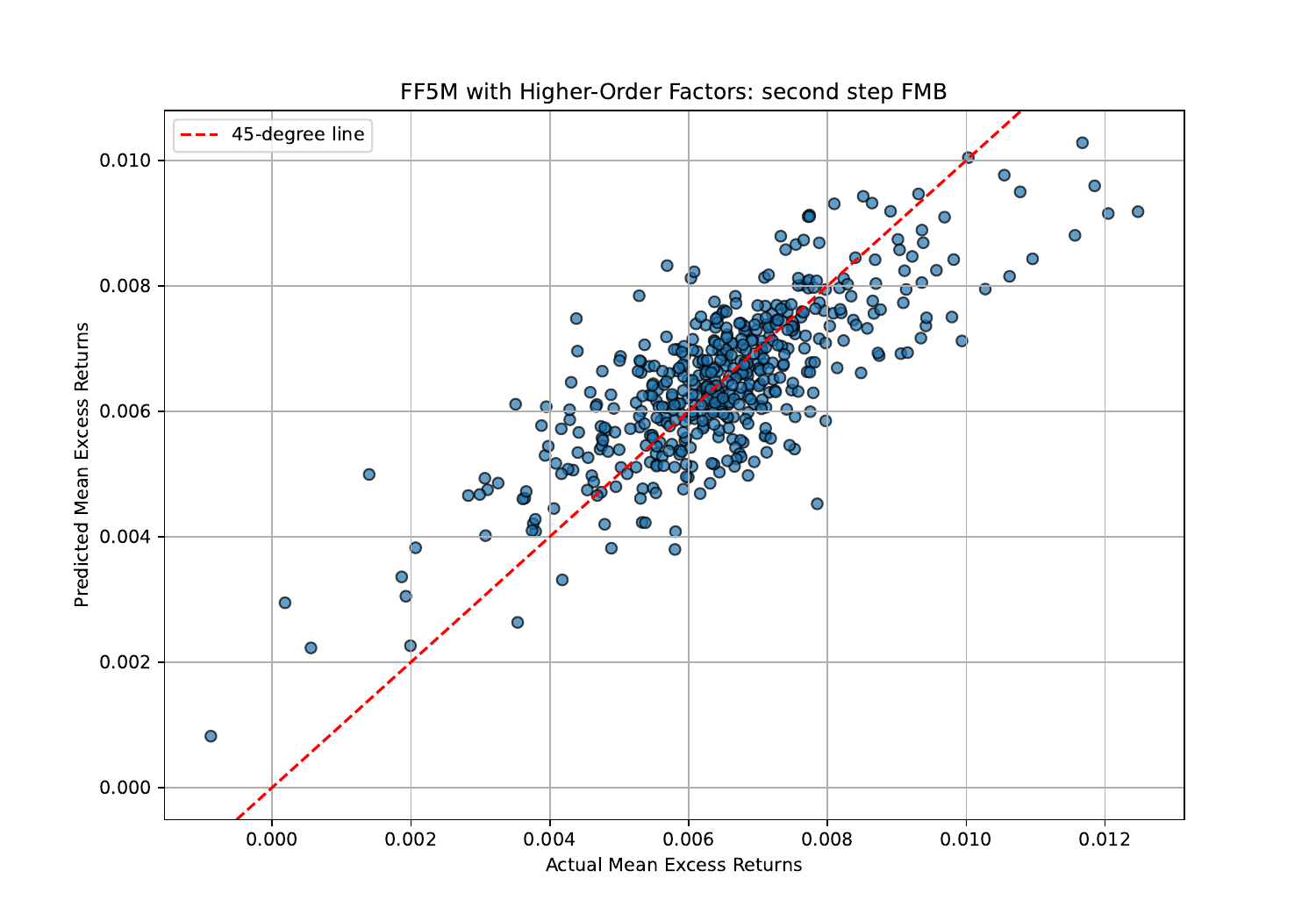}
    \label{fig:actual-vs-pred-ho}
  }
\end{figure}

We note that the predicted return used in the graph is the OLS estimate
of the betas from the first step of the two-pass FMB procedure times
the estimated prices of risk. That is, we do not impose the no-arbitrage
condition for the tradable factors that pins down the price of risk
by the sample mean of the factor. Doing so would, by construction,
result in an even worse fit for the baseline FF5M
models where all the factors are tradable. We present the results
in Table \ref{tab:Cross-sectional-performance-Restricted-Models}
in the Online Appendix where we impose the restriction for the tradable
factors. The results confirm the importance of higher-order factors
in pricing assets in the cross-section.

\subsection{Out-of-Sample}

The main results presented in Table \ref{tab:Cross-sectional-performance} are based on in-sample estimation. We further provide evidence of the out-of-sample performance of the higher-order model selected by the FS-FMB procedure. We consider two types of out-of-sample tests: (1) out-of-sample in terms of asset space and (2) out-of-sample in the time dimension.

\subsubsection{Asset space}

In Table \ref{tab:Forward-Regression-Method-CV}, we present an additional set of results based on cross-validation applied to the FS-FMB procedure. The cross-validation is applied to the second step of the FS-FMB procedure. This is an out-of-sample exercise in terms of asset space, where we select higher-order factors and estimate SDF loadings from the training assets, and then evaluate the out-of-sample performance in the out-of-sample assets.

Specifically, at each step of the FS-FMB procedure, we divide the sample into 5 equal-sized, non-overlapping folds. For each fold $k$, we use the remaining 4 folds as the training set and the left-out fold as the test set. In each training step, we estimate the cross-sectional regression of average returns on the betas and obtain estimates for the risk prices. Next, we use the estimated risk prices to compute the adjusted R-squared on the test set, where the degrees-of-freedom adjustment uses the number of test assets in each fold. This process is repeated 5 times and the adjusted cross-sectional R-squared associated with a given higher-order factor is the average, at each step, of the adjusted cross-sectional R-squareds obtained in the 5-fold cross-validation. The procedure stops when the improvement in the cross-validated adjusted R-squared falls below 1 pp. Table \ref{tab:Forward-Regression-Method-CV} reveals that the FS-FMB procedure with cross-validation (CV) selects a total of 7 higher-order factors, with an in-sample adjusted cross-sectional R-squared in the final step of 0.582. The selected factors are the same as the first 6 from the baseline method. The cross-validated adjusted cross-sectional R-squareds, which represent out-of-sample measures of fit of the cross-sectional model, reach 0.471 in the final step. Overall, these results show that the selected higher-order factor model has stable out-of-sample performance. In Appendix~\ref{sec:optimal_stopping}, we further show that the results are robust to the choice of stopping threshold by selecting it optimally and entirely out-of-sample via nested cross-validation (Table~\ref{tab:Forward-Regression-Method-CV-OptimalStop}).

\begin{table}[H]
\caption{FS-FMB Method with Cross-Validation\label{tab:Forward-Regression-Method-CV}}

\medskip{}

\begin{minipage}{\textwidth}
\small\singlespacing
This table reports the results from the second step of the FMB regression procedure for the FS-FMB procedure with 5-fold cross-validation. At each
step, we include one additional higher-order factor which maximizes
the adjusted cross-sectional R-squared computed with cross-validation
using 5-fold cross-validation. The procedure stops when the improvement in the cross-validated
adjusted R-squared in a given step is smaller than 1 pp. For each step, we report the
selected factor ($h_{j}$), the adjusted cross-sectional R-squared,
the adjusted cross-validated (CV) cross-sectional R-squared, the estimate
for the intercept ($\alpha$) and associated t-statistic. The t-statistic
is based on Newey-West corrected standard errors. The adjusted cross-validated
cross-sectional R-squared is the average at each step of the adjusted
R-squareds obtained in the 5-fold cross-validation, where the degrees-of-freedom
adjustment in each fold uses the number of test assets.
\end{minipage}

\medskip{}

\centering{}%
\begin{tabular}{lcccccc}
\hline 
Step & $h_{j}$ & Adj. R-squared & Adj. CV R-squared & & $\alpha$ & t-stat ($\alpha$)\tabularnewline
\hline 
1 & SMB2 & 0.41 & 0.328 & & 0.003 & 1.821\tabularnewline
2 & SMB2{*}Mom & 0.467 & 0.382 & & 0.004 & 2.02\tabularnewline
3 & Mom2{*}RMW & 0.495 & 0.4 & & 0.003 & 1.324\tabularnewline
4 & Mkt-RF2 & 0.529 & 0.426 & & 0.003 & 1.769\tabularnewline
5 & Mkt-RF2{*}RMW & 0.552 & 0.446 & & 0.003 & 1.636\tabularnewline
6 & Mkt-RF{*}SMB & 0.572 & 0.458 & & 0.004 & 2.161\tabularnewline
7 & Mom2{*}HML & 0.582 & 0.471 & & 0.004 & 2.145\tabularnewline
\hline 
\end{tabular}
\end{table}

\subsubsection{Time periods}

We next evaluate the out-of-sample performance in the time-series
of the higher-order factor model against the benchmark factor models
(i.e., CAPM, FF3, FF5, and FF5M). Specifically, we first split the
sample into two 50-50 subsamples based on the time series, where we
denote the first half of the sample as the training sample and the
second half as the out-of-sample. Then, we estimate the covariances
and the SDF loadings in the training sample based on the different
factor models. We denote the $R^{2}$ from these estimations as $R_{train}^{2}$.
Lastly, we compare the average asset returns in the out-of-sample
to the predicted asset returns using covariances and the SDF loadings
estimated from the training sample to obtain the out-of-sample $R^{2}$,
or $R_{oos}^{2}$. We note that both the $R_{train}^{2}$ and the
$R_{oos}^{2}$ are likely to be lower than the full sample $R^{2}$
for any model due to estimation errors, especially related to average
returns. It is well-known that estimations of the average equity returns
are noisy and require a long time series. Therefore, our focus is
comparing $R_{oos}^{2}$ for the higher-order factor model against
those of the benchmark factor models. 

Table \ref{tab:Comparison-of-R-squareds} presents the results. The
$R_{oos}^{2}$s of CAPM and the FF3 factor model are close to zero.
The $R_{oos}^{2}$s of the FF5 factor model and the FF5M factor model
are 7\% and 8.9\%, representing an improvement from the CAPM and the
FF3 factor model. The $R_{oos}^{2}$ of the higher-order factor model
selected from the FS-FMB procedure is significantly larger than the
benchmark factor models at 14.8\%, which is about twice as high as
the $R_{oos}^{2}$s of FF5 and FF5M factor models. The $R_{oos}^{2}$ of the higher-order factor model where higher-order factors are selected in the training sample is similar at 15.8\%.

One concern is that the 50-50 split where the first half of the sample
is the training sample and the second half of the sample is the out-of-sample
may be arbitrary. As a robustness check, we repeat the analyses by
randomly selecting, for 1,000 times, 50\% of the sample as the training
sample and the remaining as the out-of-sample. We report the results
in Table \ref{tab:Comparison-of-R-squareds-Robust} in the Online
Appendix. Table \ref{tab:Comparison-of-R-squareds-Robust} delivers
the same message -- the $R_{oos}^{2}$ of the higher-order factor
model is significantly higher than those of the benchmark factor models.

\subsection{Factor Zoo\label{sec:Culling-the-Factor}}

In this subsection, we show that the baseline FF5M model, augmented with the higher-order factors selected by the FS-FMB procedure, substantially reduces the majority of the factors in the factor zoo, or zoo factors.

We employ two empirical methods to support this argument. First, we perform FMB regressions to test the significance of the SDF loading of each factor in the factor zoo, controlling for our higher-order factors. Second, we construct factor mimicking portfolios for the seven higher-order factors to conduct time-series asset pricing tests on the zoo factors.

\begin{table}[H]
\caption{Comparison of R-squareds\label{tab:Comparison-of-R-squareds}}
\medskip{}

\begin{minipage}{\textwidth}
\small\singlespacing
This table presents the in-sample and out-of-sample $R^{2}$ under
different factor models. The factor models include the CAPM, the Fama-French
3-factor model (FF3), the Fama-French 5-factor model (FF5), the Fama-French
5-factor plus momentum factor model (FF5M), and the higher-order factor
model selected by the FS-FMB procedure in the full sample (column 5) and in the training sample (column 6). We first split the sample into
two 50-50 subsamples based on the time series, where we denote the
first half of the sample as the training sample and the second half
as the out-of-sample. Then, we estimate the covariances and the SDF
loadings in the training sample based on the different factor models.
We denote the $R^{2}$ from these estimations as $R_{train}^{2}$.
Lastly, we compare the average asset returns in the out-of-sample
to the predicted asset returns using covariances and the SDF loadings
estimated from the training sample to obtain the out-of-sample $R^{2}$,
or $R_{oos}^{2}$. The table reports the $R_{train}^{2}$ and $R_{oos}^{2}$
for each model.
\end{minipage}
\medskip{}

\centering{}%
\begin{tabular}{c>{\centering}p{1.5cm}>{\centering}p{1.5cm}>{\centering}p{1.5cm}>{\centering}p{1.5cm}>{\centering}p{2.5cm}>{\centering}p{2.5cm}}
\hline 
& (1) & (2) & (3) & (4) & (5) & (6)\tabularnewline
& CAPM & FF3 & FF5 & FF5M & Higher-Order (Full) & Higher-Order (Train)\tabularnewline
\hline 
$R_{train}^{2}$ & 0.003 & 0.204 & 0.273 & 0.312 & 0.494 & 0.449\tabularnewline
$R_{oos}^{2}$   & 0.001 & 0.012 & 0.070 & 0.089 & 0.148 & 0.158\tabularnewline
\hline 
\end{tabular}

\end{table}

\subsubsection{Cross-sectional method}

The first empirical strategy to evaluate whether the higher-order factors
can reduce the zoo factors is based on the two-pass Fama-MacBeth
procedure, motivated by the analysis in \citet{feng2020taming}. Specifically,
for each zoo factor, we estimate its SDF loading while controlling
for our higher-order factor model, which includes the FF5M and the
selected higher-order factors. 

Figure \ref{fig:Cull-Factor-Zoo-StandardFMB} summarizes the results.
Panel A plots the distribution of the absolute values of the t-statistics
for the intercept from the cross-sectional regressions. Panel B displays
the distribution of the absolute $t$-statistics associated with the
SDF loading of the factor zoo, using \citet{newey1987simple} corrected
standard errors. In both panels, the blue histogram represents the
estimates obtained by controlling for the FF5M factors, while the
green histogram represents the estimates obtained by additionally
controlling for the higher-order factors selected by the forward selection FMB
procedure. Vertical dashed lines indicate the median values of each distribution
-- a vertical blue dashed-line for estimates obtained controlling
for the FF5M factors and a vertical green dashed-line for estimates
obtained controlling for the higher-order factor model. Additionally,
we also include a vertical red dashed-line which corresponds to the median
value with no controls, or the factor model including only the factor
in the factor zoo. 

Panel A shows that the intercept estimates tend to be highly statistically
significant, with a median value of 4.00, when the factor model includes
only the factor in the factor zoo. When the factor model includes
both FF5M and the factor in the factor zoo, the median $t$-statistic
decreases substantially but remains significant at the 5\% level at
2.23. When the factor model further includes the higher-order factors,
the median $t$-statistic is no longer statistically significant at
the 5\% level at 1.84. 

Panel B shows that, when the factor model includes only the factor
in the factor zoo, the median SDF loading estimate for the factor
zoo is not statistically different from zero, with a $t$-statistic
of 0.86. Note that these estimates are likely biased because of omitted
factors (see, e.g., \citet{feng2020taming}). Controlling for the
FF5M, the median $t$-statistic increases substantially to 1.38. 34\%
of zoo factors become statistically significant at the 5\% level with
a $t$-statistic greater than 2 suggesting that FF5M helps to remove
some bias in the factor estimates. Importantly, when we additionally
control for the higher-order factors, we find that for 95\% of the zoo
factors the risk price estimate is not statistically different from
zero, with the median $t$-statistic dropping from 1.38 to only 0.71.
Specifically, we find that only 7 out of 148 zoo factors remain significant
at the 5\% level. Of these seven factors, three belong to the ``low
leverage'' category of \citet{jensen2023there} (\textit{cash\_at},
\textit{rd5\_at}, and \textit{rd\_sale}). The remaining significant
factors are in the categories ``size'' (\textit{rd\_me}), ``seasonality''
(\textit{seas\_6\_10an}), ``accruals'' (\textit{cowc\_grla}) and
``quality'' (\textit{qmj}). Overall, these results reveal that the
higher-order factors selected by the forward selection FMB procedure substantially reduce the majority of the zoo factors.

\subsubsection{Time-series method}

The second empirical strategy we employ to evaluate whether the higher-order factors can capture the factor zoo is based on time-series regressions,
where we study whether the zoo factors can be spanned by the FF5M
and the selected higher-order factors.

\begin{figure}[H]
  \centering
  \caption{Reducing the Factor Zoo in the Cross-Section\label{fig:Cull-Factor-Zoo-StandardFMB}}

  \medskip

  \begin{minipage}{\textwidth}
    \small\singlespacing
    This figure plots the distribution of absolute t-statistics associated with the risk prices ($t(\lambda)$) and the intercepts ($t(\alpha)$) of the factor zoo from \citet{jensen2023there} based on US equity data. For each zoo factor, we use the FMB procedure to estimate its risk price controlling for the FF5M (blue histogram) and the FF5M with the higher-order factors selected by the FS-FMB procedure (green histogram). The dashed lines in blue and green correspond to the median absolute t-statistic for the FF5M and for the FF5M with higher-order factor models, respectively. The dashed red line corresponds to the median absolute t-statistic with no controls. The t-statistics are based on Newey-West corrected standard errors.
  \end{minipage}

  \medskip

  \subfloat[\textbf{Distribution of $|t(\alpha)|$}]{
    \includegraphics[width=0.45\textwidth]{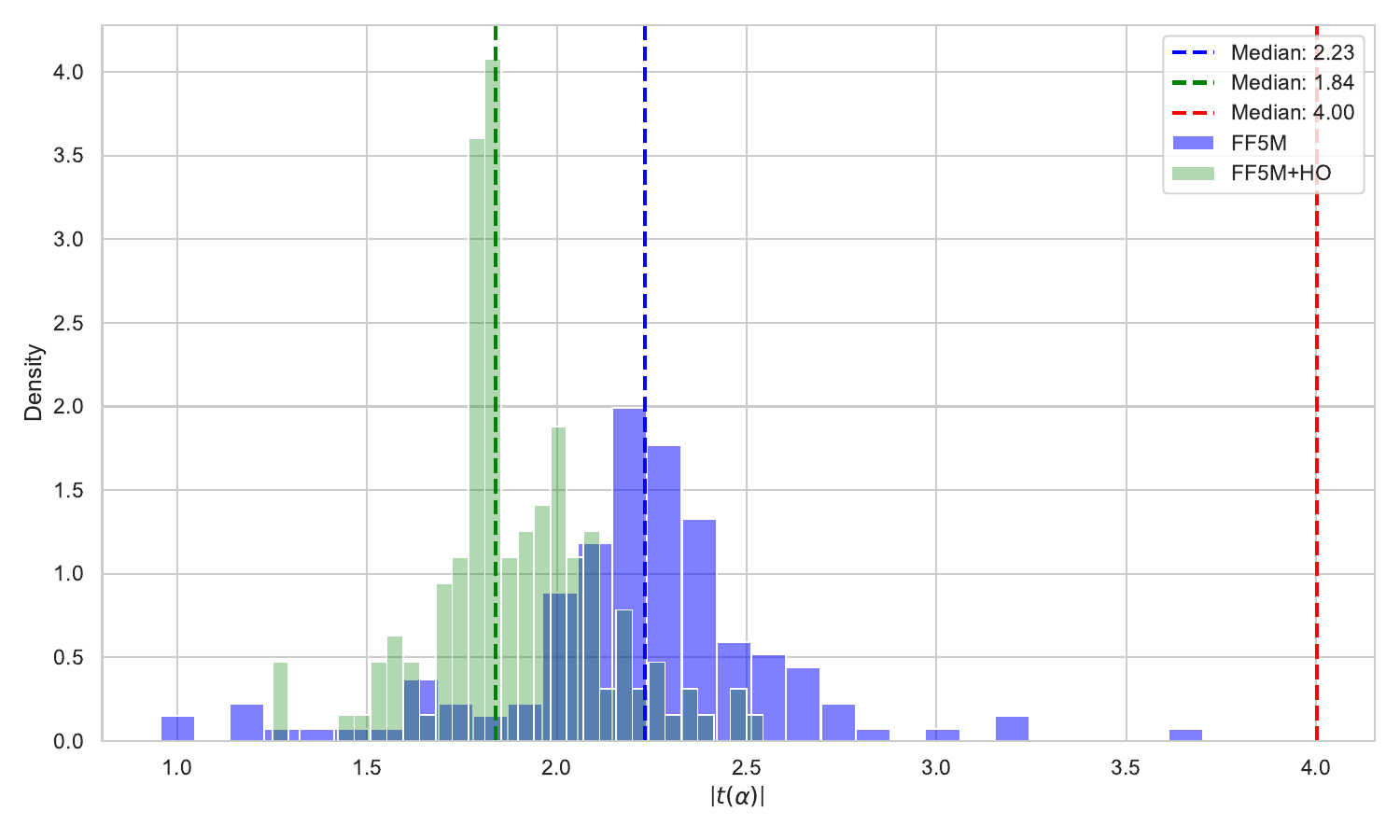}
  }
  \hfill
  \subfloat[\textbf{Distribution of $|t(\lambda)|$}]{
    \includegraphics[width=0.45\textwidth]{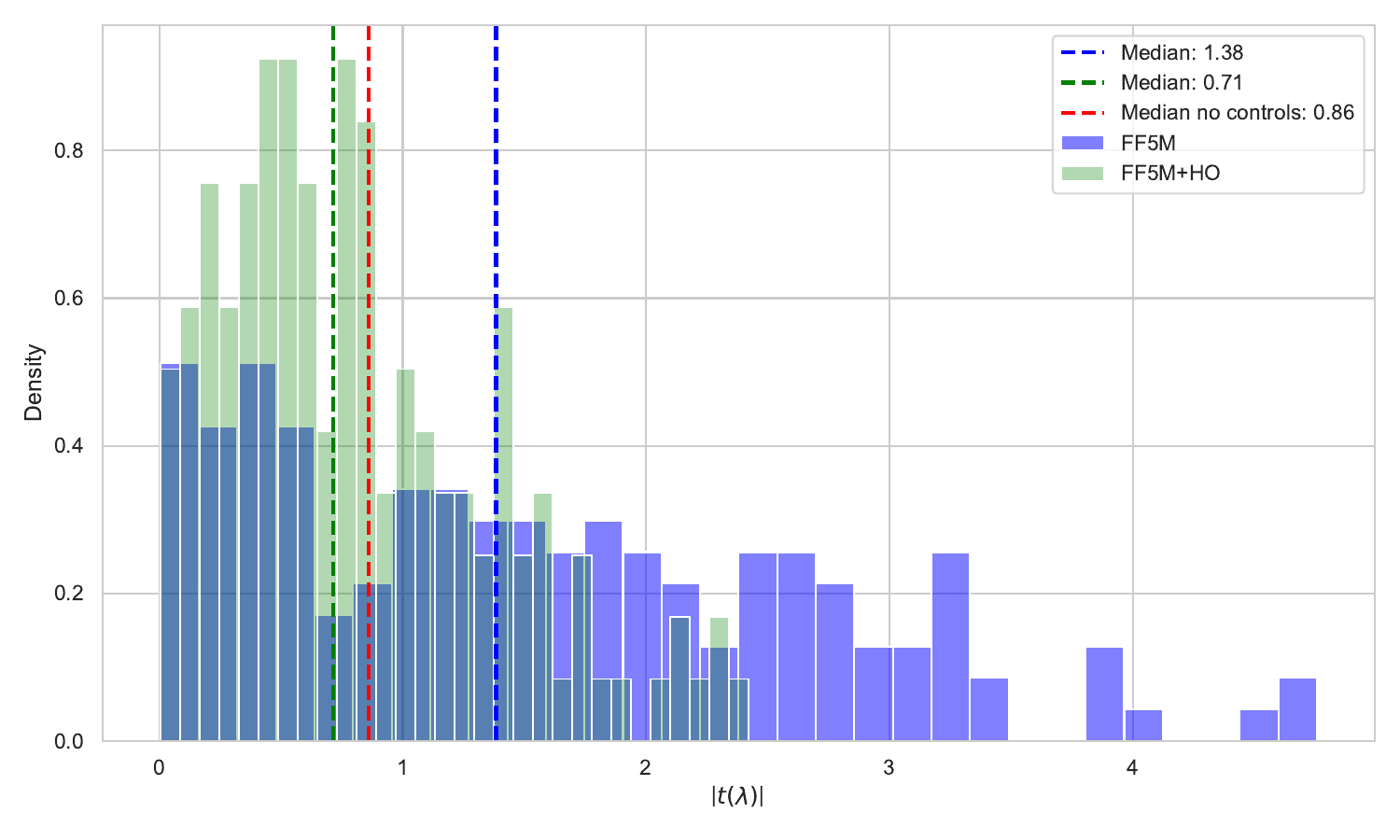}
  }
\end{figure}

Note that although the FF5M factors are tradable, the higher-order factors
are generally not directly tradable. Therefore, in order to estimate
time-series spanning regressions, we must first construct factor mimicking
portfolios of the higher-order factors based on tradable assets. In
order to do so, we project each higher-order factor onto the entire
factor zoo and obtain its factor-mimicking portfolio as the regression's
predicted values. These factor-mimicking portfolios only capture the
information in the factor zoo that is related to the higher-order
factors. The benefit of the factor-mimicking portfolio approach is
that we are able to convert the non-tradable higher-order factors
to tradable factors, but we note that the factor-mimicking portfolios
of the higher-order factors inevitably lose information from the
original factors. Figure \ref{fig:FactorMimicking_R2} shows the adjusted
$R^{2}$ of projecting each selected higher-order factor onto the
factor zoo. Mkt-RF{*}SMB has the lowest adjusted $R^{2}$ which is
less than 40\% and SMB2{*}Mom has the highest adjusted $R^{2}$ at
more than 80\%, suggesting the projections capture important fractions,
though to various degrees, of variations in the higher-order factors.
Nevertheless, we use them to check if they capture pricing information
of each factor. 

To do so, we estimate a time-series regression for each factor in
the factor zoo on the FF5M and the factor-mimicking portfolios of
the higher-order factors. The object of interest is the intercept. If
the higher-order factors indeed span a given factor in the factor
zoo, we expect to find a small and statistically insignificant intercept. 

Figure \ref{fig:Factor-zoo-TS-test-tradableHO} summarizes the results
of the time-series spanning regressions. Panel A plots the distribution
of the absolute $t$-statistics associated with the intercepts, and
Panel B shows the distribution of the absolute magnitudes for the
intercepts. The latter are annualized and expressed in percentage
points. 

Panel A shows that the median of the absolute $t$-statistic of the
average returns for the zoo factors is statistically significant at
the 5\% level at 2.06. Controlling for the FF5M brings the median
absolute $t$-statistic from 2.06 to 1.87, a decline of about 9\%.
Controlling for the higher-order factors brings the median absolute
$t$-statistic further down to 1.49, a decline of about 21\% from
1.87. Panel B reveals a similar pattern. The median of the absolute
intercept estimate of the zoo factors drops from 2.92 pp to 1.69 pp
controlling for the FF5M model. It further decreases to 1.34 pp when
the higher-order factors are included.

\subsubsection{Relationship between factor zoo and higher-order factors}

In the previous subsections, we show that the higher-order factor
model selected by the FS-FMB procedure is successful in accounting for a substantial amount of the zoo factors. Mechanically, the ability of the selected
higher-order factor model to capture the factor zoo has to come from
correlation between the zoo factors and some or all of the factors
in the higher-order factor model. In this subsection, we provide information
on the loadings of the zoo factors on the factors of the selected
higher-order factor model, which sheds light on the relative importance
of the higher-order factors in accounting for the zoo factors.

\begin{figure}[H]
  \centering
  \caption{Reducing the Factor Zoo in the Time-Series\label{fig:Factor-zoo-TS-test-tradableHO}}

  \medskip

  \begin{minipage}{\textwidth}
    \small\singlespacing
    This figure presents the results of factor spanning time-series regressions for each zoo factor controlling for the FF5M factors (blue histogram) and the FF5M factors with the mimicking portfolios corresponding to the higher-order factors selected by the forward selection FMB procedure (green histogram). Subplot (a) plots the distribution of the absolute t-statistics for the intercept ($\alpha$). Subplot (b) plots the distribution of the absolute $\alpha$s. We denote with dashed vertical lines the median value when controlling for the FF5M factors (blue dashed line), for the FF5M with the mimicking portfolios corresponding to the higher-order factors selected by the FS-FMB procedure (green dashed line), and when we include no controls (red dashed line). The factors are 148 factors from the factor zoo constructed in \citet{jensen2023there}. The t-statistic is based on Newey-West corrected standard errors.
  \end{minipage}

  \medskip

  \subfloat[\textbf{Distribution of $|t(\alpha)|$}]{
    \includegraphics[width=0.45\textwidth]{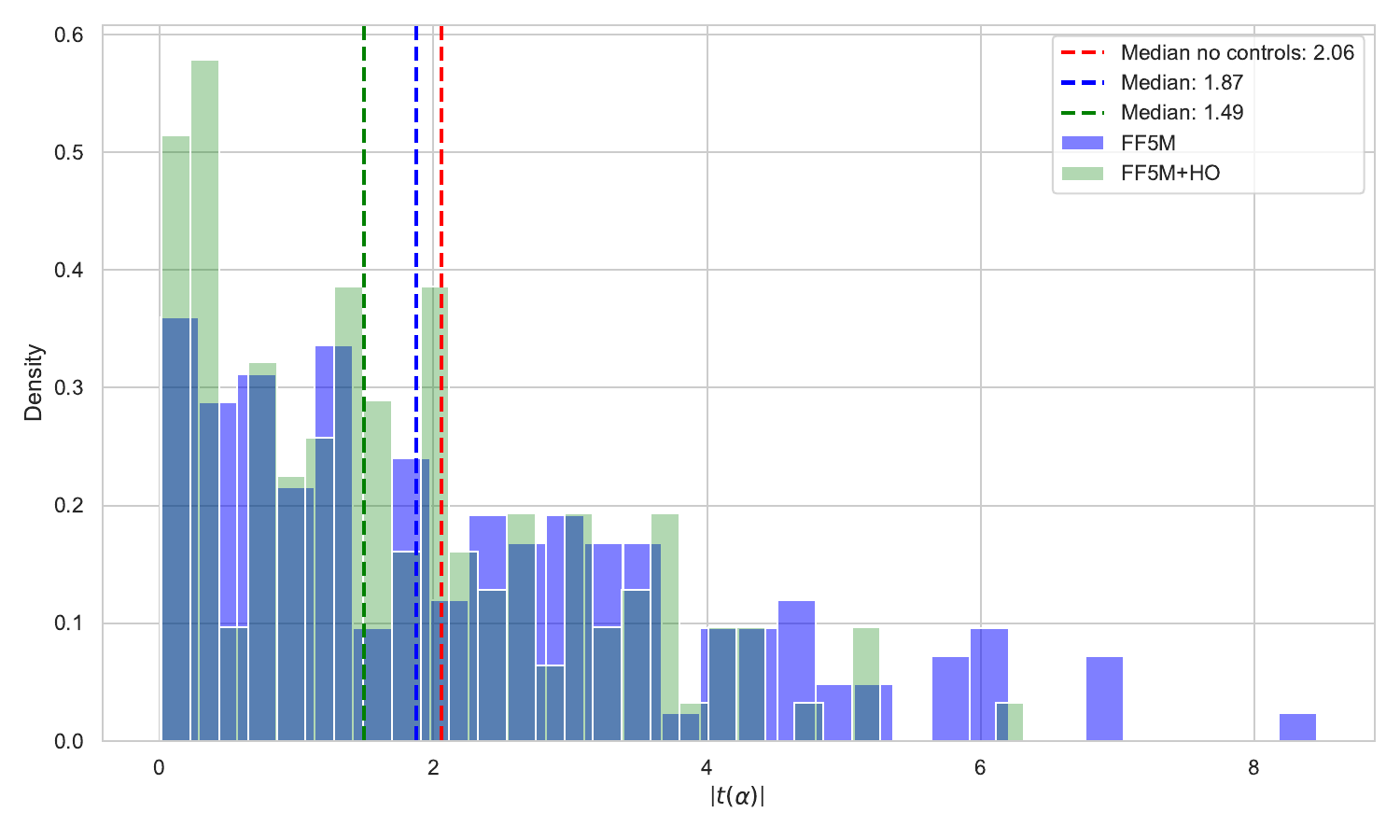}
  }
  \hfill
  \subfloat[\textbf{Distribution of $|\alpha|$}]{
    \includegraphics[width=0.45\textwidth]{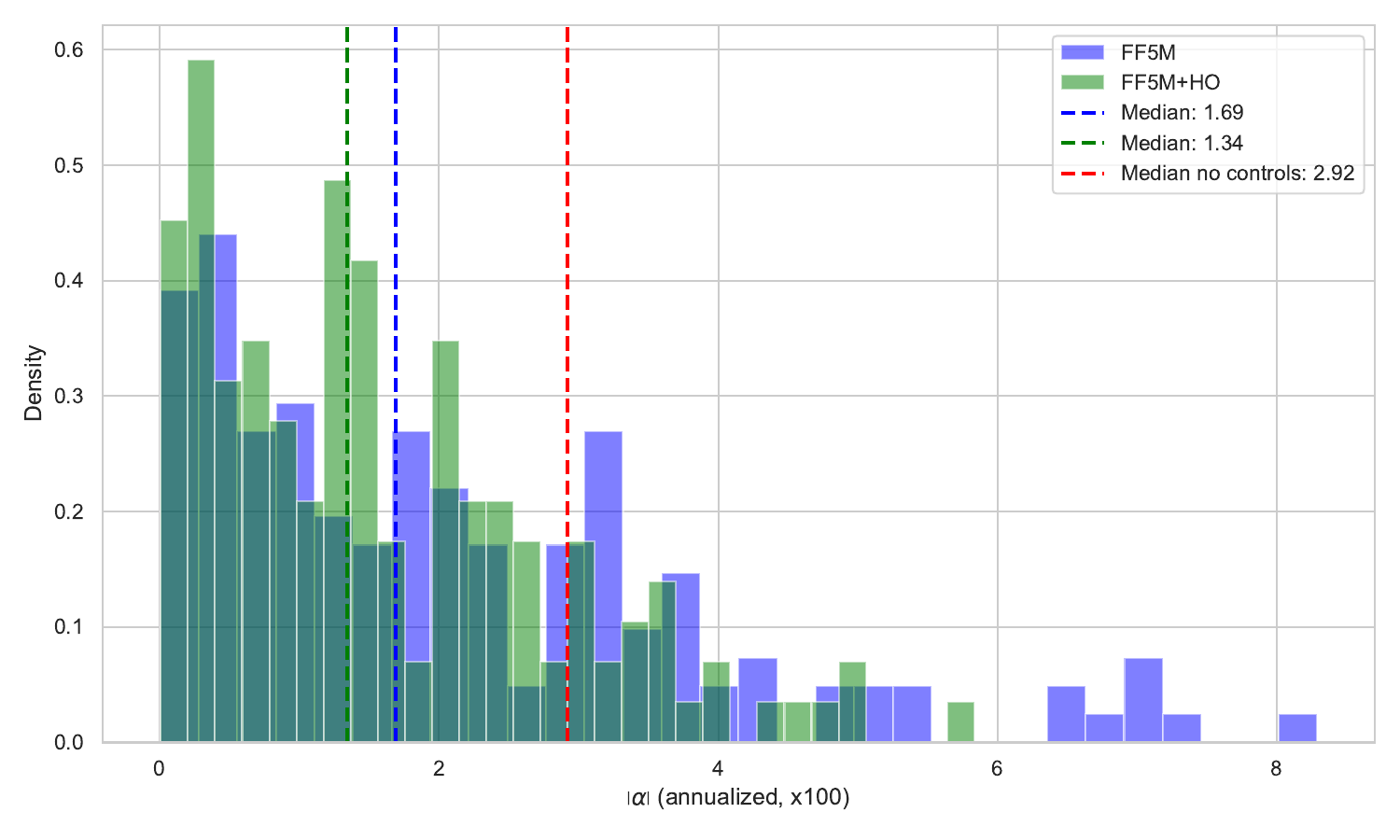}
  }
\end{figure}

Table \ref{tab:Significant-Loadings} documents the fraction of zoo factors with loading significantly different from zero at the 5\% confidence level on each of the FF5M factors and the higher-order factors selected by the forward selection FMB procedure. Column (1) shows the fraction of significance for the factors in FF5M model only and Column (2) reports the fraction of significance for the factors in the full higher-order factor model. Column (1) shows that the fractions of significance range from 45.9\% for Mkt-RF to 75\% for RMW for the 6 factors in FF5M. Column (2) shows that the fractions of significance remain stable for the 6 factors in FF5M, when we further include the selected higher-order factors. The fractions of significance for the higher-order factors range from 27\% for Mkt-RF2 to 55.4\% for HML2{*}Mkt-RF, indicating that it is not the case that a specific higher-order factor accounts for all the pricing results.
\begin{table}[H]
\caption{Loadings of the Zoo Factors on Higher-order Factor Model\label{tab:Significant-Loadings}}
\medskip{}

\begin{minipage}{\textwidth}
\small\singlespacing
This table reports the fraction of zoo factors with loading significantly different from zero at the 5\% confidence level on each of the FF5M factors and higher-order factors selected by the FS-FMB procedure. Significance is evaluated by t-statistics based on Newey-West corrected standard errors. The zoo factors are 148 factors constructed in \citet{jensen2023there}.
\end{minipage}
\medskip{}

\centering{}%
\begin{tabular}{lcc}
\hline 
& (1) & (2)\tabularnewline
\hline 
Factor & Frac Sig 5\% & Frac Sig 5\%\tabularnewline
\hline 
Mkt-RF & 0.459 & 0.459\tabularnewline
SMB & 0.568 & 0.628\tabularnewline
HML & 0.642 & 0.709\tabularnewline
RMW & 0.75 & 0.669\tabularnewline
CMA & 0.601 & 0.595\tabularnewline
Mom & 0.486 & 0.405\tabularnewline
SMB2 & & 0.446\tabularnewline
SMB2{*}Mom & & 0.331\tabularnewline
Mom2{*}RMW & & 0.338\tabularnewline
Mkt-RF2 & & 0.27\tabularnewline
Mkt-RF2{*}RMW & & 0.358\tabularnewline
Mkt-RF{*}SMB & & 0.304\tabularnewline
HML2{*}Mkt-RF & & 0.554\tabularnewline
\hline 
\# zoo factors & 148 & 148\tabularnewline
\hline 
\end{tabular}
\end{table}

\subsection{Higher-Order Factors and Macro Risk\label{sec:Higher-Order-Factors-and}}

In this subsection, we conduct analyses to better understand the economic
forces the higher-order factors capture. In the first exercise, we
examine the correlations between the selected higher-order factors
and two important macroeconomic risks as proxied by the U.S. equity
market returns and the CBOE VIX index. 

In Table \ref{tab:correlation-macro}, we present the results. Panel
A documents the correlations with U.S. equity market returns and Panel
B shows the correlations with CBOE VIX index. Beyond correlations
over the full sample, we also calculate the correlations during the
NBER recession periods and the correlations corresponding to the bottom
10\%, the middle 10\%-90\%, and the top 10\% of higher-order factors'
return distribution.

\begin{table}[]
\caption{Correlations with Macroeconomic Risk\label{tab:correlation-macro}}

\medskip{}

\begin{minipage}{\textwidth}
\small\singlespacing
This table reports the correlation coefficients for each higher-order
factor selected by the forward selection FMB procedure with the U.S. equity
market return (Panel A) and with the CBOE VIX index (Panel B). We present both
unconditional (Full) correlation coefficients and conditional correlation
coefficients for specific subsamples: one containing U.S. NBER recession
periods, and others corresponding to the bottom and top 10\% of each
factor\textquoteright s return distribution, as well as the remaining
10\%--90\% quantile. Correlation coefficients involving the VIX index
are estimated starting in January 1990, reflecting the index\textquoteright s
availability from that date onward.
\end{minipage}

\medskip{}

\centering{}%
\begin{tabular}{lccccc}
\hline 
Factor $(h_{j})$/Sample & Full & NBER Recessions & $q_{\rightarrow0.1}(h_{j})$ & $q_{0.1\rightarrow0.9}(h_{j})$ & $q_{0.9\rightarrow}(h_{j})$\\
\hline 
\multicolumn{6}{c}{\textbf{Panel A: Correlation with US Equity Market}}\\
\hline 
SMB2 & 0.016 & 0.205 & -0.062 & -0.01 & 0.1\\
SMB2{*}Mom & -0.032 & -0.345 & -0.119 & -0.018 & 0.085\\
Mom2{*}RMW & -0.113 & -0.109 & -0.085 & -0.077 & -0.001\\
Mkt-RF2 & -0.164 & -0.108 & 0.388 & 0.199 & -0.267\\
Mkt-RF2{*}RMW & -0.372 & -0.287 & 0.118 & 0.027 & -0.479\\
Mkt-RF{*}SMB & -0.206 & -0.136 & -0.293 & -0.009 & -0.299\\
HML2{*}Mkt-RF & 0.524 & 0.636 & 0.357 & 0.64 & 0.607\\
\hline 
\multicolumn{6}{c}{\textbf{Panel B: Correlation with VIX Index}}\\
\hline 
SMB2 & 0.134 & 0.03 & 0.076 & 0.189 & 0.005\\
SMB2{*}Mom & 0.001 & -0.07 & -0.233 & 0.043 & -0.053\\
Mom2{*}RMW & 0.015 & 0.116 & -0.07 & 0.195 & 0.236\\
Mkt-RF2 & 0.603 & 0.697 & 0.141 & 0.284 & 0.515\\
Mkt-RF2{*}RMW & 0.32 & 0.515 & -0.255 & 0.138 & 0.376\\
Mkt-RF{*}SMB & 0.257 & 0.486 & -0.298 & 0.197 & 0.438\\
HML2{*}Mkt-RF & -0.239 & -0.437 & -0.266 & -0.334 & 0.104\\
\hline 
\end{tabular}
\end{table}

For the first two higher-order factors (i.e., SMB2 and SMB2{*}Mom),  the magnitude of the correlations with U.S. equity market returns
is markedly higher during the NBER recession periods relative to the full sample. For example, the correlation between SMB2 and U.S.  equity market returns is only 0.016 for the full sample, but it jumps to 0.205 during the NBER recession periods. For the last four higher-order
factors (i.e., Mkt-RF2, Mkt-RF2{*}RMW, Mkt-RF{*}SMB, and HML2{*}Mkt-RF), they tend to be highly correlated with the VIX index,
especially during the NBER recession periods. For example, the correlation between Mkt-RF2{*}RMW and VIX is 0.32 for the full sample and increases
to 0.515 during the NBER recession periods. These results show that the selected higher-order factors capture different macroeconomic risks, especially during economic downturns. 

The results in Table \ref{tab:correlation-macro} indicate that the higher-order terms---the powers and interactions---may systematically capture the state-dependent macroeconomic tail risks. To formalize this idea, in Panel A of  Table \ref{tab:Exposure-to-Macroeconomic-Combined}, we first evaluate the unconditional exposures of the selected higher-order factors to a broad set of common macroeconomic variables using OLS regressions. The results show that these macroeconomic variables are able to account for a meaningful fraction of the variation in the selected higher-order factors. The adjusted $R^{2}$ values range from 3\% to 45\%, exceeding 20\% for five of the seven higher-order factors. Among these macroeconomic variables, measures related to aggregate uncertainty, in particular financial uncertainty, are highly associated with these higher-order factors.

While the baseline unconditional regressions establish a link to macroeconomic risk, if the higher-order terms capture tail dependence, we would expect the relationship to be more pronounced in the tails of the distribution. To better capture the potential state dependence, we estimate quantile regressions of the higher-order factors on the same set of macroeconomic variables. In Panel B of Table \ref{tab:Exposure-to-Macroeconomic-Combined}, we report the conditional loadings for the bottom decile ($\tau=0.10$) and the top decile ($\tau=0.90$) of the higher-order factors on the macroeconomic variables. The extreme state dependence of these higher-order factors is most evident from the pseudo-$R^2$ of these quantile regressions. The explanatory power based on the pseudo-$R^2$ values for the tail regressions in Panel B of Table \ref{tab:Exposure-to-Macroeconomic-Combined} frequently exceeds 50\%. Moreover, the coefficient estimates on financial uncertainty and measures of intermediary capital constraints become significantly more pronounced during these tail realizations.

These results suggest that the higher-order factors may capture non-linear exposure to macroeconomic tail risks, particularly those associated with financial intermediary uncertainty and distress. This interpretation is consistent with intermediary asset pricing models (e.g., \citealp{he2013intermediary,brunnermeier2014macroeconomic}), in which asset prices are disproportionately sensitive to intermediary stress during periods of extreme macroeconomic environments. The fact that the relationship between higher-order factors and macroeconomic variables, in particular financial intermediary constraints, strengthens precisely in extreme macroeconomic environments provides an economic rationale for why these non-linear terms of common factors are priced and why they may account for the broader factor zoo.


\begin{sidewaystable}
\caption{\scriptsize Exposure to Macroeconomic Factors\label{tab:Exposure-to-Macroeconomic-Combined}}
\medskip{}
\begin{minipage}{\textwidth}
\scriptsize\singlespacing
This table reports results of regressing each selected higher-order factor on a list of
common macroeconomic factors. All regressions also control for the market excess return
(Mkt-RF). Panel~A (this page) reports OLS estimates; significance is based on
Newey-West corrected standard errors. Panel~B (next page) reports quantile regressions
at $\tau=0.10$ (bottom decile) and $\tau=0.90$ (top decile) on the overlapping sample;
standard errors use the Powell~(1991) sandwich estimator with Epanechnikov kernel and
Hall-Sheather bandwidth.
{*}{*}{*}, {*}{*}, and {*} denote significance at the 1\%, 5\%, and 10\% levels,
respectively.
\end{minipage}
\medskip{}
\setlength{\tabcolsep}{3pt} 
\adjustbox{max width=\textwidth}{
\centering{}{\scriptsize{}%
\begin{tabular}{lllccccccc}
\hline 
& & & {\scriptsize (1)} & {\scriptsize (2)} & {\scriptsize (3)} & {\scriptsize (4)} & {\scriptsize (5)} & {\scriptsize (6)} & {\scriptsize (7)} \tabularnewline
{\scriptsize Factor} & {\scriptsize Reference} & {\scriptsize Category} & {\scriptsize SMB2} & {\scriptsize SMB2{*}Mom} & {\scriptsize Mom2{*}RMW} & {\scriptsize Mkt-RF2} & {\scriptsize Mkt-RF2{*}RMW} & {\scriptsize Mkt-RF{*}SMB} & {\scriptsize HML2{*}Mkt-RF} \tabularnewline
\hline 
\multicolumn{9}{c}{Panel A: OLS Regressions} & \tabularnewline
\hline 
{\scriptsize Financial uncertainty} & {\scriptsize \citet{bali2017economic}} & {\scriptsize Vol-Uncert} & {\scriptsize 1.37{*}{*}} & {\scriptsize 0.41} & {\scriptsize -0.95} & {\scriptsize 0.75{*}{*}{*}} & {\scriptsize 0.18{*}{*}} & {\scriptsize 0.28{*}{*}} & {\scriptsize 0.12} \tabularnewline
{\scriptsize Policy uncertainty} & {\scriptsize \citet{baker2016measuring}} & {\scriptsize Vol-Uncert} & {\scriptsize 0.11} & {\scriptsize -0.12} & {\scriptsize 0.09} & {\scriptsize -0.07} & {\scriptsize 0.07{*}{*}} & {\scriptsize 0.07} & {\scriptsize -0.00} \tabularnewline
{\scriptsize Geopolitical uncertainty} & {\scriptsize \citet{caldara2022measuring}} & {\scriptsize Vol-Uncert} & {\scriptsize -0.16} & {\scriptsize -0.33} & {\scriptsize 0.41} & {\scriptsize 0.09} & {\scriptsize 0.05{*}} & {\scriptsize 0.01} & {\scriptsize -0.05} \tabularnewline
{\scriptsize News vol} & {\scriptsize \citet{manela2017news}} & {\scriptsize Vol-Uncert} & {\scriptsize -0.02{*}{*}} & {\scriptsize -0.01} & {\scriptsize 0.01} & {\scriptsize 0.01{*}{*}} & {\scriptsize -0.00} & {\scriptsize -0.00} & {\scriptsize -0.01{*}{*}} \tabularnewline
{\scriptsize Common idio-vol shock} & {\scriptsize \citet{herskovic2016common}} & {\scriptsize Vol-Uncert} & {\scriptsize 0.28} & {\scriptsize 0.82{*}{*}{*}} & {\scriptsize 0.27} & {\scriptsize 0.42} & {\scriptsize -0.18} & {\scriptsize 0.17} & {\scriptsize 0.41{*}{*}} \tabularnewline
{\scriptsize Real uncertainty} & {\scriptsize \citet{bali2017economic}} & {\scriptsize Vol-Uncert} & {\scriptsize 0.32} & {\scriptsize -1.51} & {\scriptsize 0.34} & {\scriptsize 1.66{*}{*}{*}} & {\scriptsize 0.14} & {\scriptsize 1.02} & {\scriptsize 0.13} \tabularnewline
{\scriptsize Macro uncertainty} & {\scriptsize \citet{bali2017economic}} & {\scriptsize Vol-Uncert} & {\scriptsize -1.02{*}} & {\scriptsize -0.20} & {\scriptsize 0.85} & {\scriptsize -0.37} & {\scriptsize -0.23} & {\scriptsize -0.27} & {\scriptsize -0.31} \tabularnewline
{\scriptsize Market vol} & {\scriptsize \citet{martin2019expected}} & {\scriptsize Vol-Uncert} & {\scriptsize -18.82{*}} & {\scriptsize -21.97{*}} & {\scriptsize 13.34} & {\scriptsize 2.30} & {\scriptsize -2.60{*}} & {\scriptsize 3.50} & {\scriptsize -1.45} \tabularnewline
{\scriptsize Labor income growth} & {\scriptsize \citet{dittmar2002nonlinear}} & {\scriptsize Macro} & {\scriptsize -72.50} & {\scriptsize 415.60{*}} & {\scriptsize -246.40} & {\scriptsize 212.87} & {\scriptsize -122.95{*}{*}} & {\scriptsize -64.08} & {\scriptsize -92.05} \tabularnewline
{\scriptsize Consum growth squared} & {\scriptsize \citet{chapman1997cyclical}} & {\scriptsize Macro} & {\scriptsize 3.75{*}} & {\scriptsize -3.29} & {\scriptsize 1.53} & {\scriptsize -1.16} & {\scriptsize -1.44{*}{*}{*}} & {\scriptsize -2.73{*}} & {\scriptsize 1.18{*}} \tabularnewline
{\scriptsize Consum growth} & {\scriptsize \citet{chapman1997cyclical}} & {\scriptsize Macro} & {\scriptsize 1.12} & {\scriptsize 0.37} & {\scriptsize -0.34} & {\scriptsize -0.27} & {\scriptsize 0.02} & {\scriptsize -0.04} & {\scriptsize 0.03} \tabularnewline
{\scriptsize Labor income growth} & {\scriptsize \citet{campbell1996understanding}} & {\scriptsize Macro} & {\scriptsize 149.00} & {\scriptsize -829.09{*}} & {\scriptsize 489.16} & {\scriptsize -423.63} & {\scriptsize 246.12{*}{*}} & {\scriptsize 128.21} & {\scriptsize 186.49} \tabularnewline
{\scriptsize Change long-term govt yield} & {\scriptsize \citet{sweeney1986pricing}} & {\scriptsize Macro} & {\scriptsize -6.03} & {\scriptsize -25.61{*}{*}} & {\scriptsize 15.40} & {\scriptsize 8.59} & {\scriptsize -2.87} & {\scriptsize 9.66{*}{*}} & {\scriptsize 3.58} \tabularnewline
{\scriptsize Intermediary capital risk factor} & {\scriptsize \citet{he2017intermediary}} & {\scriptsize Inter} & {\scriptsize -0.85{*}} & {\scriptsize 0.99} & {\scriptsize 1.62} & {\scriptsize -0.16} & {\scriptsize -0.06} & {\scriptsize -0.57} & {\scriptsize -1.58{*}{*}} \tabularnewline
{\scriptsize Intermediary capital return factor} & {\scriptsize \citet{he2017intermediary}} & {\scriptsize Inter} & {\scriptsize 1.40} & {\scriptsize -5.41{*}{*}} & {\scriptsize 1.99} & {\scriptsize 2.02{*}} & {\scriptsize 0.34} & {\scriptsize 1.44} & {\scriptsize 1.98{*}{*}{*}} \tabularnewline
{\scriptsize Amihud liquidity} & {\scriptsize \citet{chen2018micro}} & {\scriptsize Liquidity} & {\scriptsize 0.09{*}} & {\scriptsize -0.09} & {\scriptsize -0.05} & {\scriptsize -0.01} & {\scriptsize -0.03} & {\scriptsize 0.03} & {\scriptsize 0.06{*}} \tabularnewline
{\scriptsize Liquidity} & {\scriptsize \citet{pastor2003liquidity}} & {\scriptsize Liquidity} & {\scriptsize -0.04} & {\scriptsize -0.39} & {\scriptsize 0.18} & {\scriptsize -0.57} & {\scriptsize 0.08} & {\scriptsize -0.84{*}{*}} & {\scriptsize 0.03} \tabularnewline
{\scriptsize Roll liquidity} & {\scriptsize \citet{chen2018micro}} & {\scriptsize Liquidity} & {\scriptsize -0.05} & {\scriptsize 0.41} & {\scriptsize -0.34} & {\scriptsize 0.08} & {\scriptsize 0.03} & {\scriptsize -0.07} & {\scriptsize -0.11{*}} \tabularnewline
{\scriptsize Investor sentiment (orth macro)} & {\scriptsize \citet{baker2006investor}} & {\scriptsize Sentiment} & {\scriptsize 0.02} & {\scriptsize -0.03} & {\scriptsize 0.00} & {\scriptsize 0.10{*}{*}} & {\scriptsize 0.03{*}{*}} & {\scriptsize 0.03} & {\scriptsize -0.04} \tabularnewline
{\scriptsize Investor sentiment (orth macro)} & {\scriptsize \citet{huang2015investor}} & {\scriptsize Sentiment} & {\scriptsize -0.01} & {\scriptsize 0.06} & {\scriptsize -0.06} & {\scriptsize -0.08{*}{*}} & {\scriptsize -0.00} & {\scriptsize -0.05{*}} & {\scriptsize -0.01} \tabularnewline
& & {\scriptsize Adj. $R^{2}$} & {\scriptsize 0.22} & {\scriptsize 0.10} & {\scriptsize 0.03} & {\scriptsize 0.39} & {\scriptsize 0.23} & {\scriptsize 0.26} & {\scriptsize 0.45} \tabularnewline
& & &  &  &  &  &  &  &  \tabularnewline
\hline 
\end{tabular}}}{\scriptsize\par}
\end{sidewaystable}


\begin{sidewaystable}
\ContinuedFloat
\caption[]{\scriptsize (Continued)}
\medskip{}
\setlength{\tabcolsep}{3pt} 
\adjustbox{max width=\textwidth}{
\centering{}{\scriptsize{}%
\begin{tabular}{lllccccccc}
\hline 
& & & {\scriptsize (1)} & {\scriptsize (2)} & {\scriptsize (3)} & {\scriptsize (4)} & {\scriptsize (5)} & {\scriptsize (6)} & {\scriptsize (7)} \tabularnewline
{\scriptsize Factor} & {\scriptsize Reference} & {\scriptsize Category} & {\scriptsize SMB2} & {\scriptsize SMB2{*}Mom} & {\scriptsize Mom2{*}RMW} & {\scriptsize Mkt-RF2} & {\scriptsize Mkt-RF2{*}RMW} & {\scriptsize Mkt-RF{*}SMB} & {\scriptsize HML2{*}Mkt-RF} \tabularnewline
\hline 
\multicolumn{9}{c}{Panel B: Quantile Regressions} & \tabularnewline
\hline 
\multicolumn{10}{c}{Bottom Decile ($\tau=0.10$)} \tabularnewline
\hline 
{\scriptsize Financial uncertainty} & {\scriptsize \citet{bali2017economic}} & {\scriptsize Vol-Uncert} & {\scriptsize 0.05} & {\scriptsize -2.01{*}{*}} & {\scriptsize -8.02{*}{*}{*}} & {\scriptsize 0.15} & {\scriptsize 0.24} & {\scriptsize -1.40{*}{*}{*}} & {\scriptsize 0.07} \tabularnewline
{\scriptsize Policy uncertainty} & {\scriptsize \citet{baker2016measuring}} & {\scriptsize Vol-Uncert} & {\scriptsize 0.04} & {\scriptsize -0.40} & {\scriptsize 0.73} & {\scriptsize 0.06} & {\scriptsize 0.09} & {\scriptsize 0.03} & {\scriptsize 0.05} \tabularnewline
{\scriptsize Geopolitical uncertainty} & {\scriptsize \citet{caldara2022measuring}} & {\scriptsize Vol-Uncert} & {\scriptsize -0.02} & {\scriptsize 0.14} & {\scriptsize 0.13} & {\scriptsize -0.18{*}} & {\scriptsize -0.01} & {\scriptsize 0.19} & {\scriptsize 0.10} \tabularnewline
{\scriptsize News vol} & {\scriptsize \citet{manela2017news}} & {\scriptsize Vol-Uncert} & {\scriptsize -0.00} & {\scriptsize 0.00} & {\scriptsize 0.05} & {\scriptsize -0.01} & {\scriptsize 0.01{*}} & {\scriptsize 0.00} & {\scriptsize -0.00} \tabularnewline
{\scriptsize Common idio-vol shock} & {\scriptsize \citet{herskovic2016common}} & {\scriptsize Vol-Uncert} & {\scriptsize -0.05} & {\scriptsize 0.62} & {\scriptsize 2.65} & {\scriptsize 0.21} & {\scriptsize -0.20{*}{*}} & {\scriptsize 0.14} & {\scriptsize 1.13{*}{*}{*}} \tabularnewline
{\scriptsize Real uncertainty} & {\scriptsize \citet{bali2017economic}} & {\scriptsize Vol-Uncert} & {\scriptsize 0.30} & {\scriptsize 0.61} & {\scriptsize -4.57} & {\scriptsize 0.76} & {\scriptsize 0.26} & {\scriptsize 0.08} & {\scriptsize 0.85} \tabularnewline
{\scriptsize Macro uncertainty} & {\scriptsize \citet{bali2017economic}} & {\scriptsize Vol-Uncert} & {\scriptsize -0.18} & {\scriptsize -0.42} & {\scriptsize 4.31} & {\scriptsize -0.19} & {\scriptsize -0.07} & {\scriptsize 0.67} & {\scriptsize -1.34} \tabularnewline
{\scriptsize Market vol} & {\scriptsize \citet{martin2019expected}} & {\scriptsize Vol-Uncert} & {\scriptsize 0.09} & {\scriptsize 12.52} & {\scriptsize 47.42} & {\scriptsize 4.57} & {\scriptsize -4.31{*}{*}} & {\scriptsize 6.90} & {\scriptsize -0.23} \tabularnewline
{\scriptsize Labor income growth} & {\scriptsize \citet{dittmar2002nonlinear}} & {\scriptsize Macro} & {\scriptsize 0.12} & {\scriptsize 2.22} & {\scriptsize -17.27} & {\scriptsize -1.24} & {\scriptsize -2.39} & {\scriptsize 1.75} & {\scriptsize -4.96} \tabularnewline
{\scriptsize Consum growth squared} & {\scriptsize \citet{chapman1997cyclical}} & {\scriptsize Macro} & {\scriptsize 1.46{*}{*}} & {\scriptsize -13.64{*}{*}{*}} & {\scriptsize 15.74} & {\scriptsize 0.34} & {\scriptsize -0.03} & {\scriptsize -5.25{*}{*}{*}} & {\scriptsize 2.62} \tabularnewline
{\scriptsize Consum growth} & {\scriptsize \citet{chapman1997cyclical}} & {\scriptsize Macro} & {\scriptsize 0.01} & {\scriptsize -0.60} & {\scriptsize -2.36} & {\scriptsize -0.31} & {\scriptsize 0.46{*}{*}} & {\scriptsize -0.23} & {\scriptsize -0.55} \tabularnewline
{\scriptsize Labor income growth} & {\scriptsize \citet{campbell1996understanding}} & {\scriptsize Macro} & {\scriptsize -0.03} & {\scriptsize 0.20} & {\scriptsize -0.66} & {\scriptsize -0.16} & {\scriptsize -0.10} & {\scriptsize 0.19} & {\scriptsize 0.17} \tabularnewline
{\scriptsize Change long-term govt yield} & {\scriptsize \citet{sweeney1986pricing}} & {\scriptsize Macro} & {\scriptsize -1.35} & {\scriptsize -10.70} & {\scriptsize 69.61} & {\scriptsize -3.41} & {\scriptsize -1.31} & {\scriptsize 13.19} & {\scriptsize 24.87{*}{*}} \tabularnewline
{\scriptsize Intermediary capital risk factor} & {\scriptsize \citet{he2017intermediary}} & {\scriptsize Inter} & {\scriptsize 0.09} & {\scriptsize 0.82} & {\scriptsize -9.53{*}{*}} & {\scriptsize -0.17} & {\scriptsize -0.29} & {\scriptsize -0.46} & {\scriptsize -1.33} \tabularnewline
{\scriptsize Intermediary capital return factor} & {\scriptsize \citet{he2017intermediary}} & {\scriptsize Inter} & {\scriptsize 0.12} & {\scriptsize -3.82} & {\scriptsize 24.22{*}{*}{*}} & {\scriptsize 1.70{*}{*}{*}} & {\scriptsize 1.12{*}} & {\scriptsize -0.65} & {\scriptsize 2.73{*}} \tabularnewline
{\scriptsize Amihud liquidity} & {\scriptsize \citet{chen2018micro}} & {\scriptsize Liquidity} & {\scriptsize -0.00} & {\scriptsize 0.08} & {\scriptsize 0.61} & {\scriptsize 0.02} & {\scriptsize -0.02} & {\scriptsize -0.02} & {\scriptsize -0.12} \tabularnewline
{\scriptsize Liquidity} & {\scriptsize \citet{pastor2003liquidity}} & {\scriptsize Liquidity} & {\scriptsize -0.06} & {\scriptsize -0.31} & {\scriptsize 0.91} & {\scriptsize 0.10} & {\scriptsize -0.06} & {\scriptsize -0.94{*}{*}{*}} & {\scriptsize -0.14} \tabularnewline
{\scriptsize Roll liquidity} & {\scriptsize \citet{chen2018micro}} & {\scriptsize Liquidity} & {\scriptsize -0.02} & {\scriptsize 0.29} & {\scriptsize -1.68} & {\scriptsize -0.02} & {\scriptsize -0.10} & {\scriptsize -0.02} & {\scriptsize -0.09} \tabularnewline
{\scriptsize Investor sentiment (orth macro)} & {\scriptsize \citet{baker2006investor}} & {\scriptsize Sentiment} & {\scriptsize 0.01} & {\scriptsize -0.05} & {\scriptsize -0.08} & {\scriptsize -0.01} & {\scriptsize 0.02} & {\scriptsize -0.13{*}{*}} & {\scriptsize -0.07} \tabularnewline
{\scriptsize Investor sentiment (orth macro)} & {\scriptsize \citet{huang2015investor}} & {\scriptsize Sentiment} & {\scriptsize 0.00} & {\scriptsize 0.15} & {\scriptsize 0.07} & {\scriptsize 0.02} & {\scriptsize -0.04} & {\scriptsize 0.16{*}{*}} & {\scriptsize -0.14} \tabularnewline
& & {\scriptsize Pseudo $R^{2}$} & {\scriptsize 0.07} & {\scriptsize 0.45} & {\scriptsize 0.30} & {\scriptsize 0.27} & {\scriptsize 0.50} & {\scriptsize 0.56} & {\scriptsize 0.64} \tabularnewline
\multicolumn{10}{c}{Top Decile ($\tau=0.90$)} \tabularnewline
\hline 
{\scriptsize Financial uncertainty} & {\scriptsize \citet{bali2017economic}} & {\scriptsize Vol-Uncert} & {\scriptsize 6.30{*}{*}{*}} & {\scriptsize 2.80{*}{*}} & {\scriptsize -1.41} & {\scriptsize 0.35} & {\scriptsize 0.04} & {\scriptsize 1.06{*}{*}} & {\scriptsize 0.45{*}{*}} \tabularnewline
{\scriptsize Policy uncertainty} & {\scriptsize \citet{baker2016measuring}} & {\scriptsize Vol-Uncert} & {\scriptsize 0.19} & {\scriptsize -0.21} & {\scriptsize 0.56{*}{*}} & {\scriptsize -0.13} & {\scriptsize 0.08} & {\scriptsize 0.08} & {\scriptsize 0.02} \tabularnewline
{\scriptsize Geopolitical uncertainty} & {\scriptsize \citet{caldara2022measuring}} & {\scriptsize Vol-Uncert} & {\scriptsize -0.59} & {\scriptsize -0.19} & {\scriptsize -0.05} & {\scriptsize -0.39{*}} & {\scriptsize -0.01} & {\scriptsize 0.01} & {\scriptsize -0.08} \tabularnewline
{\scriptsize News vol} & {\scriptsize \citet{manela2017news}} & {\scriptsize Vol-Uncert} & {\scriptsize -0.03} & {\scriptsize -0.03} & {\scriptsize 0.02} & {\scriptsize 0.01} & {\scriptsize 0.00} & {\scriptsize -0.01} & {\scriptsize -0.00} \tabularnewline
{\scriptsize Common idio-vol shock} & {\scriptsize \citet{herskovic2016common}} & {\scriptsize Vol-Uncert} & {\scriptsize -0.22} & {\scriptsize 1.14{*}{*}} & {\scriptsize -0.66} & {\scriptsize 0.52} & {\scriptsize -0.41{*}{*}} & {\scriptsize 0.08} & {\scriptsize -0.13} \tabularnewline
{\scriptsize Real uncertainty} & {\scriptsize \citet{bali2017economic}} & {\scriptsize Vol-Uncert} & {\scriptsize 1.85} & {\scriptsize 4.06} & {\scriptsize -2.74} & {\scriptsize -1.18} & {\scriptsize -0.88} & {\scriptsize 0.35} & {\scriptsize -0.11} \tabularnewline
{\scriptsize Macro uncertainty} & {\scriptsize \citet{bali2017economic}} & {\scriptsize Vol-Uncert} & {\scriptsize -4.59{*}} & {\scriptsize -5.22{*}{*}} & {\scriptsize 3.78{*}{*}} & {\scriptsize 3.01{*}{*}} & {\scriptsize 0.40} & {\scriptsize -0.70} & {\scriptsize -0.06} \tabularnewline
{\scriptsize Market vol} & {\scriptsize \citet{martin2019expected}} & {\scriptsize Vol-Uncert} & {\scriptsize -29.50{*}{*}{*}} & {\scriptsize -6.93} & {\scriptsize 3.49} & {\scriptsize -6.77} & {\scriptsize -0.11} & {\scriptsize -1.76} & {\scriptsize -0.96} \tabularnewline
{\scriptsize Labor income growth} & {\scriptsize \citet{dittmar2002nonlinear}} & {\scriptsize Macro} & {\scriptsize -7.56} & {\scriptsize 5.36} & {\scriptsize -0.94} & {\scriptsize -7.10} & {\scriptsize -0.29} & {\scriptsize -1.55} & {\scriptsize -0.00} \tabularnewline
{\scriptsize Consum growth squared} & {\scriptsize \citet{chapman1997cyclical}} & {\scriptsize Macro} & {\scriptsize 8.78{*}} & {\scriptsize -3.33} & {\scriptsize 4.63} & {\scriptsize -0.84} & {\scriptsize 3.46{*}{*}{*}} & {\scriptsize -3.20} & {\scriptsize 0.89} \tabularnewline
{\scriptsize Consum growth} & {\scriptsize \citet{chapman1997cyclical}} & {\scriptsize Macro} & {\scriptsize 0.51} & {\scriptsize 1.26} & {\scriptsize 0.74} & {\scriptsize 0.26} & {\scriptsize 0.01} & {\scriptsize 0.00} & {\scriptsize 0.32{*}} \tabularnewline
{\scriptsize Labor income growth} & {\scriptsize \citet{campbell1996understanding}} & {\scriptsize Macro} & {\scriptsize -0.38} & {\scriptsize 0.71} & {\scriptsize -0.16} & {\scriptsize -1.42} & {\scriptsize 0.07} & {\scriptsize 0.35} & {\scriptsize -0.11} \tabularnewline
{\scriptsize Change long-term govt yield} & {\scriptsize \citet{sweeney1986pricing}} & {\scriptsize Macro} & {\scriptsize -3.25} & {\scriptsize -9.91} & {\scriptsize 23.17} & {\scriptsize 11.73} & {\scriptsize 2.65} & {\scriptsize 16.16} & {\scriptsize -0.13} \tabularnewline
{\scriptsize Intermediary capital risk factor} & {\scriptsize \citet{he2017intermediary}} & {\scriptsize Inter} & {\scriptsize -0.59} & {\scriptsize 5.58{*}{*}} & {\scriptsize -2.09} & {\scriptsize 0.75} & {\scriptsize -0.73} & {\scriptsize -0.31} & {\scriptsize 0.03} \tabularnewline
{\scriptsize Intermediary capital return factor} & {\scriptsize \citet{he2017intermediary}} & {\scriptsize Inter} & {\scriptsize -0.03} & {\scriptsize -12.63{*}{*}{*}} & {\scriptsize 4.95{*}{*}} & {\scriptsize 2.67{*}{*}} & {\scriptsize 1.04} & {\scriptsize -0.00} & {\scriptsize 0.18} \tabularnewline
{\scriptsize Amihud liquidity} & {\scriptsize \citet{chen2018micro}} & {\scriptsize Liquidity} & {\scriptsize -0.05} & {\scriptsize -0.30} & {\scriptsize 0.05} & {\scriptsize 0.40{*}} & {\scriptsize 0.01} & {\scriptsize -0.18} & {\scriptsize 0.02} \tabularnewline
{\scriptsize Liquidity} & {\scriptsize \citet{pastor2003liquidity}} & {\scriptsize Liquidity} & {\scriptsize 0.91} & {\scriptsize -0.40} & {\scriptsize -1.02} & {\scriptsize -0.94{*}} & {\scriptsize -0.09} & {\scriptsize -0.21} & {\scriptsize 0.45{*}{*}} \tabularnewline
{\scriptsize Roll liquidity} & {\scriptsize \citet{chen2018micro}} & {\scriptsize Liquidity} & {\scriptsize 0.46} & {\scriptsize 0.63} & {\scriptsize -0.58} & {\scriptsize -0.45} & {\scriptsize -0.05} & {\scriptsize 0.29} & {\scriptsize 0.03} \tabularnewline
{\scriptsize Investor sentiment (orth macro)} & {\scriptsize \citet{baker2006investor}} & {\scriptsize Sentiment} & {\scriptsize 0.12} & {\scriptsize 0.09} & {\scriptsize 0.18} & {\scriptsize 0.20{*}{*}} & {\scriptsize 0.07{*}{*}} & {\scriptsize -0.03} & {\scriptsize 0.01} \tabularnewline
{\scriptsize Investor sentiment (orth macro)} & {\scriptsize \citet{huang2015investor}} & {\scriptsize Sentiment} & {\scriptsize -0.38{*}{*}} & {\scriptsize -0.20} & {\scriptsize 0.24{*}} & {\scriptsize 0.01} & {\scriptsize 0.12{*}{*}{*}} & {\scriptsize -0.03} & {\scriptsize -0.05} \tabularnewline
& & {\scriptsize Pseudo $R^{2}$} & {\scriptsize 0.49} & {\scriptsize 0.11} & {\scriptsize 0.52} & {\scriptsize 0.70} & {\scriptsize 0.62} & {\scriptsize 0.54} & {\scriptsize 0.58} \tabularnewline
\hline 
\end{tabular}}}{\scriptsize\par}
\end{sidewaystable}

\section{Additional Results\label{sec:Robustness}}

In this section, we provide additional results for our empirical analyses,
and we discuss extensions of our baseline setup. 

\subsection{Simulation Results}

The selected higher-order model incorporates 7 higher-order factors beyond the traditional FF5M model. Our results indicate that the model performance significantly improves with the 7 higher-order factors where the R-squared of the second-pass regression increases from 31.2\% to 58.7\%. However, as noted by \citet{lewellen2010skeptical}, spurious factors can artificially inflate model fit in the second-pass of the Fama-MacBeth procedure, particularly when the test assets exhibit a strong factor structure. This concern is closely related to the weak-factor identification problem in linear asset-pricing models \citep{bryzgalova2015spurious}. \citet{lewellen2010skeptical} illustrate this concern using the Fama-French 25 Size-BM portfolios which have a particularly strong factor structure. While our analysis mitigates this issue by employing a broad set of 484 test assets and demonstrating the superior out-of-sample performance, we further conduct simulation exercises to confirm that our results are not driven by spurious factors.

\begin{figure}[H]
\caption{Sampling Distribution of R-squareds\label{fig:Sampling-distribution-R2}}

\medskip{}

\begin{minipage}{\textwidth}
\small\singlespacing
This figure plots the sampling distribution (green histogram) of adjusted
cross-sectional R-squareds obtained by applying the FS-FMB procedure to
the FF5M factors and randomly generated higher-order factors. First,
we randomly generate 57 factors $i.i.d.\sim N\left(0,\sigma^2(\text{Mkt-RF})\right)$.
Next, at each step of the FS-FMB procedure, we select one factor to maximize
the adjusted cross-sectional R-squared until $\Delta\left(\text{adj }R^{2}\right)\leq1\%$.
The procedure is repeated 1,000 times to obtain the sampling distribution
of R-squareds. In subplot (a), we include the distribution obtained when the FF5M model is
augmented with 7 randomly generated factors, repeating the procedure 1,000 times.
The pink-dashed line corresponds to the adjusted R-squared of the FF5M model. In subplot (b), we include the distribution obtained when the number of additionally selected factors is constrained to be $\leq 7$ (blue histogram).
The red-dashed line corresponds to the R-squared of the FF5M model with higher-order
factors selected by the forward selection FMB procedure (see Table \ref{tab:Cross-sectional-performance}).
\end{minipage}

\medskip{}

\centering
\subfloat[]{
  \includegraphics[width=0.45\textwidth]{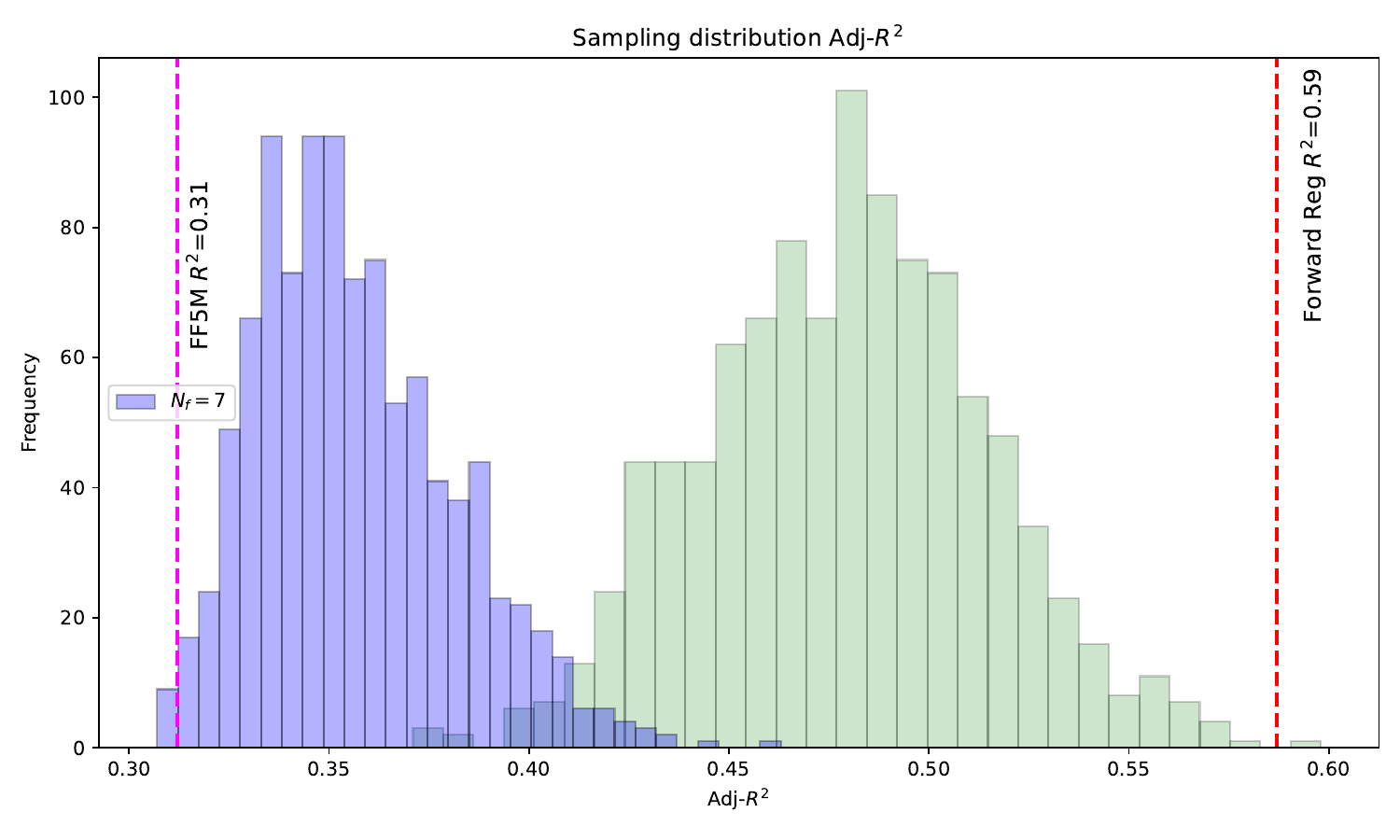}
}
\hfill
\subfloat[]{
  \includegraphics[width=0.45\textwidth]{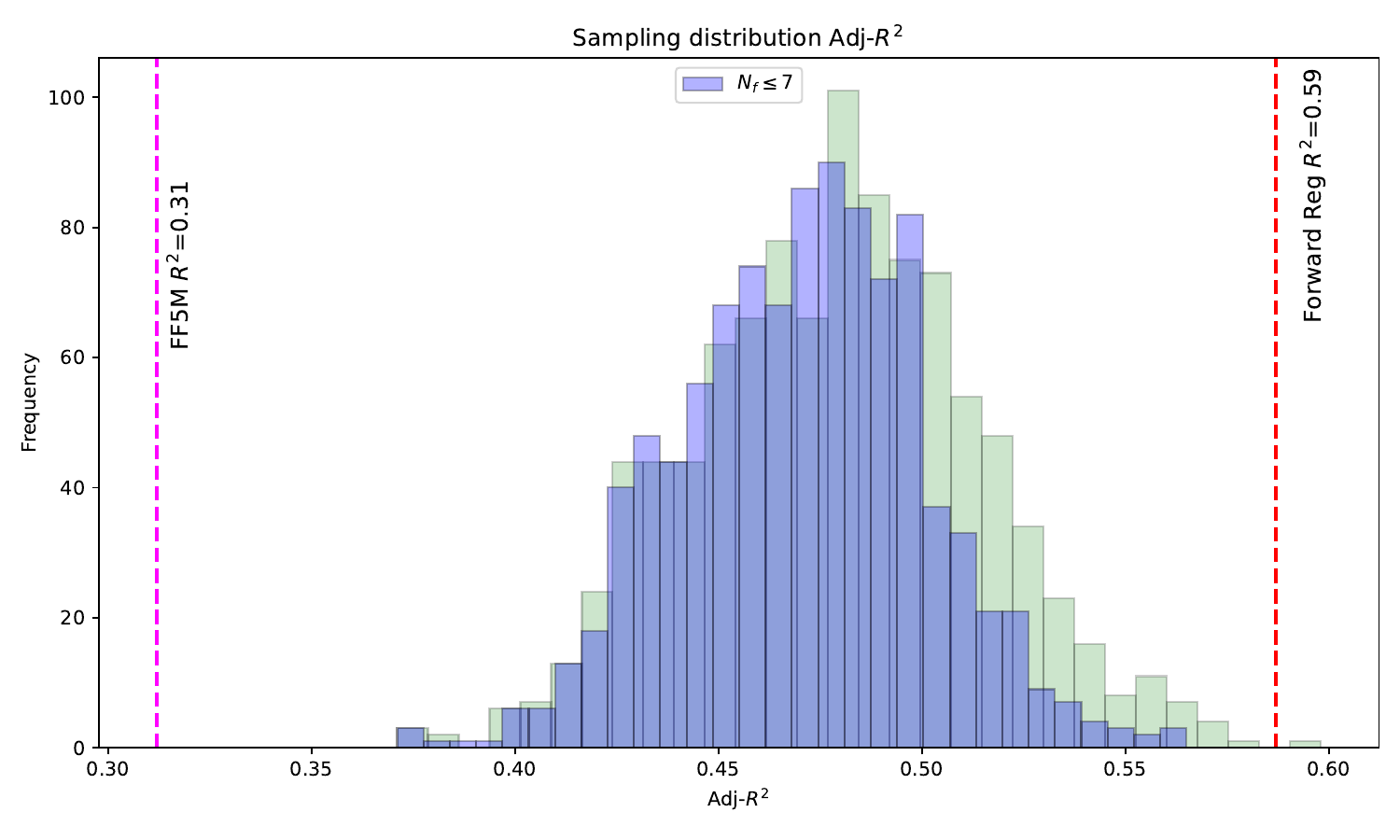}
}
\end{figure}

Therefore, we estimate the sampling distribution of cross-sectional
R-squareds obtained using randomly generated factors, following the
approach advocated by \citet{lewellen2010skeptical}. Specifically,
in each simulation, we generate a set of 57 random factors, matching
the number of candidate higher-order factors. These random factors
are assumed to be $i.i.d.\sim N(0,\sigma^{2})$, and we set $\sigma^{2}$
equal to the sample variance of the market factor (Mkt-RF). We then
apply a forward selection FMB procedure, sequentially selecting factors
that maximize the adjusted cross-sectional R-squared when added to
the FF5M factors. We repeat the procedure 1,000 times by generating
1,000 independent sets of factors.

The green histogram in Figure \ref{fig:Sampling-distribution-R2}
summarizes the sampling distribution of the resulting R-squareds from
the 1,000 simulations. In the simulation exercise, we obtain a cross-sectional
R-squared larger than 0.59 only in 0.1\% of the simulations, meaning
that the adjusted R-squared from the higher-order model is significant
at the 0.1\% level. This simulation exercise highlights the fact that
it is extremely unlikely that our results are obtained by chance.

Note that the procedure can select a number of factors different than
7, which is the number of higher-order factors presented in Table
\ref{tab:Cross-sectional-performance}. In fact, in the simulation
exercise we conduct, the maximum number of random factors selected
is 12. In Panel A, we further compare the sampling distribution
of R-squareds with that obtained by directly adding 7 randomly generated
factors, $i.i.d.\sim N(0,\sigma^{2})$, to the FF5M factors, repeating
the procedure 1,000 times (blue histogram). In this case, the maximum
adjusted cross-sectional R-squared is only 0.463.

Finally, we compare the sampling distribution of R-squareds
with that obtained with the additional constraint that the number
of selected factors cannot be larger than 7 (blue histogram) in Panel B. In this case, we never observe R-squareds larger than 0.59 (maximum adjusted
cross-sectional R-squared of 0.565).

\subsection{An alternative factor zoo}

In order to investigate the robustness of our results in terms of accounting for the factor zoo, we replicate the analysis with the two-pass Fama-MacBeth procedure using an alternative set of zoo factors, based on the factors constructed by \citet{ChenZimmermann2021}. Specifically, for each factor in \citet{ChenZimmermann2021}, we estimate its SDF loading while controlling for our higher-order factor model, which includes the FF5M and the selected higher-order factors. Figure \ref{fig:Cull-Factor-Zoo-StandardFMB-ChenZimmermann} summarizes the results. Panel A shows that the intercept estimates tend to be highly statistically significant, with a median value of 3.97, when the factor model includes only the factor in the factor zoo. When the factor model includes both FF5M and the factor in the factor zoo, the median $t$-statistic decreases substantially but remains significant at the 5\% level at 2.23. When the factor model further includes the higher-order factors, the median $t$-statistic is no longer statistically significant at the 5\% level at 1.87. Panel B reveals that, when the factor model includes only the factor in the factor zoo, the median SDF loading estimate for the factor zoo is not statistically different from zero, with a $t$-statistic of 0.80. Note that these estimates are likely biased because of omitted factors (see, e.g., \citet{feng2020taming}). Controlling for the FF5M, the median $t$-statistic increases substantially to 1.91. 46\% of zoo factors become statistically significant at the 5\% level with a $t$-statistic greater than 2 suggesting that FF5M helps to remove some bias in the factor estimates. Importantly, when we additionally control for the higher-order factors, we find that for 93\% of the zoo factors the risk price estimate is not statistically different from zero, with the median $t$-statistic dropping from 1.91 to only 0.72. Specifically, we find that only 12 out of 159 zoo factors remain significant at the 5\% level. These results confirm our findings based on the factors in \citet{jensen2023there}.\footnote{The factors in \citet{ChenZimmermann2021} that remain significant at the 5\% confidence level are: MomSeason06YrPlus, CBOperProf, ChInv, STreversal, DelFINL, DivYieldST, InvGrowth, GrSaleToGrInv, MomOffSeason16YrPlus, DebtIssuance, CompositeDebtIssuance, OrderBacklog.}

Table \ref{tab:Significant-Loadings-Chen-Zimmermann} in the Online
Appendix presents the fraction of zoo factors in \citet{ChenZimmermann2021}
with loading significantly different from zero at the 5\% confidence
level on each of the FF5M factors and the higher-order factors
selected by the forward selection FMB procedure. Column (1) shows
the fraction of significance for the factors in FF5M model only and
Column (2) reports the fraction of significance for the factors in
the full higher-order factor model. Column (1) shows that the fractions
of significance range from 42.8\% for CMA to 57.2\% for SMB and RMW
for the 6 factors in FF5M. Column (2) shows that the fractions of
significance remain stable for the 6 factors in FF5M, when we further
include the selected higher-order factors. The fractions of significance
for the higher-order factors range from 20.8\% for Mkt-RF2{*}RMW to
47.8\% for SMB2{*}Mom. These results are consistent with the findings
we obtain using the zoo factors in \citet{jensen2023there}, further
highlighting that it is not the case that a specific higher-order factor
accounts for all the results.

\begin{figure}[H]
\caption{Reducing Alternative Zoo Factors in the Cross-Section\label{fig:Cull-Factor-Zoo-StandardFMB-ChenZimmermann}}

\medskip{}

\begin{minipage}{\textwidth}
\small\singlespacing
This figure plots the distribution of absolute t-statistics associated with the risk prices ($t(\lambda)$) and the intercepts ($t(\alpha)$) of the alternative zoo factors from \citet{ChenZimmermann2021} based on US equity data. For each zoo factor, we use the FMB procedure to estimate its risk price controlling for the FF5M (blue histogram) and the FF5M with the higher-order factors selected by the FS-FMB procedure (green histogram). The dashed lines in blue and green correspond to the median absolute t-statistic for the FF5M and for the FF5M with higher-order factor models, respectively. The dashed red line corresponds to the median absolute t-statistic with no controls. The t-statistics are based on Newey-West corrected standard errors.
\end{minipage}

\medskip{}

\centering
\subfloat[\textbf{Distribution of $|t(\alpha)|$}]{
  \includegraphics[width=0.45\textwidth]{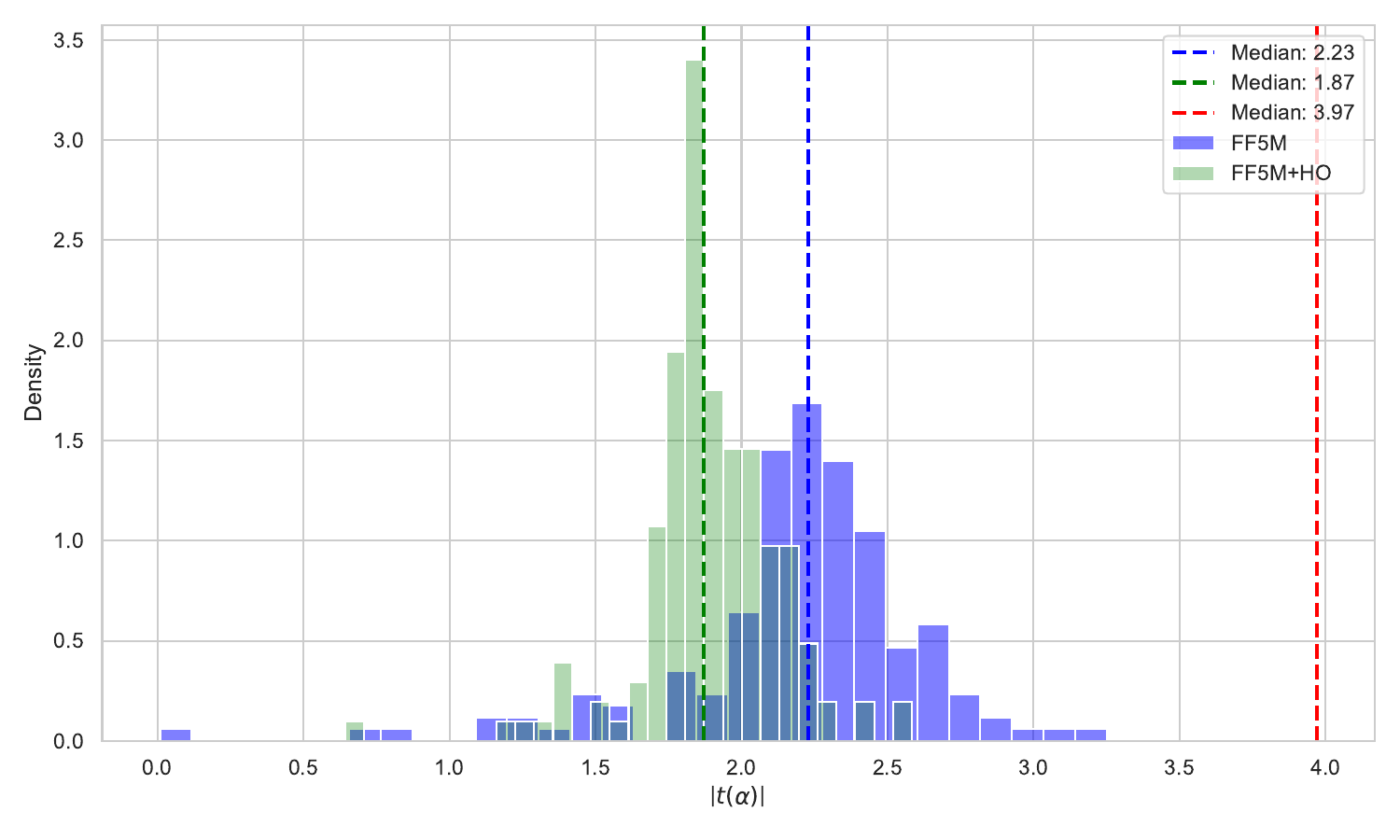}
}
\hfill
\subfloat[\textbf{Distribution of $|t(\lambda)|$}]{
  \includegraphics[width=0.45\textwidth]{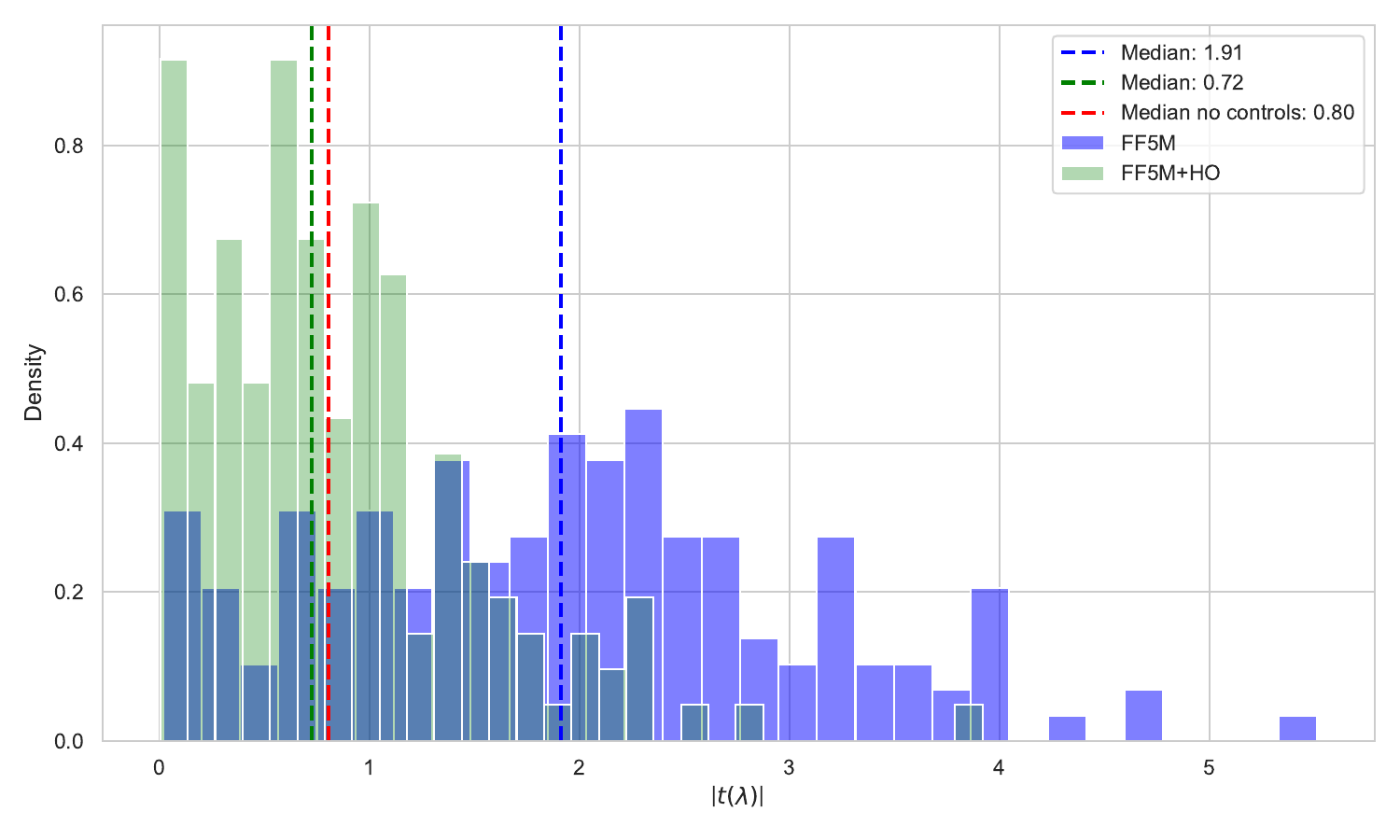}
}
\end{figure}

\subsection{Alternative sets of candidate factors}

We start by considering alternative sets of higher-order factors formed
using the FF5M factors. Table \ref{tab:Forward-Regression-Method_Robust}
presents the estimation results of the forward selection FMB procedure with
four alternative sets of higher-order factors. Specifically, Panel
A presents estimation results using a set of higher-order factors of
degree 2. The latter is a lower degree than the one from the baseline
analysis (i.e., degree 3). These higher-order factors include all pairwise
interactions of the FF5M, and the FF5M factors squared. In this case,
the FS-FMB procedure selects four higher-order factors, reaching
an adjusted cross-sectional R-squared of 0.517 in the final step,
approximately 7 pp smaller than in the baseline setup (see Table \ref{tab:Cross-sectional-performance}).
Panel B presents estimation results using a set of higher-order factors
of higher degree, and specifically of degree up to 4. These factors,
augment the higher-order factors of degree 3 of the baseline setup with
the FF5M to the fourth power (i.e., $f_{i}^{4}$), and pairwise interactions
of the FF5M factors of degree 4, i.e., $f_{i}^{2}\times f_{j}^{2},f_{i}^{3}\times f_{j},f_{i}\times f_{j}^{3}$.
In this case, the FS-FMB procedure selects 9 higher-order factors,
reaching an adjusted cross-sectional R-squared of 0.6 in the final
step. The latter is a R-squared value very similar to that obtained
in the baseline setup, with higher-order factors up to degree 3. These
results suggest that higher-order factors up to degree 3 are most
relevant in explaining the cross-section of asset returns, corroborating
our choice of using higher-order factors of degree 3 in the baseline
analysis.

Furthermore, Panel C and D report the results when augmenting the
FF5M with only pairwise interactions of degree 2 and 3 (Panel C) and
only the powers of degree 2 and 3 (Panel D). The results from these
two panels reveal that incorporating pairwise interactions between
FF5M is crucial. In fact, the model with interactions of degree 2
and 3 reaches an adjusted cross-sectional R-squared of 0.542, selecting
7 higher-order factors. In contrast, the model with powers of degree
2 and 3 selects 3 higher-order factors, for an adjusted cross-sectional
R-squared of 0.467 in the final step. These results highlight that
both the higher-order terms and the interactions of FF5M are important components of the SDF, especially higher-order interactions.

\bibliographystyle{ecta-fullname} 
\bibliography{ref,bib_interactions}  

\newpage

\begin{appendix}

\setcounter{table}{0}
\renewcommand{\thetable}{AI.\arabic{table}}
\setcounter{figure}{0}
\renewcommand{\thefigure}{AI.\arabic{figure}}

\begin{center}
ONLINE APPENDICES
\end{center}

\medskip


This paper has 6 online appendices. In Appendix \ref{sec: additional assumptions}, we describe additional assumptions for Theorem \ref{thm: asy normality}. In Appendix \ref{sec: proof of theorem 1}, we give the proof of Theorem \ref{thm: convergence rate}. In Appendix \ref{sec: proof of theorem asy normality}, we give the proof of Theorem \ref{thm: asy normality}. In Appendix \ref{sec: technical lemmas}, we provide several technical lemmas that are useful for proving the theorems. In Appendix \ref{sec:optimal_stopping}, we introduce a method to select the optimal threshold. In Appendix \ref{sec: additional figures}, we present additional tables and figures.

\section{Assumptions for Theorem \ref{thm: asy normality}}\label{sec: additional assumptions}

In this Appendix, we describe additional assumptions required for the asymptotic normality result for the debiased estimator $\widehat\psi_D$ in Theorem \ref{thm: asy normality}.

\begin{assumption}[First Stage Sparsity]\label{as: first stage sparsity}
For all $j\in[p]$, there exists a set $\S_{0,j}\subset \{j\}^c$ and a vector $\varphi_j = (\varphi_{j,1},\dots,\varphi_{j,p})^{\top}\in \R^{p}$ such that (i) $\supp(\varphi_j) = \S_{0,j}$, (ii) $|\S_{0,j}|\leq \bar s_{T}$, and (iii) the vector $R_j = (R_{j,1},\dots,\R_{j,N})^{\top} = \C_{\{j\}} - \C\varphi_{j}$ satisfies
$
\| \C_{\{j\}^c}^{\top} R_j \|_{\infty}^2 \lesssim N^2\log(Np)/T. 
$
\end{assumption}

This assumption means that for each $j\in[p]$, there exists a sparse approximate projection of the vector $\C_{\{j\}}$ on $\C_{\{j\}^c}$: most cross-sectional variation in the components of the vector $\C_{\{j\}}$ can be explained by at most $\bar s_{T}$ columns in the matrix $\C_{\{j\}^c}$ in the sense that the residual of the linear projection of $\C_{\{j\}}$ on these columns is nearly orthogonal to all columns of the matrix.

\begin{assumption}[Uniform Law of Large Numbers, II]\label{as: ulln 2}
We have
$$
\max_{i\in[2N+2]}\max_{j\in[p]}\left|\frac{1}{T}\sum_{t\in[T]} u_{i, t}\varphi_j^{\top} v_t - \E[u_{i, t} \varphi_j^{\top} v_t]\right| \lesssim_P \sqrt{\frac{\log(Np)}{T}}.
$$
\end{assumption}
This assumption complements Assumption \ref{as: ulln} by imposing the condition that a quantitative version of the law of large numbers applies to the random variables $u_{i,t}\varphi_j^{\top}v_t$ uniformly over $i\in[2N+2]$ and $j\in[p]$. 

\begin{assumption}[First Stage Estimator]\label{as: first stage estimation} For all $j\in[p]$ and some constant $c>0$, we have (i) $|\widehat{\S}_j|\leq \bar s_{T}$ \wpa, (ii) $|\widehat{\S}_j|\geq (1+c)(K/k)|\S_{0,j}|\log T$ \wpa, and (iii) $|\widehat{\S}_j|\lesssim_P |\S_{0,j}|\log T$.
\end{assumption}

This assumption mirrors Assumption \ref{as: estimator sparsity} for the estimation of a different object. In particular, it will be clear from the proofs in Appendix \ref{sec: proof of theorem asy normality} that for all $j\in[p]$, $\widehat\S_j$ serves as an estimator of $\S_{0,j}$. The first part of Assumption \ref{as: first stage estimation} means that we carry out at most $\bar s_{T}$ rounds in Algorithm \ref{alg: fs fama-macbeth} when we choose the set $\widehat\S_j$, which is thus rather natural given the existence of a sparse approximate projection with at most $\bar s_{T}$ components in Assumption \ref{as: first stage sparsity}. The second part of Assumption \ref{as: first stage estimation} requires that the number of selected factors $|\widehat\S_j|$ should be larger than the number of important factors $|\S_{0,j}|$ at least by a multiplicative constant proportional to $\log T$. The third part gives an upper bound on the number of selected factors.

\begin{assumption}[No Strong Omitted Factors]\label{as: extra lln}
For all $j\in[p]$, we have
$$
\frac{1}{N\sqrt T}\sum_{i\in[N]}\sum_{t\in[T]}R_{j,i}\varepsilon_{i,t}m_t = o_P(1).
$$
\end{assumption}

Observe that for each $i\in[N]$ and $j\in[p]$, we expect $T^{-1/2}\sum_{t\in[T]} R_{j, i}\varepsilon_{i,t}m_t = O_P(1)$ given that $\E[R_{j,i}\varepsilon_{i,t}m_t] = 0$. Thus, Assumption \ref{as: extra lln} means that by taking the average over $i\in[N]$, we can further wash out the variability from $T^{-1/2}\sum_{t\in[T]} R_{j, i}\varepsilon_{i,t}m_t$, which is only possible if the random variables $\varepsilon_{i,t}$ do not have a strong cross-sectional dependence and is thus equivalent to the condition that there are no strong omitted factors in the model.

\begin{assumption}[Factors]\label{as: factors cool}
For some constant $C>0$, we have
$$
1/C \leq \lambda_{\min}(\E[v_tv_t^{\top}]) \leq \lambda_{\max}(\E[v_tv_t^{\top}]) \leq C
\quad\text{and}\quad 
\left\|( \E[v_tv_t^{\top}] )^{-1}\right\|_{\infty,1} \leq C.
$$
\end{assumption}
The lower bound in the first part of this assumption means that there is no multicollinearity between factors in the population, and the upper bound is a technical regularity condition. To understand the second part of this assumption, for each $j\in[p]$, let $a_j$ be the $j$th column of the matrix $\E[(v_tv_t^{\top})^{-1}]$. It is then possible to show that $\|\eta_j\|_1 = \sigma_{z,j}^2\|a_j\|_1 - 1$ for $\eta_j$ appearing in \eqref{eq: eta definition}; see Lemma \ref{lem: matrix inverse and l1 norm} in Appendix \ref{sec: technical lemmas}. Thus, given that $\sigma_{z,j}^2\leq C$, which follows from the first part of the assumption, the condition $\|( \E[v_tv_t^{\top}])^{-1}\|_{\infty,1} \leq C$ means that the vectors $\eta_j$ are uniformly bounded in the $\ell_1$-norm.

Next, for all $j\in[p]$, let $\phi_j = (\phi_{j,1},\dots,\phi_{j,p})^{\top}\in\R^p$ be the vector defined by $\phi_{j,j} = 0$ and
\begin{equation}\label{eq: phi definition}
\phi_{j,\{j\}^c} = \left( \E[v_{t,\{j\}^c}v_{t,\{j\}^c}^{\top}] \right)^{-1}\E\left[v_{t,\{j\}^c}\left(\frac{1}{N}\sum_{i\in[N]} r_{i, t} R_{j,i}\right)\right],
\end{equation}
so that $\phi_{j,\{j\}^c}$ is the vector of coefficients of the time-series least-squares projection of the cross-sectional sample average $N^{-1}\sum_{i\in[N]}r_{i,t}R_{j,i}$ on $v_{t,\{j\}^c}$.

\begin{assumption}[Extra Sparsity]\label{as: additional sparsity wtf}
For all $j\in[p]$, there exists a set $\S_{0,j,\phi}\subset\{j\}^c$ such that $|\S_{0,j,\phi}|\leq \bar s_{T}$ and $\|\phi_{\S^c_{0,j,\phi}}\|_1 = o(1/\sqrt{\log(Np)})$.
\end{assumption}

This assumption means that the vectors $\phi_j$ are approximately sparse. Note, however, that the amount of sparsity imposed in this assumption is rather small: the $\ell_1$-norm of the ``remainder'' $\phi_{\S_{0,j,\phi}^c}$ is required to vanish asymptotically but the convergence rate may be very slow.


\begin{assumption}[Growth Condition, II]\label{as: growth conditions 2}
We have $\bar s^2_{T}\log^4(Np) = o(T)$. 
\end{assumption}
\begin{assumption}[Central Limit Theorem]\label{as: woohoo}
For all $j\in[p]$,
$$
\frac{1}{\sqrt T}\sum_{t\in[T]}\frac{z_{t,j}m_t - \E[z_{t,j}m_t]}{\sigma_{z, j}^2} \to_d N(0,\sigma_{\psi,j}^2),
$$
where $\sigma_{\psi,j}^2 = \lim_{T\to\infty} T^{-1}\sum_{t=1}^T\sum_{s=1}^T \E\left[(z_{t, j}m_t - \E[z_{t, j}m_t])(z_{s, j} m_s - \E[z_{s, j} m_s)]/\sigma_{z, j}^4 \right]$.
\end{assumption}

Assumption \ref{as: growth conditions 2} strengthens Assumption \ref{as: growth conditions} but, like Assumption \ref{as: growth conditions}, it allows the number of factors $p$ to be much larger than the number of time periods $T$. Assumption \ref{as: woohoo} is a version of the central limit theorem for time series data.

\section{Proof of Theorem \ref{thm: convergence rate}}\label{sec: proof of theorem 1}
Throughout this Appendix, denote $\Y = (\E[r_{1,t}],\dots,\E[r_{N,t}])^{\top}$ and $\widehat{\Y} = (\bar r_1,\dots,\bar r_N)^{\top}$. Also, let $\ell\colon\R^p\to\R$ be the function defined by $\ell(x) = -\|\widehat{\Y} - \widehat{\C}x\|_2^2/N$ for all $x\in\R^p$. In addition, for any vector $x = (x_1,\dots,x_p)^{\top}\in\R^p$ and any set $\S\subset[p]$, let $x[\S] = (\overline{x}_1,\dots,\overline{x}_p)^{\top}\in\R^p$ be the vector defined by $\overline{x}_j = x_j$ for all $j\in\S$ and $\overline{x}_j = 0$ for all $j\in\S^c$. Observe that the notation $x[\S]$ is different from the notation $x_{\S}$, which we use throughout the paper. Further, for all $j\in[p]$, let $\widetilde{\S}_{D,j} = \widehat{\S}_{D,j}\setminus\{j\}\subset[p]$ and let $\widehat{\varphi}_j = (\widehat{\varphi}_{j,1},\dots,\widehat{\varphi}_{j,p})^{\top}\in\R^p$ be the vector defined by
$
\widehat{\varphi}_{\widetilde{\S}_{D,j}} = (\widehat{\C}_{\widetilde{\S}_{D,j}}^{\top}\widehat{\C}_{\widetilde{\S}_{D,j}})^{\top}\widehat{\C}_{\widetilde{\S}_{D,j}}^{\top}\widehat{\C}_{\{j\}}
$
and $\widehat\varphi_{l} = 0$ for all $l\in\widetilde{\S}_{D,j}^c$.

Before proving Theorem \ref{thm: convergence rate}, we state and prove four auxiliary lemmas.

\begin{lemma}\label{lem: individual covariance bound}
Suppose that Assumptions \ref{as: ulln} and \ref{as: growth conditions} are satisfied. Then
$$
\max_{i\in[N]}\max_{j\in[p]}\Big|\widehat \C_{i,j} - \C_{i,j}\Big| \lesssim_P \sqrt{\frac{\log(Np)}{T}}.
$$
\end{lemma}
\begin{proof}
For all $(i,j)\in[N]\times[p]$, we have
\begin{align}
\widehat \C_{i,j} 
& = \frac{1}{T}\sum_{t\in[T]} \beta_i^{\top}v_t(f_{t,j} - \bar f_j) + \frac{1}{T}\sum_{t\in[T]} \varepsilon_{i,t}(f_{t,j} - \bar f_j) \nonumber\\
& =  \frac{1}{T}\sum_{t\in[T]} \beta_i^{\top}v_tv_{t,j} + \frac{1}{T}\sum_{t\in[T]} \beta_i^{\top}v_t(\E[f_{t,j}] - \bar f_j)\nonumber\\
& \quad + \frac{1}{T}\sum_{t\in[T]} \varepsilon_{i,t} v_{t,j} + \frac{1}{T}\sum_{t\in[T]} \varepsilon_{i,t} (\E[f_{t,j}] - \bar f_j).\label{eq: chat expansion}
\end{align}
Subtracting $\C_{i,j} = \beta_i^{\top} \E[v_t v_{t,j}]$ from both sides and applying Assumptions \ref{as: ulln} and \ref{as: growth conditions} yields the asserted claim.
\end{proof}

\begin{lemma}\label{lem: c-psi bound}
Suppose that Assumptions \ref{as: ulln} and \ref{as: growth conditions} are satisfied. Then 
$$
\|(\widehat{\C} - \C)\psi\|_2 \lesssim_P \sqrt{\frac{N\log(Np)}{T}}.
$$
\end{lemma}
\begin{proof}
The claim follows from combining \eqref{eq: chat expansion} and Assumptions \ref{as: ulln} and \ref{as: growth conditions}.
\end{proof}

\begin{lemma}\label{lem: empirical sparse eigenvalues}
Suppose that Assumptions \ref{as: ulln}, \ref{as: sparse eigenvalues 2}, and \ref{as: growth conditions} are satisfied. Then
$$
\frac{k}{2}\leq \lambda_{\min}\left(\frac{\widehat{\C}_{\S}^{\top}\widehat{\C}_{\S}}{N}\right) \leq \lambda_{\max}\left(\frac{\widehat{\C}_{\S}^{\top}\widehat{\C}_{\S}}{N}\right) \leq 2K
$$
for all $\S\subset[p]$ such that $|\S|\leq 3\bar s_{T} + 1$ \wpa.
\end{lemma}
\begin{proof}
For any matrix $A$, let $s_{\min}(A) = \sqrt{\lambda_{\min}(A^{\top}A)}$ and $s_{\max}(A) = \sqrt{\lambda_{\max}(A^{\top}A)}$ be its smallest and largest singular values. Then uniformly over $\S\subset[p]$ such that $|\S|\leq 3\bar s_{T}$, we have
\begin{align*}
|s_{\min}(\widehat \C_{\S}) - s_{\min}(\C_{\S})|^2&
 \leq \|\widehat{\C}_{\S} - \C_{\S}\|_2^2
 \leq \sum_{i\in[N]}\sum_{j\in\S}|\widehat{\C}_{i,j} - \C_{i,j}|^2 \\
 & \lesssim_P \frac{|S|N\log(Np)}{T} \leq \frac{\bar s_{T}N \log(Np)}{T} = o(N),
\end{align*}
where the first inequality follows from Weyl's inequality, the second from the definition of $\ell_2$-norm, the third from Lemma \ref{lem: individual covariance bound}, and the last bound from Assumption \ref{as: growth conditions}. Combining this bound with Assumption \ref{as: sparse eigenvalues 2} yields the lower bound in the asserted claim. To obtain the upper bound, we proceed similarly by noting that Weyl's inequality also gives $|s_{\max}(\widehat \C_{\S}) - s_{\max}(\C_{\S})|^2 \leq \|\widehat{\C}_{\S} - \C_{\S}\|_2^2$.
\end{proof}

\begin{lemma}\label{lem: restricted eigenvalues}
Suppose that Assumptions \ref{as: ulln}, \ref{as: sparse eigenvalues 2}, and \ref{as: growth conditions} are satisfied. Then for any $c>0$,
$$
-\frac{k}{\sqrt{1+c}}\|y-x\|_2^2 \geq \ell(y) - \ell(x) - \nabla\ell(x)^{\top}(y-x)\geq - \sqrt{1+c}K\|y-x\|_2^2
$$
for all $x,y\in\R^p$ satisfying $\|y-x\|_0\leq 3\bar s_{T}+1$ \wpa.
\end{lemma}

\begin{proof}
Observe that for all $x\in\R^p$,
\begin{equation}\label{eq: nabla expression}
\ell(x) = -\frac{1}{N}\widehat{\Y}^{\top}\widehat{\Y} + \frac{2}{N}\widehat{\Y}^{\top}\widehat{\C}x - \frac{1}{N}x^{\top}\widehat{\C}^{\top}\widehat{\C}x\quad\text{and}\quad
\nabla\ell(x) = \frac{2}{N}\widehat{\C}^{\top}\widehat{\Y} - \frac{2}{N}\widehat{\C}^{\top}\widehat{\C}x.
\end{equation}
Thus, for all $x,y\in\R^p$,
\begin{align*}
&  \ell(y) - \ell(x) - \nabla\ell(x)^{\top}(y-x) = \frac{2}{N}\widehat{\Y}^{\top}\widehat{\C}(y-x) - \frac{1}{N}y^{\top}\widehat{\C}^{\top}\widehat{\C}y + \frac{1}{N}x^{\top}\widehat{\C}^{\top}\widehat{\C}x \\
&\quad  - \frac{2}{N}\widehat{\Y}^{\top}\widehat{\C}(y-x) + \frac{2}{N}x^{\top}\widehat{\C}^{\top}\widehat{\C}(y-x) = -\frac{1}{N}(y-x)^{\top}\widehat{\C}^{\top}\widehat{\C}(y-x).
\end{align*}
The asserted claim follows from combining this equality with Lemma \ref{lem: empirical sparse eigenvalues}.
\end{proof}

\begin{proof}[Proof of Theorem \ref{thm: convergence rate}]
For all $\S\subset[p]$, denote
$$
\mathcal F(\S) = \max_{x\in\R^p\colon\supp(x)\subset\S}\ell(x).
$$
Observe that under Assumption \ref{as: sparse eigenvalues 1}, with probability approaching one, the set $\widehat\S$ is equal to the set $\S$ produced by Algorithm \ref{alg: fs fama-macbeth} with $R^2_{FM}(\cdot)$ replaced by $\mathcal F(\cdot)$ and
$$
\widehat\psi_{\widehat\S} \in \argmax_{x\in\R^p\colon \supp(x)\subset\widehat\S}\ell(x).
$$
Thus, to prove the asserted claim, we can apply Theorem 6 in \citet{EKDN18}. To do so, we first derive three auxiliary bounds.

First, using \eqref{eq: nabla expression}, we have
$
\nabla\ell(\psi[\S_0]) = 2\widehat{\C}^{\top}(\widehat{\Y} - \widehat{\C}_{\S_0}\psi_{\S_0})/N.
$
Thus,
$$
\|\nabla\ell(\psi[\S_0])\|_{\infty}\leq \frac{2}{N}\max_{j\in[p]}\|\widehat{\C}_{\{j\}}\|_2\times\|\widehat{\Y} - \widehat{\C}_{\S_0}\psi_{\S_0}\|_2
$$
Here, $\max_{j\in[p]}\|\widehat{\C}_{\{j\}}\|_2 \lesssim_P \sqrt N $ by Lemma \ref{lem: empirical sparse eigenvalues}. Also, by \eqref{eq: mean and covariances},
$$
\|\widehat{\Y} - \widehat{\C}_{\S_0}\psi_{\S_0}\|_2
\leq \|\widehat{\Y} - \Y\|_2 + \|\C\psi - \C_{\S_0}\psi_{\S_0}\|_2 + \|(\C_{\S_0} - \widehat{\C}_{\S_0})\psi_{\S_0}\|_2.
$$
In addition, $\|\widehat{\Y} - \Y\|_2 \lesssim_P \sqrt{N\log(Np)/T}$ by Assumption \ref{as: ulln} and $\|\C\psi - \C_{\S_0}\psi_{\S_0}\|_2 \lesssim_P \sqrt{N\log(Np)/T}$ by Assumption \ref{as: sdf loadings}. Moreover,
\begin{align*}
\| (\C_{\S_0} - \widehat{\C}_{\S_0})\psi_{\S_0} \|
& \leq \|(\C - \widehat{\C})\psi\|_2 + \|(\C_{\S_0^c} - \widehat{\C}_{\S_0^c})\psi_{\S_0^c}\|_2 \\
& \leq \|(\C - \widehat{\C})\psi\|_2 + \max_{j\in[p]}\|\C_{\{j\}} - \widehat{\C}_{\{j\}}\|_2 \times \|\psi_{\S_0^c}\|_1 \lesssim_P \sqrt{\frac{N\log(Np)}{T}}
\end{align*}
by the triangle inequality, Lemmas \ref{lem: individual covariance bound} and \ref{lem: c-psi bound}, and Assumptions \ref{as: sdf loadings} and \ref{as: growth conditions}. Thus,
\begin{equation}\label{eq: nabla bound}
\|\nabla\ell(\psi[\S_0])\|_{\infty} \lesssim_P \sqrt{\frac{\log(Np)}{T}},
\end{equation}
which is the first bound we need.

Second,
\begin{equation}\label{eq: final step for elenberg}
|\ell(\psi[\S_0]) - \ell(\mathbf 0_p)| \leq \frac{1}{N}\|\widehat\Y - \widehat\C\psi[\S_0]\|_2^2 + \frac{1}{N}\|\widehat\Y\|_2^2\lesssim_P 1
\end{equation}
by the bounds above and Assumption \ref{as: returns}.

Third, given that $|\S_0| + |\widehat\S| \leq 2\bar s_{T}$ \wpas by Assumptions \ref{as: sdf loadings} and \ref{as: estimator sparsity} and $|\widehat{\S}|\geq (1+c)(K/k)|\S_0|\log T$ \wpas by Assumption \ref{as: estimator sparsity}, it follows from Lemma \ref{lem: restricted eigenvalues} above and Theorem 1 and Corollary 1 in \citet{EKDN18} that
\begin{equation}\label{eq: i am exhausted elenberg}
\ell(\widehat\psi) - \ell(\mathbf 0_p) \geq (1-T^{-1})(\ell(\psi[\S_0]) - \ell(\mathbf 0_p)
\end{equation}
\wpa. 

Now, combining \eqref{eq: nabla bound}, \eqref{eq: final step for elenberg}, and \eqref{eq: i am exhausted elenberg} above with Theorem 6 in \citet{EKDN18}, we obtain
\begin{equation}\label{eq: something we need}
\|\widehat\psi - \psi[\S_0]\|_2^2  \lesssim_P (|\S_0| + |\widehat{\S}|)\frac{\log(Np)}{T} + \frac{1}{T} \lesssim_P \frac{|\S_0|\log^2(NTp)}{T},
\end{equation}
where the second inequality follows from Assumption \ref{as: estimator sparsity}. Thus,
\begin{align*}
\|\widehat\psi - \psi\|_2^2 & \lesssim \|\widehat\psi - \psi[\S_0]\|_2^2 + \|\psi[\S_0] - \psi\|_2^2 \\
& \leq \|\widehat\psi - \psi[\S_0]\|_2^2 + \|\psi[\S_0] - \psi\|_1^2 \lesssim_P \frac{|\S_0|\log^2(NTp)}{T}
\end{align*}
by the triangle inequality and Assumption \ref{as: sdf loadings}, yielding the first asserted claim. Also,
\begin{align*}
\|\widehat\psi - \psi\|_1^2 & \lesssim \|\widehat\psi - \psi[\S_0]\|_1^2 + \|\psi[\S_0] - \psi\|_1^2 \\
& \leq (|\S_0|+|\widehat\S|)\|\widehat\psi - \psi[\S_0]\|_2^2 + \|\psi[\S_0] - \psi\|_1^2 \lesssim_P \frac{|\S_0|^2\log^3(NTp)}{T}
\end{align*}
by the Cauchy-Schwarz inequality and Assumption \ref{as: sdf loadings}, yielding the second asserted claim.
\end{proof}

\section{Proof of Theorem \ref{thm: asy normality}}\label{sec: proof of theorem asy normality}

In this Appendix, we use the same additional notation as that introduced in the previous Appendix. Also, for all $j\in[p]$, let $\ell_j\colon\R^{p-1}\to \R$ be the function defined by $\ell_j(x) = \| \widehat{\C}_{\{j\}} - \widehat\C_{\{j\}^c} x \|_2^2/N$ for all $x\in\R^{p-1}$. In addition, for all $j\in[p]$, let $\widetilde\S_j = \widehat\S_{D,j}\setminus\{j\}$ and let $\widehat\varphi_j = (\widehat\varphi_{j,1},\dots,\widehat\varphi_{j,p})^{\top}\in\R^p$ be the vector defined by $\widehat\varphi_{j,j} = 0$ and
\begin{equation}\label{eq: varphi estimator definition}
\widehat\varphi_{j,\{j\}^c} = \left(  \widehat\C_{\widetilde\S_j}^{\top}\widehat\C_{\widetilde\S_j} \right)^{-1}\widehat\C_{\widetilde\S_j}^{\top}\widehat\C_{\{j\}}.
\end{equation}
To prove Theorem \ref{thm: asy normality}, we first state and prove seven auxiliary lemmas.

\begin{lemma}\label{lem: varphi bound}
Suppose that Assumptions  \ref{as: sparse eigenvalues 2}, \ref{as: first stage sparsity}, and \ref{as: growth conditions 2} are satisfied. Then for all $j\in[p]$, we have $\|\varphi_j\|_2\lesssim 1$.
\end{lemma}
\begin{proof}
Fix $j\in[p]$. It follows from the equality $R_j = \C_{\{j\}} - \C_{\S_{0,j}}\varphi_{j,\S_{0,j}}$ that
$$
\varphi_{j,\S_{0,j}} = (\C_{\S_{0,j}}^{\top}\C_{\S_{0,j}})^{-1} (\C_{\S_{0,j}}^{\top}\C_{\{j\}} - \C_{\S_{0,j}}^{\top}R_j),
$$
where the matrix $\C_{\S_{0,j}}^{\top}\C_{\S_{0,j}}$ is non-singular by Assumptions \ref{as: sparse eigenvalues 2} and \ref{as: first stage sparsity}. Thus,
$$
\| \varphi_{j,\S_{0,j}}\|_2 \leq \frac{\| \C_{\S_{0,j}}^{\top}\C_{\{j\}} \|_2 + \| \C_{\S_{0,j}}^{\top}R_j \|_2}{\lambda_{\min}(\C_{\S_{0,j}}^{\top}\C_{\S_{0,j}})}\lesssim 1 + \sqrt{\frac{|\S_{0,j}|\log(Np)}{T}} \lesssim 1
$$
by Assumptions \ref{as: sparse eigenvalues 2}, \ref{as: first stage sparsity}, and \ref{as: growth conditions 2}. Since $\varphi_{j,l} = 0$ for all $l\in \S_{0,j}^c$, the asserted claim follows. 
\end{proof}

\begin{lemma}\label{lem: r upper bound}
Suppose that Assumptions \ref{as: sparse eigenvalues 2},  \ref{as: first stage sparsity}, and \ref{as: growth conditions 2} are satisfied. Then for all $j\in[p]$, we have $N\lesssim \|R_j\|_2^2 \lesssim N$.
\end{lemma}
\begin{proof}
Fix $j\in[p]$. Let $\nu = (\nu_1,\dots,\nu_p)^\top\in\R^p$ be the vector defined by $\nu_j = 1$, $\nu_l = - \varphi_{j,l}$ for all $l\in\S_{0,j}$, and $\nu_l = 0$ for all $l\in(\S_{0,j}\cup\{j\})^c$. Then $R_j = \C\nu$ and $\|\nu\|_0\leq \bar s_{T}+1$ by Assumption \ref{as: first stage sparsity}. Therefore,
$$
\|R_j\|_2^2 = R_j^{\top}R_j = (\C\nu)^{\top}(\C\nu) = \nu^{\top}(\C^{\top}\C)\nu \geq N k \|\nu\|_2^2 \geq kN,
$$
yielding the lower bound in the asserted claim. Further,
\begin{align}
\|R_j\|_2^2
& = R_j^{\top}(\C_{\{j\}} - \C\varphi_j) 
\leq \|R_j\|_2\times \|\C_{\{j\}}\|_2 + \|\C_{\{j\}^c}^{\top}R_j\|_{\infty}\times \|\varphi_j\|_1\nonumber \\
& \lesssim \sqrt N \|R_j\|_2 + N\sqrt{\log(Np)/T}\sqrt{|\S_{0,j}|}\|\varphi_j\|_2 \lesssim \sqrt N\|R_j\|_2 + N,\label{eq: simple upper bound for r}
\end{align}
where the second inequality follows from Assumption \ref{as: sparse eigenvalues 2} and \ref{as: first stage sparsity} and the third from Assumptions \ref{as: first stage sparsity} and \ref{as: growth conditions 2} and Lemma \ref{lem: varphi bound}. In turn, \eqref{eq: simple upper bound for r} implies the upper bound in the asserted claim.
\end{proof}


\begin{lemma}\label{lem: c-varphi bound}
Suppose that Assumptions \ref{as: ulln}, \ref{as: ulln 2}, and \ref{as: growth conditions 2} are satisfied. Then for all $j\in[p]$, we have
$$
\|(\widehat{\C} - \C)\varphi_j\|_2 \lesssim_P \sqrt{\frac{N\log(Np)}{T}}.
$$
\end{lemma}
\begin{proof}
The claim follows from the same argument as that in the proof of Lemma \ref{lem: c-psi bound}, where we now use Assumption \ref{as: ulln 2} and replace Assumption \ref{as: growth conditions} by Assumption \ref{as: growth conditions 2}.
\end{proof}

\begin{lemma}\label{lem: restricted eigenvalues j}
Suppose that Assumptions \ref{as: ulln}, \ref{as: sparse eigenvalues 2}, and \ref{as: growth conditions 2} are satisfied. Then for any $c>0$,
$$
-\frac{k}{\sqrt{1+c}}\|y-x\|_2^2 \geq \ell_j(y) - \ell_j(x) - \nabla\ell_j(x)^{\top}(y-x)\geq - \sqrt{1+c}K\|y-x\|_2^2
$$
for all $x, y\in\R^{p-1}$ satisfying $\|y-x\|_0\leq 3\bar s_{T}+1$ \wpa.
\end{lemma}

\begin{proof}
The claim follows from the same argument as that in the proof of Lemma \ref{lem: restricted eigenvalues}, where we replace Assumption \ref{as: growth conditions} by Assumption \ref{as: growth conditions 2}.
\end{proof}

\begin{lemma}\label{lem: varphi estimation}
Suppose that Assumptions \ref{as: ulln}, \ref{as: sparse eigenvalues 2}, \ref{as: first stage sparsity}, \ref{as: ulln 2}, \ref{as: first stage estimation}, and \ref{as: growth conditions 2} are satisfied. Then for all $j\in[p]$,
$$
\|\widehat\varphi_j - \varphi_j\|_2 \lesssim_P \sqrt{\frac{\bar s_{T}\log(Np)}{T}}.
$$
\end{lemma}
\begin{proof}
Fix $j\in[p]$. As in the proof of Theorem \ref{thm: convergence rate}, we start with three preliminary bounds. First,
\begin{align*}
\nabla\ell_j(\varphi_{j,\{j\}^c})
& = \frac{2}{N}\widehat{\C}_{\{j\}^c}^{\top}(\widehat{\C}_{\{j\}} - \widehat{\C}_{\S_{0,j}}\varphi_{\S_{0,j}})\\
& = \frac{2}{N}(\widehat{\C}_{\{j\}^c} - \C_{\{j\}^c})^{\top}(\widehat{\C}_{\{j\}} - \widehat{\C}_{\S_{0,j}}\varphi_{\S_{0,j}}) + \frac{2}{N}\C_{\{j\}^c}^{\top}(\widehat{\C}_{\{j\}} - \widehat{\C}_{\S_{0,j}}\varphi_{\S_{0,j}}).
\end{align*}
Here,
\begin{align}
& \|\widehat{\C}_{\{j\}} - \widehat{\C}_{\S_{0,j}}\varphi_{\S_{0,j}}\|_2
\leq \|\widehat{\C}_{\{j\}} - \C_{\{j\}}\|_2 \nonumber\\
& \qquad + \|\C_{\{j\}} - \C_{\S_{0,j}}\varphi_{\S_{0,j}}\|_2 + \|(\C_{\S_{0,j}} - \widehat{\C}_{\S_{0,j}})\varphi_{\S_{0,j}}\|_2 \lesssim_P \sqrt N\label{eq: we will need this bound below}
\end{align}
by Lemmas \ref{lem: individual covariance bound}, \ref{lem: r upper bound}, and \ref{lem: c-varphi bound}. Thus,
$$
\Big\| (\widehat{\C}_{\{j\}^c} - \C_{\{j\}^c})^{\top}(\widehat{\C}_{\{j\}} - \widehat{\C}_{\S_{0,j}}\varphi_{\S_{0,j}}) \Big\|_{\infty} \lesssim_P N\sqrt{\frac{\log(Np)}{T}}
$$
by the Cauchy-Schwarz inequality and Lemma \ref{lem: individual covariance bound}. Also,
\begin{align*}
& \|C_{\{j\}^c}^{\top}(\widehat{\C}_{\{j\}} - \widehat{\C}_{\S_{0,j}}\varphi_{\S_{0,j}})\|_\infty
\leq \|C_{\{j\}^c}^{\top}(\widehat{\C}_{\{j\}} - \C_{\{j\}})\|_\infty \\
& \qquad  + \|C_{\{j\}^c}^{\top}(\C_{\{j\}} - \C_{\S_{0,j}}\varphi_{\S_{0,j}})\|_\infty + \|C_{\{j\}^c}^{\top}(\C_{\S_{0,j}} - \widehat{\C}_{\S_{0,j}})\varphi_{\S_{0,j}}\|_\infty \lesssim_P N\sqrt{\frac{\log(Np)}{T}}
\end{align*}
by Assumptions \ref{as: sparse eigenvalues 2} and \ref{as: first stage sparsity} and Lemmas \ref{lem: individual covariance bound} and \ref{lem: c-varphi bound}. Thus,
\begin{equation}\label{eq: nabla bound j}
\|\nabla\ell_j(\varphi_{j,\{j\}^c})\|_{\infty} \lesssim_P \sqrt{\frac{\log(Np)}{T}},
\end{equation}
which is the first bound we need.

Second,
\begin{equation}\label{eq: final step for elenberg j}
|\ell_j(\varphi_{j,\{j\}^c}) - \ell_j(\mathbf 0_{p-1})| \leq \frac{1}{N}\|\widehat\C_{\{j\}} - \widehat\C_{\S_{0,j}}\varphi_{\S_{0,j}}\|_2^2 + \frac{1}{N}\|\widehat\C_{\{j\}}\|_2^2\lesssim_P 1
\end{equation}
by the bound in \eqref{eq: we will need this bound below} and Lemma \ref{lem: empirical sparse eigenvalues}, which is applicable because Assumption \ref{as: growth conditions 2} implies Assumption \ref{as: growth conditions}.

Third, given that $|\S_{0,j}| + |\widehat\S_j| \leq 2\bar s_{T}$ \wpas by Assumptions \ref{as: first stage sparsity} and \ref{as: first stage estimation} and $|\widehat{\S}_j|\geq (1+c)(K/k)|\S_{0,j}|\log T$ \wpas by Assumption \ref{as: first stage estimation}, it follows from Lemma \ref{lem: restricted eigenvalues j} above and Theorem 1 and Corollary 1 in \citet{EKDN18} that
\begin{equation}\label{eq: i am exhausted elenberg j}
\ell_j(\widetilde\varphi_{j}) - \ell_j(\mathbf 0_p) \geq (1-T^{-1})(\ell_j(\varphi_{j,\{j\}^c}) - \ell(\mathbf 0_{p-1}))
\end{equation}
\wpas, where we denoted
$
\widetilde\varphi_{j} = (  \widehat\C_{\widehat\S_j}^{\top}\widehat\C_{\widehat\S_j})^{-1}\widehat\C_{\widehat\S_j}^{\top}\widehat\C_{\{j\}};
$
note the difference between $\widetilde\varphi_j$ and $\widehat\varphi_{j,\{j\}^c}$ in \eqref{eq: varphi estimator definition}.

Now, we proceed as in the proof of Theorem 6 in \citet{EKDN18}. (We cannot directly apply their theorem as it would be a result about $\widetilde\varphi_j$, and we need a result about $\widehat\varphi_j$.) We have
\begin{equation}\label{eq: proof replication 1}
\ell_j(\widehat\varphi_{j,\{j\}^c}) - \ell_j(\varphi_{j,\{j\}^c}) \geq \ell_j(\widetilde\varphi_j) - \ell_j(\varphi_{j,\{j\}^c}) \geq T^{-1}(\ell_j(\mathbf 0_{p-1}) - \ell_j(\varphi_{j,\{j\}^c}))
\end{equation}
\wpa, where the first inequality follows from the fact that $\widehat\S_j\subset\widetilde\S_j$ and the second from rearranging the terms in \eqref{eq: i am exhausted elenberg j}. Also, denoting $\Delta = \widehat\varphi_{j,\{j\}^c} - \varphi_{j,\{j\}^c}$, we have
\begin{equation}\label{eq: proof replication 2}
\ell_j(\widehat\varphi_{j,\{j\}^c}) - \ell_j(\varphi_{j,\{j\}^c}) - \nabla\ell_j(\varphi_{j,\{j\}^c})^{\top}\Delta \leq - k\|\Delta\|_2^2/2
\end{equation}
\wpas by Lemma \ref{lem: restricted eigenvalues j} since $\|\widehat\varphi_{j,\{j\}^c} - \varphi_{j,\{j\}^c}\|_0\leq 3\bar s_{T}$ \wpas by Assumptions \ref{as: estimator sparsity}, \ref{as: first stage sparsity}, and \ref{as: first stage estimation}. Combining \eqref{eq: proof replication 1} and \eqref{eq: proof replication 2}, we have
\begin{align*}
k\|\Delta\|_2^2 / 2 
& \lesssim_P \nabla\ell_j(\varphi_{j,\{j\}^c})^{\top}\Delta + T^{-1}(\ell_j(\varphi_{j,\{j\}^c}) - \ell_j(\mathbf 0_{p-1})) \\
& \lesssim_P \|\nabla\ell_j(\varphi_{j,\{j\}^c})\|_{\infty} \|\Delta\|_1 + \frac{1}{T}
 \lesssim_P \|\Delta\|_2 \sqrt{\frac{\bar s_{T}\log(Np)}{T}} + \frac{1}{T},
\end{align*}
where the second inequality follows from \eqref{eq: final step for elenberg j} and the third from \eqref{eq: nabla bound j}. The last bound in turn implies the asserted claim.
\end{proof}

\begin{lemma}\label{eq: phi definition equivalent}
Suppose that Assumption \ref{as: factors cool} is satisfied. Then for all $j\in[p]$,
$$
\phi_{j,\{j\}^c} = \frac{1}{N}\sum_{i\in[N]}(\beta_{i,j}\eta_{j,\{j\}^c} + \beta_{i,\{j\}^c}) R_{j,i}.
$$
\end{lemma}
\begin{proof}
The asserted claim follows from substituting \eqref{eq: model} and $v_{t,j} = \eta_{j,\{j\}^c}v_{t,\{j\}^c} + z_{t,j}$ into \eqref{eq: phi definition} and cancelling the matrices, which is feasible under Assumption \ref{as: factors cool}.
\end{proof}

\begin{lemma}\label{lem: cool sparsity bound}
Suppose that Assumptions \ref{as: first stage sparsity} and \ref{as: factors cool} are satisfied. Then for all $j\in[p]$, we have $\|\phi_j\|_{\infty} \lesssim \sqrt{\log(Np)/T}$.
\end{lemma}
\begin{proof}
Fix $j\in[p]$ and recall that for all $i\in[N]$ and $l\in[p]$, we have 
$$
\C_{i,l} = \beta_{i,j}\E[v_{t,j}v_{t,l}] + \beta_{i,\{j\}^c}^{\top}\E[v_{t,\{j\}^c}v_{t,l}].
$$
Combining this equality with \eqref{eq: eta definition}, it follows that
$$
(\C_{i, l})_{l\in\{j\}^c}^{\top} = \E[v_{t,\{j\}^c}v_{t,\{j\}^c}^{\top}](\beta_{i, j}\eta_{j, \{j\}^c} + \beta_{i, \{j\}^c}).
$$
Therefore,
$$
\left\| \frac{1}{N}\sum_{i\in[N]} \E[v_{t,\{j\}^c}v_{t,\{j\}^c}^{\top}](\beta_{i,j}\eta_{j,\{j\}^c} + \beta_{i,\{j\}^c}) R_{j,i} \right\|_{\infty} \lesssim \sqrt{\frac{\log(Np)}{T}}
$$
by Assumption \ref{as: first stage sparsity}. Hence,
$$
\|\phi_{j,\{j\}^c}\|_{\infty} = \left\|\frac{1}{N}\sum_{i\in[N]}(\beta_{i,j}\eta_{j,\{j\}^c} + \beta_{i,\{j\}^c}) R_{j,i} \right\|_{\infty} \lesssim \sqrt{\frac{\log(Np)}{T}},
$$
where the equality follows from Lemma \ref{eq: phi definition equivalent} and the inequality from Lemma \ref{lem:principal-submatrix-inverse-bound} and Assumption \ref{as: factors cool}. Combining this bound with the fact that $\phi_{j,j} = 0$, yields the asserted claim.
\end{proof}

\begin{lemma}\label{lem: hopefully final}
Suppose that Assumptions \ref{as: first stage sparsity}, \ref{as: factors cool}, and \ref{as: growth conditions 2} are satisfied. Then for all $j\in[p]$,
$$
R_j^{\top}R_j = \sum_{i\in[N]}R_{j,i}\beta_{i,j}\sigma_{z,j}^2 + o_P(N).
$$
\end{lemma}
\begin{proof}
Fix $j\in[p]$. Note that for all $i\in[N]$,
\begin{align*}
\C_{i,j} 
& = \beta_i^\top \E[v_tv_{t,j}] = \beta_{i,j}\E[v_{t,j}^2] + \beta_{i,\{j\}^c}^\top \E[v_{t,\{j\}^c}v_{t,j}] \\
& = \beta_{i,j}\sigma_{z,j}^2 + \beta_{i,j}\eta_{j,\{j\}^c}^\top\E[v_{t,\{j\}^c}v_{t,\{j\}^c}^\top]\eta_{t,\{j\}^c} + \beta_{i,\{j\}^c}^\top \E[v_{t,\{j\}^c}v_{t,\{j\}^c}^\top]\eta_{t,\{j\}^c},
\end{align*}
where the last equality follows from $v_{t,j} = \eta_{t,\{j\}^c}^\top v_{t,\{j\}^c} + z_{t,j}$ and \eqref{eq: eta definition}. Also, for all $i\in[N]$ and $l\in\{j\}^c$,
\begin{align*}
\C_{i,l}
& = \beta_i^\top \E[v_tv_{t,l}] = \beta_{i,j}\E[v_{t,j}v_{t,l}] + \beta_{i,\{j\}^c}^\top \E[v_{t,\{j\}^c}v_{t,l}] \\
& = \beta_{i,j}\eta_{j,\{j\}^c}^\top\E[v_{t,\{j\}^c}v_{t,l}] + \beta_{i,\{j\}^c}^\top \E[v_{t,\{j\}^c}v_{t,l}]
 \end{align*}
 by the same argument. Thus,
$
\C_{i,j} = \beta_{i,j}\sigma_{z,j}^2 + \sum_{l\in\{j\}^c}\C_{i,l}\eta_{j,l}, 
$
and so
\begin{align*}
R_j^\top R_j 
& = \sum_{i\in[N]}R_{j,i}\left(\C_{i,j} - \sum_{l\in\{j\}^c}\C_{i,l}\varphi_{j,l}\right)
& = \sum_{i\in[N]}R_{j,i}\left(\beta_{i,j}\sigma_{z,j}^2 + \sum_{l\in\{j\}^c}\C_{i,l}(\eta_{j,l} - \varphi_{j,l})\right).
\end{align*}
Therefore,
\begin{align*}
\left| R_j^\top R_j - \sum_{i\in[N]}R_{j,i}\beta_{i,j}\sigma_{z,j}^2 \right|
& \leq \Big|R_j^\top\C_{\{j\}^c}(\eta_{j,\{j\}^c} - \varphi_{j,\{j\}^c})\Big| \\
& \leq \Big( \|\eta_{j,\{j\}^c}\|_1 + \|\varphi_{j,\{j\}^c}\|_1 \Big)\| \C_{\{j\}^c}^\top R_j \|_{\infty} \\
& \lesssim (1+\sqrt{\bar s_{T}})N\sqrt{\log(Np)/T},
\end{align*}
where the last inequality follows from Assumptions \ref{as: first stage sparsity} and \ref{as: factors cool} and Lemmas \ref{lem: varphi bound} and \ref{lem: matrix inverse and l1 norm}. Combining this bound with Assumption \ref{as: growth conditions 2} yields the asserted claim.
\end{proof}

\begin{proof}[Proof of Theorem \ref{thm: asy normality}]
Fix $j\in[p]$ and, for brevity, write $\widetilde{\S}$, $\widehat\varphi$, and $\varphi$ instead of $\widetilde\S_j = \widehat{\S}_{D,j}\setminus\{j\}$, $\widehat\varphi_j$, and $\varphi_j$, respectively. Then
\begin{align}
\sqrt T(\widehat\psi_{D,j} - \psi_j) & = \frac{\sqrt T(\widehat{\C}_{\{j\}} - \widehat{\C}_{\widetilde{\S}}\widehat{\varphi}_{\widetilde{\S}})^{\top}(\widehat{\Y} - (\widehat{\C}_{\{j\}} - \widehat{\C}_{\widetilde{\S}}\widehat{\varphi}_{\widetilde{\S}})\psi_j )}{(\widehat{\C}_{\{j\}} - \widehat{\C}_{\widetilde{\S}}\widehat{\varphi}_{\widetilde{\S}})^{\top}(\widehat{\C}_{\{j\}} - \widehat{\C}_{\widetilde{\S}}\widehat{\varphi}_{\widetilde{\S}})} \nonumber\\
&  = \frac{\sqrt T(\widehat{\C}_{\{j\}} - \widehat{\C}_{\widetilde{\S}}\widehat{\varphi}_{\widetilde{\S}})^{\top}(\widehat{\Y} - \widehat{\C}_{\{j\}}\psi_j - \widehat{\C}_{\widetilde{\S}}\psi_{\widetilde{\S}})}{(\widehat{\C}_{\{j\}} - \widehat{\C}_{\widetilde{\S}}\widehat{\varphi}_{\widetilde{\S}})^{\top}(\widehat{\C}_{\{j\}} - \widehat{\C}_{\widetilde{\S}}\widehat{\varphi}_{\widetilde{\S}})}\label{eq: fwl equality}
\end{align}
where the first line follows from the Frisch-Waugh-Lowell theorem, and the second from observing that $(\widehat{\C}_{\{j\}} - \widehat{\C}_{\widetilde{\S}}\widehat{\varphi}_{\widetilde{\S}})^{\top}\widehat\C_{\widetilde{\S}} = \mathbf 0_{|\widetilde{\S}|}$. To derive the asymptotic distribution of the quantity on the right-hand side of \eqref{eq: fwl equality}, we start with a few preliminary bounds. First,
\begin{equation}\label{eq: prelim bound 1}
\| \widehat{\C}\widehat{\varphi} - \widehat{\C}\varphi \|_2^2 \lesssim_P N\| \widehat{\varphi} - \varphi \|_2^2 \lesssim_P N\bar s_{T}\log(Np)/T,
\end{equation}
where the first inequality follows from Lemma \ref{lem: empirical sparse eigenvalues} since $\|\widehat{\varphi} - \varphi\|_0 \leq |\widehat{\S}| + |\widehat{\S}_j| + |\S_{0,j}| \leq 3\bar s_{T}$ \wpas  by Assumptions \ref{as: estimator sparsity}, \ref{as: first stage sparsity}, and \ref{as: first stage estimation} and the second from Lemma \ref{lem: varphi estimation}. Second,
\begin{align}
&\| \psi[\widehat{\S}_{D, j}] - \psi[\S_0] \|_2^2 
\lesssim \|\psi[\widehat{\S}_{D, j}] - \psi \|_2^2 + \| \psi - \psi[\S_0] \|_2^2
 \lesssim \| \psi[\widehat{\S}] - \psi \|_2^2 + \| \psi - \psi[\S_0] \|_2^2 \nonumber\\
&\quad \lesssim \| \widehat{\psi} - \psi \|_2^2 + \| \psi - \psi[\S_0]  \|_2^2 
\lesssim_P \| \widehat{\psi} - \psi[\S_0] \|_2^2 + \| \psi - \psi[\S_0]  \|_2^2  
\lesssim_P \bar s_{T} \log(Np) / T, \label{eq: prelim bound 2}
\end{align}
where the first inequality follows from the triangle inequality, the second from the fact that $\widehat{\S} \subset \widehat{\S}_{D,j}$, the third from the fact that $\supp(\widehat\psi) = \widehat{\S}$, the fourth from the triangle inequality,  and the fifth from the intermediate bound in \eqref{eq: something we need} in the proof of Theorem \ref{thm: convergence rate} and Assumption \ref{as: sdf loadings}. Third,
\begin{equation}\label{eq: prelim bound 3}
\| \widehat{\C}_{\widehat{\S}_{D,j}} \psi_{\widehat{\S}_{D,j}} - \widehat{\C}_{\S_0}\psi_{\S_0} \|_2^2
\lesssim_P N\| \psi[\widehat{\S}_{D,j}] - \psi[\S_0] \|_2^2 \lesssim_P N\bar s_{T} \log(N p) / T,
\end{equation}
where the first inequality follows from Lemma \ref{lem: empirical sparse eigenvalues} since $\|\psi[\widehat{\S}_{D,j}] - \psi[\S_0]\|_0 \leq |\widehat{\S}| + |\widehat{\S}_j| + 1 + |\S_{0}| \leq 3\bar s_{T} + 1$ \wpas  by Assumptions \ref{as: sdf loadings}, \ref{as: estimator sparsity}, and \ref{as: first stage estimation}, and the second from \eqref{eq: prelim bound 2}. Fourth,
\begin{align}
\| \widehat{\C}_{\S_0}\psi_{\S_0} - \widehat{\C}\psi \|_2
& = \left\| \sum_{j\in\S_0^c} \widehat{\C}_{\{j\}}\psi_j \right\|_2
\leq \sum_{j\in\S_0^c} \| \widehat{\C}_{\{j\}}\psi_j \|_2 \nonumber\\
& \leq \max_{j\in\S_0^c}\|\widehat{\C}_{\{j\}}\|_2\sum_{j\in\S_0^c}|\psi_j|
\lesssim_P \sqrt{N\bar s_{T}\log(N p)/T},\label{eq: prelim bound 4}
\end{align}
where the first inequality follows from the triangle inequality and the third from Lemma \ref{lem: empirical sparse eigenvalues} and Assumption \ref{as: sdf loadings}. 

Thus, getting back to \eqref{eq: fwl equality}, we have
$\widehat{\C}_{\{j\}} - \widehat{\C}_{\widetilde{\S}}\widehat{\varphi}_{\widetilde{\S}} = \widehat{\C}_{\{j\}} - \C_{\{j\}} + \C\varphi + R_j - \widehat{\C}_{\widetilde S}\widehat\varphi_{\widetilde{\S}}$, where $\| \widehat{\C}_{\{j\}} - \C_{\{j\}}  \|_2^2\lesssim_P N\log(Np)/T$ by Lemma \ref{lem: individual covariance bound} and
\begin{align*}
\|\widehat{\C}_{\widetilde S}\widehat\varphi_{\widetilde{\S}} - \C\varphi\|_2^2
& = \|\widehat{\C}\widehat{\varphi} - \C\varphi\|_2^2
\lesssim \| \widehat{\C}\widehat{\varphi} - \widehat{\C}\varphi \|_2^2 + \|\widehat{\C}\varphi - \C\varphi \|_2^2 \lesssim_P N \bar s_{T}\log(Np)/T
\end{align*}
by the triangle inequality, \eqref{eq: prelim bound 1}, Assumption \ref{as: first stage sparsity}, and Lemma \ref{lem: c-varphi bound}. In addition, $\widehat{\Y} - \widehat{\C}_{\{j\}}\psi_j - \widehat{\C}_{\widetilde{\S}}\psi_{\widetilde\S} = \widehat{\Y} - \Y + \C\psi - \widehat{\C}_{\widehat{\S}_{D,j}}\psi_{\widehat{\S}_{D,j}}$, where $\|\widehat{\Y} - \Y\|_2^2 \lesssim_P N\log(Np)T$ by Assumption \ref{as: ulln} and
\begin{align*}
\| \widehat{\C}_{\widehat{\S}_{D,j}}\psi_{\widehat{\S}_{D,j}} - \C\psi \|_2^2
& \lesssim \| \widehat{\C}_{\widehat{\S}_{D,j}} \psi_{\widehat{\S}_{D,j}} - \widehat{\C}_{\S_0}\psi_{\S_0} \|_2^2 + \| \widehat{\C}_{\S_0}\psi_{\S_0} - \widehat{\C}\psi\|_2^2 + \|\widehat{\C}\psi -  \C\psi \|_2^2 \\
& \lesssim_P N\bar s_{T}\log(Np)/T
\end{align*}
by the triangle inequality, \eqref{eq: prelim bound 3}, \eqref{eq: prelim bound 4}, and Lemma \ref{lem: c-psi bound}. Moreover,
$$
|\sqrt T R_j^{\top}\widehat{\C}_{\widehat{\S}_{D,j}^c}\psi_{\widehat{\S}_{D,j}^c}| \lesssim N\sqrt{\log(Np)}\|\psi_{\widehat{\S}_{D,j}^c}\|_1 \leq N\sqrt{\log(Np)}\|\widehat\psi - \psi\|_1 = o_P(N),
$$
where the first inequality follows from Assumption \ref{as: first stage sparsity}, the second from the fact that $\widehat{\S}\subset \widehat{\S}_{D,j}$, and the third form Theorem \ref{thm: convergence rate} and Assumption \ref{as: growth conditions 2}.
 Combining these bounds and using Lemma \ref{lem: r upper bound} and Assumption \ref{as: growth conditions 2}, it follows that
$$
\sqrt T(\widehat{\psi}_{D,j} - \psi_j) = \frac{\sqrt T R_j^{\top}(\widehat {\Y} - \widehat{\C}\psi) + o_P(N)}{R_j^{\top}R_j + o_P(N)}.
$$
Next, observe that $R_j^{\top}(\widehat{\Y} - \widehat{\C}\psi) = \mathrm I_1 + \mathrm I_2 +  \mathrm I_3 + \mathrm I_4 $, where
\begin{align*}
\mathrm I_1 &= \sum_{i\in[N]} R_{j,i} \beta_i^{\top}\left( \E[v_tv_t^{\top}] - \frac{1}{T}\sum_{t=1}^T f_t(f_t - \bar f)^{\top} \right)\psi,\\
\mathrm I_2 & = \frac{1}{T}\sum_{i\in[N]}\sum_{t\in[T]}R_{j,i}\beta_i^{\top}(f_t - \E[f_t]),\quad \mathrm I_3 = \frac{1}{T}\sum_{i\in[N]}\sum_{t\in[T]}R_{j,i}\varepsilon_{i,t}, \\
\mathrm I_4 & = -\frac{1}{T}\sum_{i\in[N]}\sum_{t\in[T]}R_{j,i}\varepsilon_{i,t}(f_t - \bar f)^{\top}\psi.
\end{align*}
Here,
\begin{align*}
\mathrm I_4 &  = -\frac{1}{T}\sum_{i\in[N]}\sum_{t\in[T]}R_{j,i}\varepsilon_{i,t}v_t^{\top}\psi + o_P(N/\sqrt T)
\end{align*}
by Lemma \ref{lem: r upper bound} and Assumptions \ref{as: ulln} and \ref{as: growth conditions 2}. Thus,
$$
\mathrm I_3 + \mathrm I_4 = \frac{1}{T}\sum_{i\in[N]}\sum_{t\in[T]} R_{j,i}\varepsilon_{i,t}m_t = o_P(N/\sqrt T)
$$
by Assumption \ref{as: extra lln}. Similarly,
$$
\mathrm I_1 + \mathrm I_2 = \frac{1}{T}\sum_{i\in[N]}\sum_{t\in[T]}R_{j,i}\beta_i^{\top}\Big(v_t m_t - \E[v_t m_t] \Big)+ o_P(N/\sqrt T),
$$
again by Lemma \ref{lem: r upper bound} and Assumptions \ref{as: ulln} and \ref{as: growth conditions 2}. Thus, substituting $v_{t,j} = \eta_{j,\{j\}^c}v_{t,\{j\}^c} + z_{t,j}$ and using Lemma \ref{eq: phi definition equivalent}
,\begin{align*}
\mathrm I_1 + \mathrm I_2 
& = \frac{1}{T}\sum_{i\in[N]}\sum_{t\in[T]}R_{j,i}\beta_{i,j} (z_{t,j}m_t - \E[z_{t,j}m_t]) \\
& \quad + \frac{N}{T}\sum_{t\in[T]} \phi_{j,\{j\}^c}^{\top} (v_{t,\{j\}^c} m_t - \E[v_{t,\{j\}^c} m_t]) + o_P(N/\sqrt T) \\
& = \frac{1}{T}\sum_{i\in[N]}\sum_{t\in[T]}R_{j,i}\beta_{i,j} (z_{t,j}m_t - \E[z_{t,j}m_t]) + o_P(N/\sqrt T),
\end{align*}
where the second equality follows from Lemma \ref{lem: cool sparsity bound} and Assumptions \ref{as: ulln}, \ref{as: additional sparsity wtf}, and \ref{as: growth conditions 2}. Moreover, $R_j^{\top}R_j = \sum_{i\in[N]}R_{j,i}\beta_{i,j}\sigma_{z,j}^2 + o_P(N)$ by Lemma \ref{lem: hopefully final}.

Combining presented bounds and using Lemma \ref{lem: r upper bound}, it follows that
$$
\sqrt T(\widehat\psi_{D,j} - \psi_j) = \frac{1}{\sqrt T}\sum_{t\in[T]}\frac{z_{t,j}m_t - \E[z_{t,j}m_t]}{\sigma_{z, j}^2} + o_P(1),
$$
which implies the asserted claim by invoking Assumption \ref{as: woohoo}.
\end{proof}

\section{Technical Lemmas}\label{sec: technical lemmas}

\begin{lemma}
\label{lem:principal-submatrix-inverse-bound}
Let $A = (A_{i,j})_{i,j\in[n]}\in \R^{n\times n}$ be a symmetric matrix such that $\lambda_{\max}(A) \leq C_1$, $\lambda_{\min}(A) \geq 1/C_1$, and $ \|A^{-1}\|_{\infty,1} \le C_1$ for some constant $C_1\geq 1$.  For any $k\in[n]$, let $B = (A_{i, j})_{i, j\in[n]\setminus{k}}$ be a sub-matrix of $A$. Then $B$ is invertible and satisfies 
$
\|B^{-1}\|_{\infty,1} \le C,
$
where $C$ is a constant depending only on $C_1$.
\end{lemma}

\begin{proof}
Without loss of generality, we will assume that $k=n$, so that
\begin{equation}\label{eq: block form}
  A \;=\;
  \begin{pmatrix}
    B & b \\
    b^T & a
  \end{pmatrix},
\end{equation}
where $B = (A_{i,j})_{i,j\in[n-1]}$, $b = (A_{i,n})_{i\in[n-1]}$, and $a = A_{n,n}$. Since $\lambda_{\min}(B)\geq\lambda_{\min}(A)>0$, the matrix $B$ is invertible, yielding the first asserted claim.

To prove the second asserted claim, note that
$$
\|B^{-1}\|_{\infty,1} = \sup_{y\in\R^{n-1}\colon \|y\|_{\infty}= 1}\|B^{-1}y\|_{\infty}.
$$
Further, fix any $y\in\R^{n-1}$ such that $\|y\|_{\infty} = 1$ and let $x\in\R^{n-1}$ and $z\in\R$ be such that
\begin{equation}\label{eq: original matrix form}
  A 
  \begin{pmatrix}
    x \\
    z
  \end{pmatrix}
  \;=\;
  \begin{pmatrix}
    y \\
    0
  \end{pmatrix}.
\end{equation}
Note that such $x$ and $z$ exist because $A$ is invertible. Also,
\[
  \left\|\begin{pmatrix} x \\ z \end{pmatrix}\right\|_\infty
  \;\le\; \|A^{-1}\|_{\infty,1}\;\left\|\begin{pmatrix}y\\0\end{pmatrix}\right\|_{\infty}
  \;\le\; C_1 \,\|y\|_{\infty}
\]
by assumptions of the lemma, and so $\|x\|_\infty \le C_1 \|y\|_\infty$ and $|z| \le C_1 \|y\|_\infty$. In addition, substituting \eqref{eq: block form} into \eqref{eq: original matrix form}, it follows that $Bx + b z =y$, and so $B^{-1}y = x + B^{-1}bz$. Thus, by the triangle inequality,
\begin{equation}\label{eq: key chain matrix bound}
\|B^{-1}y\|_{\infty} \leq \|x\|_{\infty} + \|B^{-1}b\|_{\infty}\times |z| \leq C_1(1 + \|B^{-1}b\|_{\infty})\|y\|_{\infty}.
\end{equation}
It thus remains to bound $\|B^{-1}b\|_{\infty}$. To do so, observe that it follows from \eqref{eq: block form} that
$$
A
  \begin{pmatrix}
    B^{-1}b \\
    0
  \end{pmatrix}
  =
  \begin{pmatrix}
    b \\
    b^{\top}B^{-1}b
  \end{pmatrix}.
$$
Thus,
\begin{align*}
\|B^{-1}b\|_{\infty} &\leq \|A^{-1}\|_{\infty,1} \Big(\|b\|_{\infty} + |b^{\top}B^{-1}b|\Big) \\
& \leq  \|A^{-1}\|_{\infty,1} \left( \|b\|_2 + \frac{\|b\|_2^2}{\lambda_{\min}(A)} \right)
\leq \|A^{-1}\|_{\infty,1}\left(\lambda_{\max}(A) + \frac{\lambda_{\max}(A)^2}{\lambda_{\min}(A)}\right),
\end{align*}
where the first inequality follows from the definition of the $\ell_{\infty,1}$-norm, the second from noting that $\|b\|_{\infty}\leq \|b\|_2$ and $\|B^{-1}\|_2 = 1/\lambda_{\min}(B)\leq 1/\lambda_{\min}(A)$, and the third from noting that $\|b\|_2 \leq \|A\|_2$ by the definition of the $\ell_2$-norm. Combining this bound with \eqref{eq: key chain matrix bound} and using $\|y\|_{\infty}=1$ yields the second asserted claim and completes the proof of the lemma.
\end{proof}

\begin{lemma}\label{lem: matrix inverse and l1 norm}
Let $x = (x_1,\dots,x_p)^\top\in R^p$ be a random vector such that the matrix $\E[xx^\top]$ is non-singular and denote $\Omega=(\E[xx^\top])^{-1}$. For any $j\in[p]$, consider the least-squares projection
$$
x_j = x_{\{j\}^c}^\top\gamma_j + \epsilon_j,\quad \E[x_{\{j\}^c}\epsilon_j] = \mathbf 0_{p-1}.
$$
Then $\|\gamma_j\|_1 = \E[\epsilon_j^2] \sum_{l\in[p]}|\Omega_{lj}| - 1$.
\end{lemma}
\begin{proof}
Without loss of generality, we will assume that $j = p$. Then
$$
\E[xx^\top] = \E\left[\left(\begin{array}{c}
x_{\{p\}^{c}}\\
x_{p}
\end{array}\right)\left(\begin{array}{c}
x_{\{p\}^{c}}\\
x_{p}
\end{array}\right)^{\top}\right]
 = \left(\begin{array}{cc}
A & B\\
C & D
\end{array}\right),
$$
where $A = \E[x_{\{p\}^c}x_{\{p\}^c}^\top]$, $B = \E[x_{\{p\}^c}x_p]$, $C = \E[x_px_{\{p\}^c}^\top]$, and $D = \E[x_p^2]$. Then
$$
\Omega = \left(\begin{array}{cc}
A^{-1}+A^{-1}BCA^{-1}/E & -A^{-1}B/E\\
-CA^{-1}/E & 1/E
\end{array}\right),
$$
where $E = D - CA^{-1} B$. Therefore, $\sum_{l\in[p]}|\Omega_{lp}| = (\|A^{-1}B\|_1 + 1)/E$. However, $A^{-1}B = \gamma_p$ and $E = \E[\epsilon_p^2]$. The asserted claim follows.
\end{proof}


\section{Optimal Stopping Threshold Selection}\label{sec:optimal_stopping}

In the baseline FS-FMB procedure with cross-validation (Table~\ref{tab:Forward-Regression-Method-CV}), the stopping threshold is set exogenously to 1 percent. Here, we select the stopping threshold out-of-sample via nested cross-validation, ensuring that the threshold choice does not exploit information from the assets used for model evaluation.

We consider a grid of candidate thresholds $\tau \in \{0.01\%, 0.25\%, 0.5\%, 1\%, 1.5\%, 2\%\}$. For each candidate, we employ a nested cross-validation design with $K_{\text{out}} = 5$ outer folds and $K_{\text{in}} = 5$ inner folds. For each outer fold $j$, the inner cross-validation runs the full FS-FMB procedure on the outer training assets---selecting higher-order factors and determining when to stop---using threshold $\tau$. After the inner procedure terminates, we estimate the risk prices from a cross-sectional regression on all outer training assets and evaluate the model on the held-out outer test assets, computing the adjusted R-squared. The optimal threshold is the one that maximizes the average outer-fold adjusted R-squared across the $K_{\text{out}}$ folds:
\begin{equation}\label{eq:optimal_tau}
    \tau^* = \arg\max_{\tau} \; \frac{1}{K_{\text{out}}} \sum_{j=1}^{K_{\text{out}}} \bar{R}^2_{\text{out},j}(\tau).
\end{equation}
Crucially, the outer test assets are never used during factor selection or threshold comparison, so the selection of $\tau^*$ is entirely out-of-sample.

Table~\ref{tab:threshold_selection} reports the average outer-fold adjusted R-squared for each candidate threshold. The nested cross-validation selects $\tau^* = 0.5\%$, a more permissive threshold than the baseline value of 1\%.

\begin{table}[H]
\caption{Stopping Threshold Selection\label{tab:threshold_selection}}

\medskip{}

\begin{minipage}{\textwidth}
\small\singlespacing
This table reports the average outer-fold adjusted R-squared from the nested cross-validation for each candidate stopping threshold $\tau$. The threshold that maximizes the out-of-sample adjusted R-squared is selected as optimal.
\end{minipage}

\medskip{}

\centering{}%
\begin{tabular}{lc}
\hline
$\tau$ & Avg. outer-fold Adj. R-squared\tabularnewline
\hline
0.01\% & 0.414\tabularnewline
0.25\% & 0.414\tabularnewline
\textbf{0.5\%} & \textbf{0.422}\tabularnewline
1\% & 0.39\tabularnewline
1.5\% & 0.364\tabularnewline
2\% & 0.357\tabularnewline
\hline
\end{tabular}
\end{table}

Table~\ref{tab:Forward-Regression-Method-CV-OptimalStop} reports the results from the FS-FMB procedure with 5-fold cross-validation using the optimally selected threshold $\tau^* = 0.5\%$. Although $\tau^*$ is more permissive than the baseline value of 1\%, the procedure selects the same 7 higher-order factors and produces identical results. This is because, in the baseline specification reported in Table~\ref{tab:Forward-Regression-Method-CV}, the stopping rule is also evaluated on the cross-validated adjusted R-squared rather than on the in-sample R-squared. Since the cross-validated marginal improvements are smaller than their in-sample counterparts---reflecting the well-known shrinkage of out-of-sample fit relative to in-sample fit---the 8th factor (HML2{*}Mkt-RF), which adds 1.4 pp in-sample, contributes less than 1 pp in cross-validated R-squared and is therefore excluded under both thresholds. This confirms that the baseline model selection is robust and not an artifact of the particular choice of stopping rule: even when the threshold is chosen optimally and entirely out-of-sample, the selected higher-order factor model is largely unchanged.

\begin{table}[H]
\caption{FS-FMB Method with Cross-Validation and Optimal Stopping Threshold\label{tab:Forward-Regression-Method-CV-OptimalStop}}

\medskip{}

\begin{minipage}{\textwidth}
\small\singlespacing
This table reports the results from the second step of the FMB regression procedure for the FS-FMB procedure with 5-fold cross-validation. The stopping threshold $\tau^* = 0.5\%$ is selected out-of-sample via nested cross-validation (Table~\ref{tab:threshold_selection}). At each step, we include one additional higher-order factor which maximizes the adjusted cross-sectional R-squared computed with cross-validation using 5-fold cross-validation. The procedure stops when the improvement in the cross-validated R-squared in a given step is smaller than $\tau^*$. For each step, we report the selected factor ($h_{j}$), the adjusted cross-sectional R-squared, the adjusted cross-validated (CV) cross-sectional R-squared, the estimate for the intercept ($\alpha$), and the associated t-statistic. The t-statistic is based on Newey-West corrected standard errors.
\end{minipage}

\medskip{}

\centering{}%
\begin{tabular}{lcccccc}
\hline
Step & $h_{j}$ & Adj. R-squared & Adj. CV R-squared & & $\alpha$ & t-stat ($\alpha$)\tabularnewline
\hline
1 & SMB2 & 0.41 & 0.328 & & 0.003 & 1.821\tabularnewline
2 & SMB2{*}Mom & 0.467 & 0.382 & & 0.004 & 2.02\tabularnewline
3 & Mom2{*}RMW & 0.495 & 0.4 & & 0.003 & 1.324\tabularnewline
4 & Mkt-RF2 & 0.529 & 0.426 & & 0.003 & 1.769\tabularnewline
5 & Mkt-RF2{*}RMW & 0.552 & 0.446 & & 0.003 & 1.636\tabularnewline
6 & Mkt-RF{*}SMB & 0.572 & 0.458 & & 0.004 & 2.161\tabularnewline
7 & Mom2{*}HML & 0.582 & 0.471 & & 0.004 & 2.145\tabularnewline
\hline
\end{tabular}
\end{table}

\section{Additional Figures \& Tables}\label{sec: additional figures}

\begin{figure}[H]
\caption{Distribution of Higher-Order Factors\label{fig:Appendix_Distribution-of-Higher-Order-Factors}}

\medskip{}

\begin{minipage}{\textwidth}
\small\singlespacing

This figure plots the distribution of the higher-order factors grouped
in pairwise interactions of degree 2 (Panel A), pairwise interactions
of degree 3 (Panel B), powers of degree 2 (Panel C) and powers of
degree 3 (Panel D). The distribution is clipped at 4 standard deviations
from the mean, so that any value below the lower bound is replaced
by the lower bound and any value above the upper bound is replaced
by the upper bound.

\end{minipage}

\medskip{}

\begin{centering}
\subfloat[]{
\centering{}\includegraphics[scale=0.3]{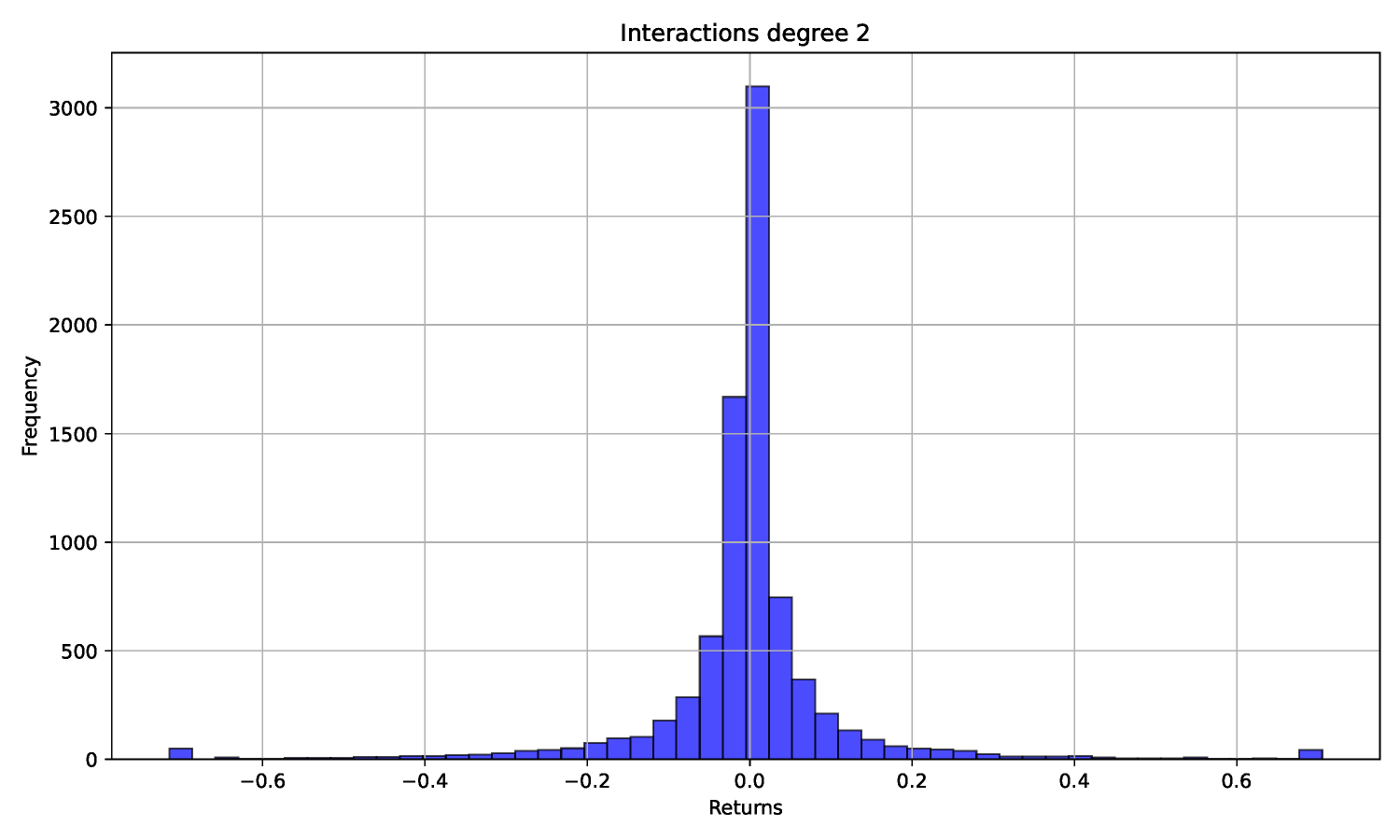}}\subfloat[]{
\centering{}\includegraphics[scale=0.3]{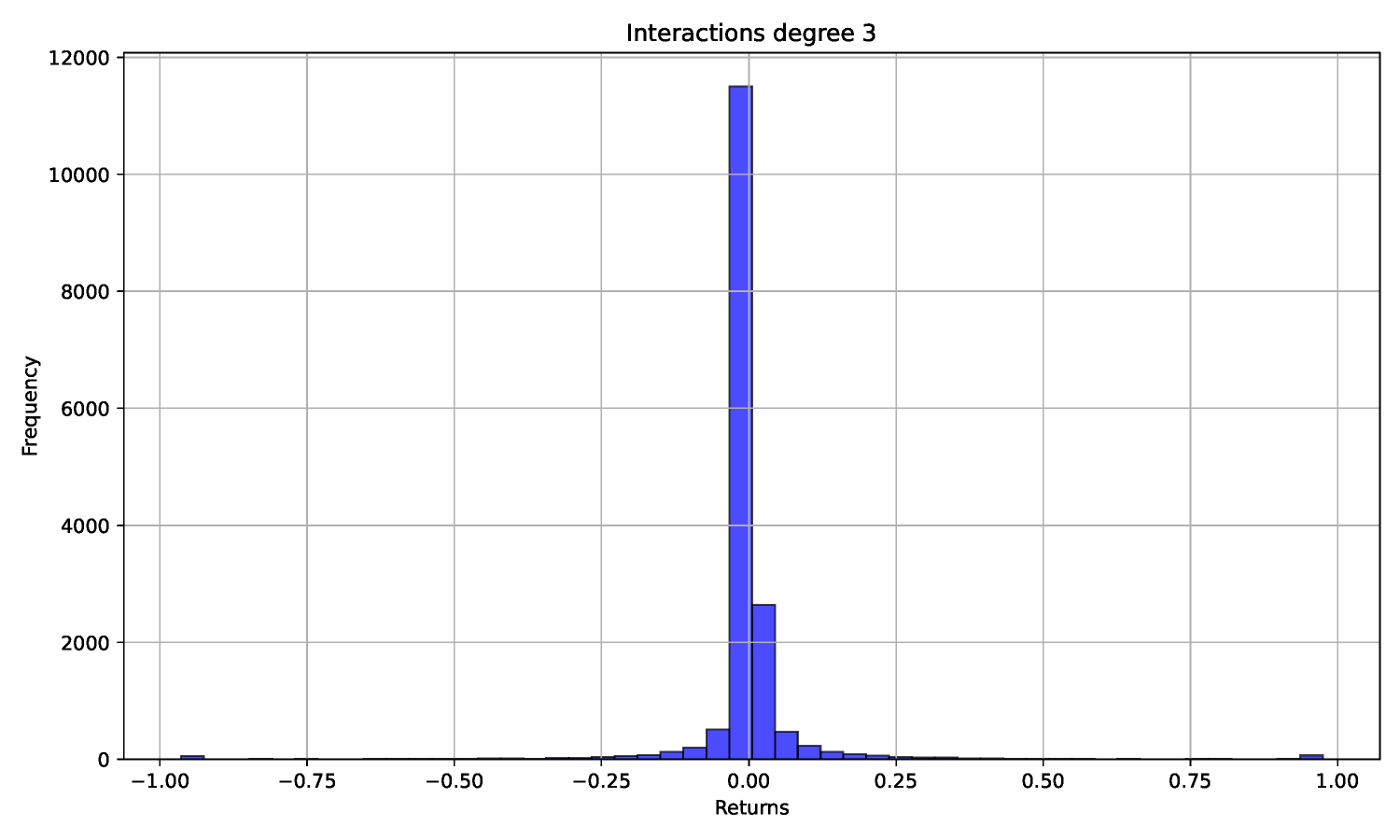}}
\par\end{centering}
\centering{}\subfloat[]{
\centering{}\includegraphics[scale=0.3]{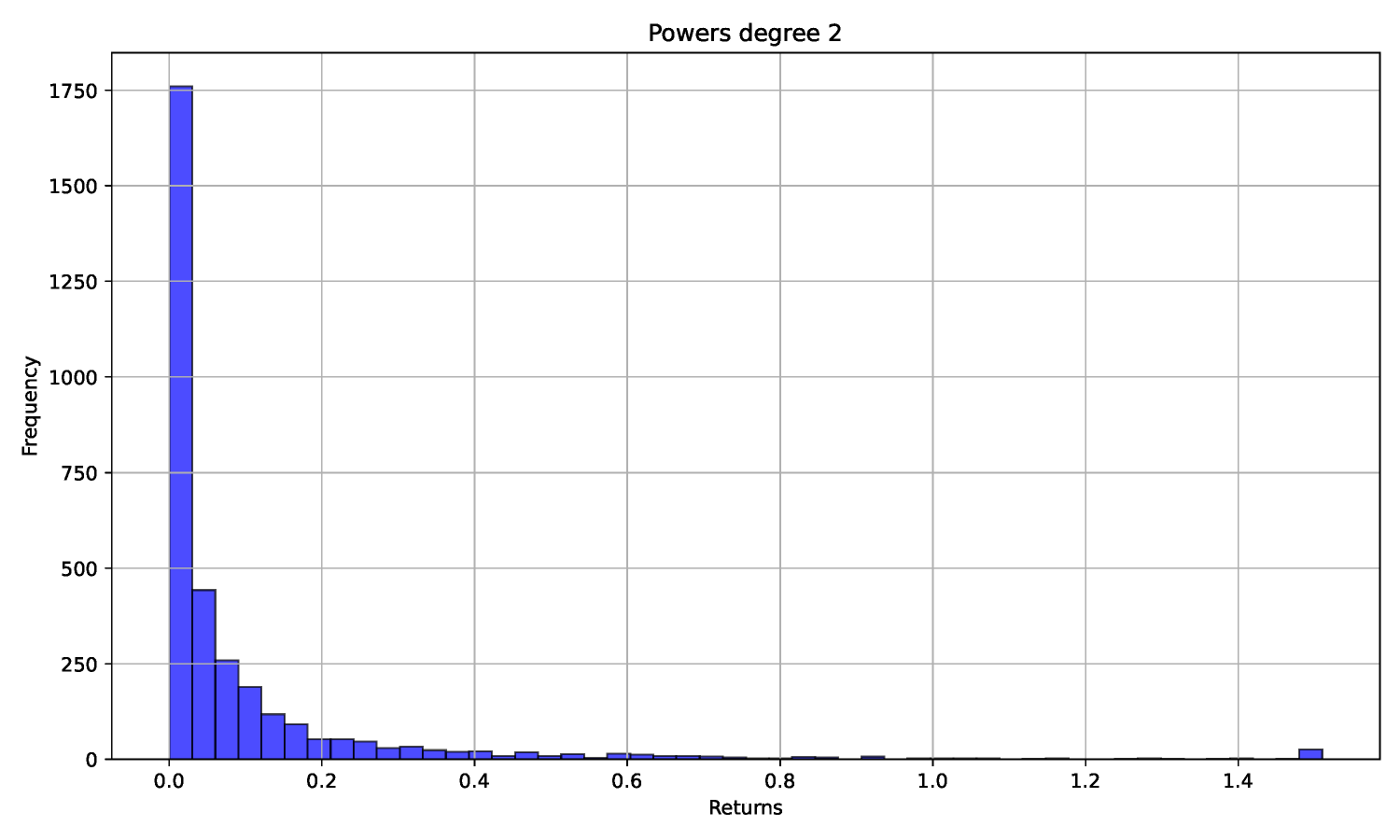}}\subfloat[]{
\centering{}\includegraphics[scale=0.3]{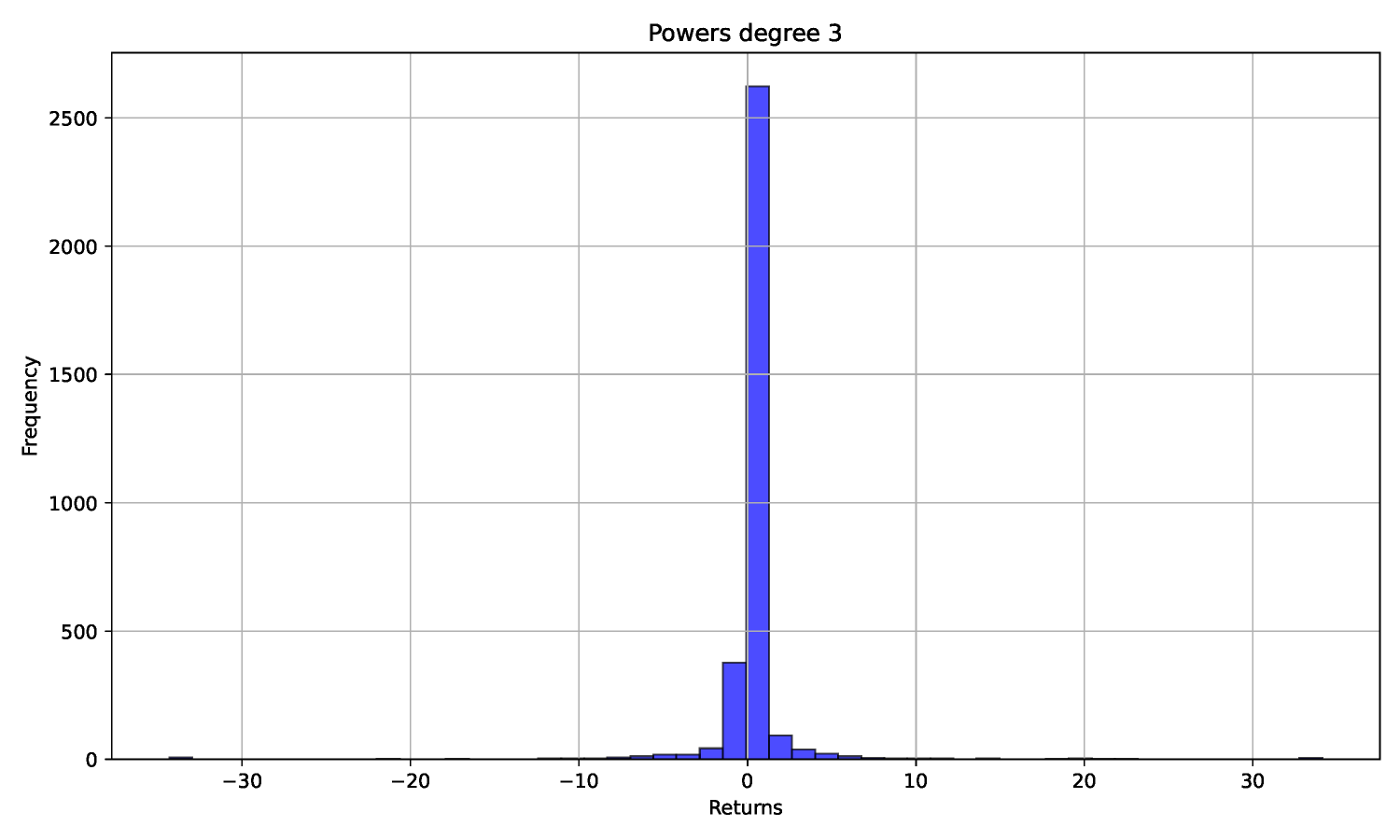}}
\end{figure}


\begin{figure}[H]

\caption{$R^{2}$ of Factor Mimicking Regressions\label{fig:FactorMimicking_R2}}

\medskip{}
\begin{minipage}{\textwidth}
\small\singlespacing
This figure reports, for each of the higher-order factors selected by
the FS-FMB procedure, the adjusted R-squared of factor mimicking regressions.
The factor mimicking regressions are time-series OLS regressions of
higher-order factor returns on the zoo factors from \citet{jensen2023there}.
\end{minipage}
\medskip{}

\centering{}\includegraphics[scale=0.5]{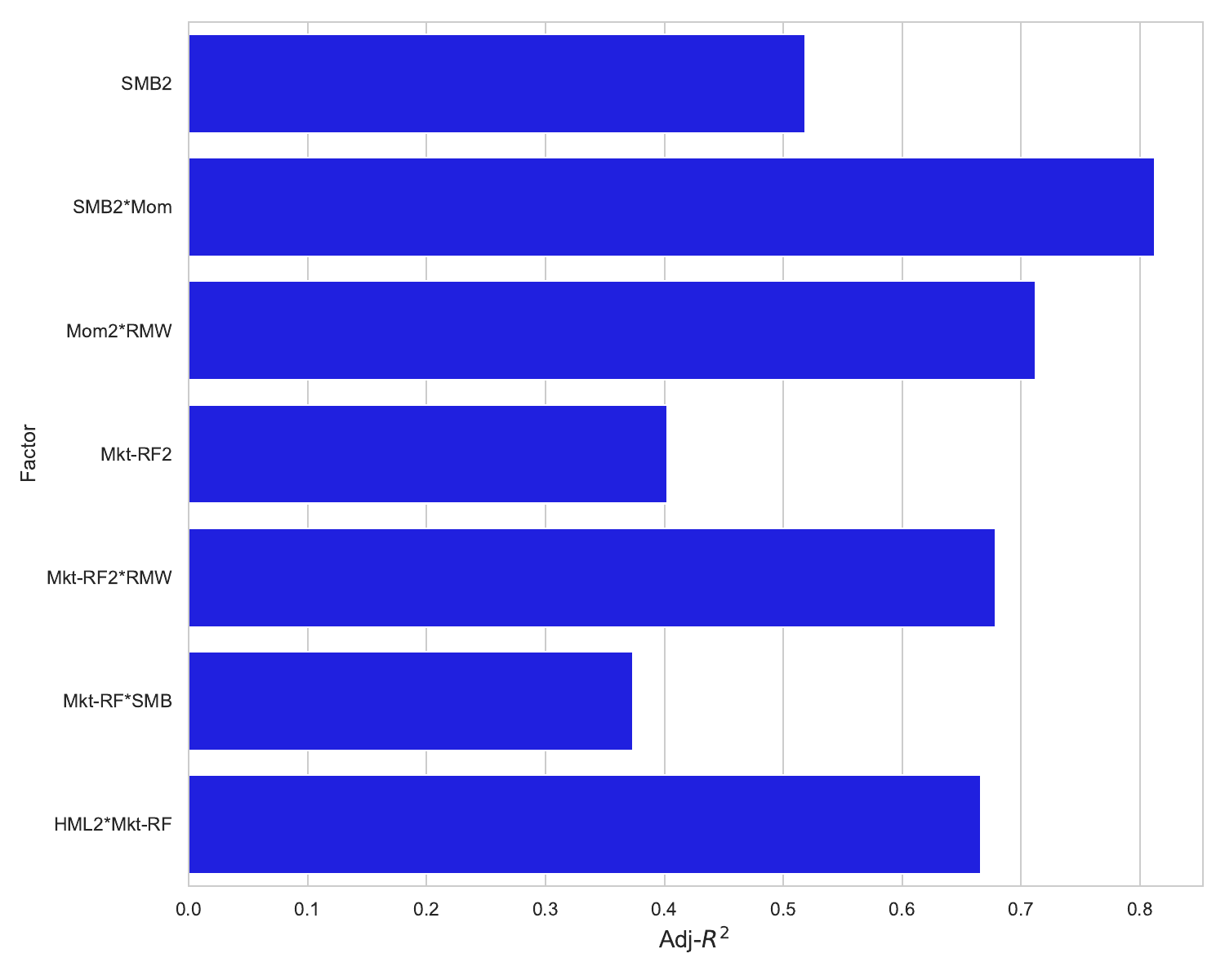}
\end{figure}


\begin{table}[H]
\caption{Cross-Sectional Performance Restricted Models\label{tab:Cross-sectional-performance-Restricted-Models}}

\medskip{}

\begin{minipage}{\textwidth}
\small\singlespacing

This table reports the results from the second step of the FMB procedure,
imposing the restriction that the risk prices of the tradable factors
must equal their sample means. The tradable factors are the FF5M
factors. The higher-order factors are those selected by the forward
selection FMB procedure (see Table \ref{tab:Cross-sectional-performance}).
The table reports the adjusted cross-sectional R-squared, the estimate
for the intercept ($\alpha$) and associated t-statistic. The t-statistic
is based on Newey-West corrected standard errors.
\end{minipage}

\medskip{}

\centering{}%
\begin{tabular}{llccc}
\hline
\# & Model & Adj. R-squared & $\alpha$ & t-stat ($\alpha$)\tabularnewline
\hline
1 & CAPM & -0.439 & 0.001 & 0.28\tabularnewline
2 & FF3 & -0.371 & 0.001 & 0.334\tabularnewline
3 & FF5 & -0.111 & 0.004 & 2.263\tabularnewline
4 & FF5M & 0.022 & 0.006 & 2.839\tabularnewline
5 & Higher-Order & 0.405 & 0.004 & 2.663\tabularnewline
\hline 
\end{tabular}
\end{table}


\begin{table}[H]
\caption{Comparison of R-squareds with Random Train and Test Sample\label{tab:Comparison-of-R-squareds-Robust}}
\medskip{}
\begin{minipage}{\textwidth}
\small\singlespacing

This table reports the in-sample and out-of-sample R-squareds in the
cross-sectional regression step of the Fama-MacBeth method for different
models. We randomly generate training and test samples of equal length
by randomly drawing monthly returns without replacement, and estimate
for each asset the covariance with respect to the factors and the
SDF loadings on the training sample, and compare actual mean returns
with model predictions in the test sample. For the test sample, we
report R-squareds after re-centering the predicted values from
the model. For the FF5M + Higher-Order model, we consider the higher-order
factors selected by the forward selection FMB procedure (see Table \ref{tab:Cross-sectional-performance}).
\end{minipage}
\medskip{}

\centering{}%
\begin{tabular}{c>{\centering}p{2cm}>{\centering}p{2cm}>{\centering}p{2cm}>{\centering}p{2cm}>{\centering}p{2cm}}
\hline 
 & (1) & (2) & (3) & (4) & (5)\tabularnewline
 & CAPM & FF3 & FF5 & FF5M & Higher-Order\tabularnewline
\hline 
$R_{train}^{2}$ & 0.06 & 0.167 & 0.274 & 0.316 & 0.498\tabularnewline
$R_{test}^{2}$ & 0.063 & 0.04 & 0.076 & 0.085 & 0.14\tabularnewline
\hline 
\end{tabular}
\end{table}

\begin{table}[H]
\caption{Loadings of Alternative Zoo Factors\label{tab:Significant-Loadings-Chen-Zimmermann}}
\medskip{}
\begin{minipage}{\textwidth}
\small\singlespacing
This table reports the fraction of zoo factors with loading significantly
different from zero at the 5\% confidence level on each of the FF5M
factors and higher-order factors selected by the FS-FMB procedure. Significance is evaluated by t-statistics based
on Newey-West corrected standard errors. The zoo factors are 159 factors
constructed in \citet{ChenZimmermann2021}.
\end{minipage}
\medskip{}

\centering{}%
\begin{tabular}{lcc}
\hline 
Factor & Frac Sig 5\% & Frac Sig 5\%\tabularnewline
\hline 
Mkt-RF & 0.478 & 0.509\tabularnewline
SMB & 0.572 & 0.579\tabularnewline
HML & 0.459 & 0.421\tabularnewline
RMW & 0.572 & 0.553\tabularnewline
CMA & 0.428 & 0.491\tabularnewline
Mom & 0.535 & 0.421\tabularnewline
SMB2 &  & 0.365\tabularnewline
SMB2{*}Mom &  & 0.478\tabularnewline
Mom2{*}RMW &  & 0.358\tabularnewline
Mkt-RF2 &  & 0.302\tabularnewline
Mkt-RF2{*}RMW &  & 0.208\tabularnewline
Mkt-RF{*}SMB &  & 0.302\tabularnewline
HML2{*}Mkt-RF &  & 0.302\tabularnewline
\hline 
\# zoo factors & 159 & 159\tabularnewline
\hline 
\end{tabular}
\end{table}


\begin{table}[H]
\caption{Forward Selection FMB procedure with Alternative
Higher-Order Factors\label{tab:Forward-Regression-Method_Robust}}

\medskip{}

\begin{minipage}{\textwidth}
\scriptsize

This table reports the results of the estimation of the FS-FMB procedure
for alternative sets of higher-order factors. Panel A considers the
FF5M factors and all the higher-order factors of degree 2 (i.e., $f_{i}\times f_{j}$
and $f_{i}^{2}$). Panel B considers the FF5M factors and all the
higher-order factors up to degree 4 (i.e., $f_{i}\times f_{j}$, $f_{i}^{2}\times f_{j}$,
$f_{i}^{2}\times f_{j}^{2}$, $f_{i}^{3}\times f_{j}$, $f_{i}^{2}$,
$f_{i}^{3}$, and $f_{i}^{4}$). Panel C considers the FF5M factors
and all the pairwise interactions of degree 2 and 3 (i.e., $f_{i}\times f_{j}$
and $f_{i}^{2}\times f_{j}$). Panel D considers the FF5M factors and
all the powers of degree 2 and 3 (i.e., $f_{i}^{2}$ and $f_{i}^{3}$).
For each step, we report the selected factor ($h_{j}$), the adjusted cross-sectional
R-squared, the estimate for the intercept ($\alpha$) and associated
t-statistic. The t-statistic is based on Newey-West corrected standard
errors.

\end{minipage}

\medskip{}

\centering{}%
\begin{tabular}{lcccc}
\hline 
Step & $h_{j}$ & Adj. R-squared & $\alpha$ & t-stat ($\alpha$)\tabularnewline
\hline 
\multicolumn{5}{c}{\textbf{Panel A: higher-order degree 2}}\tabularnewline
\hline 
1 & SMB2 & 0.41 & 0.003 & 1.821\tabularnewline
2 & HML{*}Mom & 0.464 & 0.004 & 2.033\tabularnewline
3 & Mkt-RF2 & 0.503 & 0.005 & 2.738\tabularnewline
4 & SMB{*}Mom & 0.517 & 0.005 & 2.692\tabularnewline
\hline 
\multicolumn{5}{c}{\textbf{Panel B: higher-order degree up to 4}}\tabularnewline
\hline 
1 & SMB2{*}CMA2 & 0.423 & 0.005 & 2.581\tabularnewline
2 & Mkt-RF2{*}SMB2 & 0.464 & 0.004 & 2.329\tabularnewline
3 & Mom3{*}CMA & 0.485 & 0.004 & 2.062\tabularnewline
4 & Mkt-RF2{*}RMW & 0.498 & 0.004 & 1.931\tabularnewline
5 & HML2 & 0.528 & 0.003 & 1.589\tabularnewline
6 & Mkt-RF2 & 0.547 & 0.004 & 2.271\tabularnewline
7 & CMA2{*}Mom & 0.57 & 0.004 & 2.13\tabularnewline
8 & CMA2{*}SMB & 0.587 & 0.004 & 2.074\tabularnewline
9 & Mkt-RF2{*}Mom & 0.6 & 0.004 & 2.132\tabularnewline
\hline 
\multicolumn{5}{c}{\textbf{Panel C: interactions degree 2 and 3}}\tabularnewline
\hline 
1 & SMB2{*}Mkt-RF & 0.394 & 0.003 & 1.456\tabularnewline
2 & CMA2{*}SMB & 0.437 & 0.003 & 1.429\tabularnewline
3 & Mkt-RF2{*}SMB & 0.466 & 0.004 & 2.023\tabularnewline
4 & CMA2{*}RMW & 0.487 & 0.004 & 2.084\tabularnewline
5 & HML{*}CMA & 0.509 & 0.004 & 2.321\tabularnewline
6 & CMA2{*}Mkt-RF & 0.527 & 0.004 & 2.248\tabularnewline
7 & CMA2{*}Mom & 0.542 & 0.004 & 2.044\tabularnewline
\hline 
\multicolumn{5}{c}{\textbf{Panel D: powers degree 2 and 3}}\tabularnewline
\hline 
1 & SMB2 & 0.41 & 0.003 & 1.821\tabularnewline
2 & Mkt-RF2 & 0.444 & 0.005 & 2.478\tabularnewline
3 & RMW3 & 0.467 & 0.005 & 2.807\tabularnewline
\hline 
\end{tabular}
\end{table}


\begin{table}[H]
\caption{Forward Selection FMB Procedure Starting
from the CAPM\label{tab:Forward-Regression-Method_Robust-1}}

\medskip{}

\begin{minipage}{\textwidth}
\small\singlespacing

This table reports the results from the second step of the FMB method
with selection by the FS-FMB procedure starting from the CAPM. At each
step, we augment the CAPM by one additional factor selected from the
remaining FF5M factors and the higher-order factor ($h_{j}$) which
maximizes the adjusted cross-sectional R-squared. The method stops
when the improvement in the R-squared in a given step is smaller than
1 pp. The table reports the adjusted cross-sectional R-squared, the estimate for the
intercept ($\alpha$), and the associated t-statistic. The t-statistic
is based on Newey-West corrected standard errors.

\end{minipage}

\medskip{}

\centering{}%
\begin{tabular}{lcccc}
\hline 
Step & $h_{j}$ & Adj. R-squared & $\alpha$ & t-stat ($\alpha$)\tabularnewline
\hline 
1 & HML2 & 0.285 & 0.007 & 2.916\tabularnewline
2 & Mom2{*}Mkt-RF & 0.323 & 0.005 & 2.153\tabularnewline
3 & SMB2{*}Mkt-RF & 0.377 & 0.006 & 2.48\tabularnewline
4 & Mkt-RF2 & 0.406 & 0.005 & 2.405\tabularnewline
5 & HML{*}Mom & 0.429 & 0.004 & 2.337\tabularnewline
6 & Mom & 0.485 & 0.005 & 2.682\tabularnewline
7 & CMA2{*}Mom & 0.499 & 0.006 & 3.087\tabularnewline
8 & SMB2{*}HML & 0.509 & 0.005 & 2.777\tabularnewline
9 & Mom2{*}SMB & 0.537 & 0.005 & 2.976\tabularnewline
10 & Mkt-RF3 & 0.552 & 0.005 & 2.722\tabularnewline
\hline 
\end{tabular}
\end{table}

\end{appendix}



\end{document}